\documentclass[format=acmsmall, review=false]{acmart}

\usepackage{acme20} 

\usepackage[utf8]{inputenc}
\usepackage{amsmath}
\usepackage{array}
\usepackage{graphicx}
\usepackage[english]{babel}
\usepackage[utf8]{inputenc}
\usepackage{algorithm}
\usepackage[noend]{algpseudocode}

\usepackage{enumitem}
\usepackage{wasysym}
\usepackage{framed}
\usepackage{pgfgantt,rotating}
\usepackage{subfloat}
\usepackage{multirow}
\usepackage{blindtext}

\newtheorem{proposition}{Proposition}[section]
\newcommand{\Dk}{\lceil \Delta/\kappa \rceil}
\newcommand{\Dksq}{\lceil \Delta/\kappa^2 \rceil}
\newcommand{\dk}{\lceil \delta/\kappa \rceil}

\title{Permissionless Consensus}

\author{Andrew Lewis-Pye}
\affiliation{%
  \institution{\department{Department of Mathematics} \institution{London School of Economics} \city{London} \country{UK}}
}
\email{a.lewis7@lse.ac.uk}

\author{Tim Roughgarden}
\affiliation{%
  \institution{\department{Data Science Institute}
    \institution{Columbia University \& a16z Crypto} \city{New York} \country{USA}}
}
\email{tim.roughgarden@gmail.com}

\date{Jan 2023}

\begin{abstract} 
  Blockchain protocols typically aspire to run in a {\em
    permissionless} setting, in which nodes are owned and operated by
  a potentially large number of diverse and unknown entities, with
  each node free to use multiple identifiers and to start or stop
  running the protocol at any time.  
This paper 
offers a
  framework for reasoning about the rich design space of blockchain
  protocols and their capabilities and limitations in
  permissionless settings.

We propose a hierarchy of settings with fundamentally
different ``degrees of permissionlessness,'' specified by the amount
of knowledge that a protocol has about the current participants:
\begin{itemize}

\item {\em Fully permissionless setting:} A protocol has no knowledge
  about current participation. Proof-of-work protocols are typically
  interpreted as operating in this setting.

\item {\em Dynamically available setting:} A protocol knows a
  (dynamically evolving) list of identifiers (e.g., public keys associated
  with locked-up stake), and the current participants are a subset of
  these identifiers. Proof-of-stake longest-chain protocols are typically
  interpreted as operating in this setting.

\item {\em Quasi-permissionless setting:} The current participants are
  all of the identifiers in the set above. Proof-of-stake PBFT-style
  protocols are typically interpreted as operating in this setting.

\end{itemize}

This paper also presents a number of results that showcase the
versatility of our analysis framework for reasoning about blockchain
protocols that operate in one these settings.  For example:
\begin{itemize}

\item In the fully permissionless setting, even in the synchronous
  setting and with severe restrictions on the total size of the
  Byzantine players, every deterministic protocol for Byzantine
  agreement has an infinite execution in which honest players never terminate.

\item In the dynamically available and partially synchronous setting,
  no protocol can solve the Byzantine agreement problem with high
  probability, even if there are no Byzantine players at all.

\item In the quasi-permissionless and partially synchronous setting,
  by contrast, assuming a suitable bound on the total size of the
  Byzantine players, there is a deterministic proof-of-stake protocol
  for state machine replication. 
  
\item In the dynamically available,
  authenticated, and synchronous setting, no optimistically responsive
  state machine replication protocol guarantees consistency and
  liveness, even when there are no Byzantine players at all.

\item In the quasi-permissionless and synchronous setting, every
  proof-of-stake state machine replication protocol that uses only
  time-malleable cryptographic primitives is vulnerable to long-range
  attacks.

\end{itemize}
\end{abstract}

\begin{document}

\maketitle

\section{Introduction}

\subsection{Consensus: Permissioned vs.\ Permissionless} \label{fbi} 

The goal of a {\em consensus protocol} is, loosely speaking, to keep a
distributed network of nodes ``in sync,'' even in the presence of an
unpredictable communication network and incorrect behavior by some of
the nodes.  
In a typical traditional application, such as replicating a database
in order to boost its uptime, the nodes would generally be owned and
operated by one or a small number of entities, like a government
agency or a corporation.  Such applications are most naturally modeled
with what is now called the {\em permissioned setting}, meaning that
the set of nodes that will be running the consensus protocol is fixed
and known at the time of protocol deployment.  Many consensus
protocols are, quite reasonably, designed and analyzed with the
permissioned setting in mind.

Over the past 40+ years, the feasibility of permissioned consensus
has been studied intensely for a wide range of problems,
communication network models,
and fault models.
For example, for the Byzantine agreement problem,
classical work has characterized, as a function of the degree of
asynchrony (e.g., synchronous vs.\ partially synchronous vs.\
asynchronous) and the severity of faults (e.g., crash vs.\ omission
vs.\ Byzantine), the maximum-possible fraction of faulty nodes
that can be tolerated by a correct consensus protocol.

Like any mathematical framework worth its salt, the traditional models
and language for discussing consensus problems and protocols are
versatile and have proved relevant for contemporary, unforeseen
applications.  The focus of this work is blockchain protocols.  The
fundamental goal of a typical such protocol is state machine
replication---keeping a network of nodes in sync on the current state
of the blockchain.  For example, the goal of the Bitcoin protocol is
to keep the nodes running it in agreement on the current set of
``unspent transaction outputs (UTXOs).'' The goal of the Ethereum
protocol is to keep the nodes running it in agreement on the current
state of the Ethereum Virtual Machine.

Unlike traditional consensus protocols, blockchain protocols such as
Bitcoin and Ethereum aspire to run in a ``permissionless
  setting,'' which means, roughly, that the protocol has no idea who
  is running it.
The shift from permissioned to permissionless entails three distinct
challenges:

\begin{quote}
\textbf{The unknown players challenge.} The set of distinct
  entities that might run the protocol is unknown at the time of
  protocol deployment and is of unknown size.
\end{quote}

\begin{quote}
\textbf{The player inactivity challenge.}
Such entities can start or stop running the protocol at any time.
\end{quote}

\begin{quote}
\textbf{The sybil challenge.} Each such entity can operate
  the protocol using
  an unknown and unbounded set of identifiers (a.k.a.\ ``sybils'').
\end{quote}
We stress that these three challenges are logically independent, and
every subset of them leads to a distinct and well-defined setting that
has its own set of possibility and impossibility results (see also
Section~\ref{rw}).  Indeed, different impossibility results in this
paper are driven by different subsets of these challenges.

Why impose these three additional challenges to the tried-and-true
permissioned setting?  Because while the traditional motivation for
replication is tolerance to hardware and software faults, the primary
motivation for it in blockchain protocols such as Bitcoin and Ethereum
is ``decentralization,'' meaning that anyone should be able to
contribute meaningfully to running the protocol and that no entity or
small cartel of entities should be in a position to ``turn it off''
or otherwise interfere with its operation.  All known approaches to
designing blockchain protocols that are decentralized in this sense
must grapple with the unknown players, player inactivity, and sybil
challenges.

All else equal, consensus cannot be easier with these additional
challenges than without.
Many consensus protocols designed for the permissioned setting are not
even well defined outside of it.  For example, if a protocol makes
progress by assembling agreeing votes from a supermajority (i.e., more
than two-thirds) of the nodes running it, how could it be run in a
setting in which the number of nodes is unknown, or in a setting in
which one node can masquerade as many?  The best-case scenario would
seem to be that consensus problems that are solvable in the
permissioned setting remain solvable in a permissionless setting,
albeit with additional protocol complexity and/or under compensatory
assumptions; the worst-case scenario would be that such problems
cannot be solved at all in a permissionless setting.  A priori, there is
no reason to expect that any consensus protocol could provide
non-trivial guarantees
when facing the unknown players, player inactivity, and sybil
challenges.

The Bitcoin protocol~\cite{nakamoto2008bitcoin} was the first to
demonstrate that, under appropriate assumptions, state machine
replication is in fact possible in a permissionless setting.  We now
know that the design space of blockchain protocols is surprisingly
rich, and even the protocols that are most widely used in practice
differ substantially across multiple dimensions.  For example, 
different protocols employ different types of
sybil-resistance mechanisms to address the sybil challenge and
ensure that a participant's control over the protocol is independent
of the number of identities under which it participates.  The Bitcoin
protocol takes a ``proof-of-work'' approach to sybil-resistance (with
a participant's control determined by its computational power,
specifically hashrate).  Many more modern blockchain protocols use
alternative approaches including, most commonly, ``proof-of-stake''
(with a participant's power controlled by how much of a protocol's
native cryptocurrency it has locked up in escrow).  Blockchain
protocols can also differ in their approaches to permissionless
consensus along several other dimensions, including the approach to
voting (e.g., the longest-chain rule versus more traditional quorum
certificates) and the assumed cryptographic primitives (e.g., digital
signature schemes and cryptographic hash functions only, or also more
exotic primitives like verifiable delay functions or ephemeral keys).

Summarizing: (i) permissionless consensus is a fundamental component
of modern blockchain protocols and must address the unknown players,
player inactivity, and sybil challenges; (ii) 
these challenges appear to make consensus strictly harder to achieve
than in the traditional permissioned setting; and (iii) any model for
reasoning thoroughly about the design space of blockchain protocols
and their capabilities and limitations must accommodate an unruly
rogues' gallery of examples.  This paper provides such a model,
demonstrates several senses in which permissionless settings are
provably more difficult than the traditional permissioned setting, and
identifies unavoidable protocol ingredients and additional assumptions
that are necessary (and often sufficient) to achieve permissionless
consensus.\footnote{This paper subsumes two earlier versions with
  different titles~\cite{lewis2020resource,lewis2021byzantine} and
  includes a more fleshed-out model (e.g., the explicit hierarchy in
  Section~\ref{ss:hierarchy}) and a number of new results
  (Theorems~\ref{ortheorem}, \ref{posPoS}, \ref{lrtheorem},
  \ref{niceneg} and \ref{posPoS2}).}












\subsection{Degrees of Permissionlessness: A Hierarchy of Four Settings}\label{ss:hierarchy}


The traditional foils to a consensus protocol are message delays and
faulty nodes, and in much of the literature these challenges are
parameterized through hierarchies of successively weaker
assumptions. For example, the synchronous and asynchronous models can
be viewed as opposite ends of a ``degree of asynchrony''
hierarchy, while crash and Byzantine faults bookend a ``fault
severity'' hierarchy.

The first contribution of this paper is to make explicit a ``degree of
permissionlessness'' hierarchy that is defined by four settings
(Table~\ref{table:hierarchy}).%
\footnote{Sections~\ref{ss:rw_unknown}--\ref{ss:rw_sybils} discuss how these settings compare to
those studied or implicit in previous works.}
As the degree of permissionlessness of the setting
increases---roughly, as the knowledge a protocol has about current
participation decreases---positive results become more impressive and
harder to come by (and impossibility results are increasingly to be
expected).  Informally, the four settings are:
\begin{enumerate}

\item {\em Fully permissionless setting.} At each moment in time, the
  protocol has no knowledge about which players are currently running
  it.  Proof-of-work (PoW) protocols such as Bitcoin are typically
  interpreted as operating in this setting.

\item {\em Dynamically available setting.} At each moment in time, 
the protocol is aware of a dynamically evolving list of identifiers (e.g., public keys
that currently have stake committed in a designated staking contract).
Importantly, the evolution of this list can be a function of the
protocol's execution, and different players may have different views of
this list (depending on the messages they have received thus far).
Membership in this list is a necessary condition for participating in
the protocol at that time.
Proof-of-stake (PoS) longest-chain protocols such as Ouroboros \cite{kiayias2017ouroboros} 
and Snow White \cite{daian2019snow} are typically interpreted as operating in this setting.

\item {\em Quasi-permissionless setting.} Membership in the dynamically evolving list of
  identifiers is also a sufficient condition for current participation in the
  protocol. That is, every identifier in the list controlled by an honest
  player is assumed to be running the protocol (correctly) at that
  time.  Proof-of-stake PBFT-style protocols such as Algorand are
  typically interpreted as operating in this setting.\footnote{We use
    the term ``PBFT-style protocol'' informally to mean a consensus
    protocol that, like PBFT~\cite{castro1999practical} and many other protocols,
    proceeds via a series of voting stages and guarantees
    safety via some form of quorum certificates.}
Certain proof-of-work protocols such as Byzcoin
\cite{kokoris2016enhancing}, Hybrid \cite{pass2016hybrid} and
Solida~\cite{abraham2016solida} also operate in this setting (and,
unlike Bitcoin, do not function as intended in the fully
permissionless setting).

\item {\em Permissioned setting.} The list of identifiers is fixed at the
  time of the protocol's deployment, with one identifier per participant and
  with no dependence on the
  protocol's execution.  At each moment in time, membership in this
  list is necessary and sufficient for current participation in the
  protocol.  Tendermint is a canonical example of a blockchain
  protocol that is designed for the permissioned
  setting.\footnote{Tendermint Core is a permissioned protocol (see \url{https://docs.tendermint.com/v0.34/tendermint-core/validators.html}):
    ``Validators are expected to be online, and the set of validators
    is permissioned/curated by some external process. Proof-of-stake
    is not required, but can be implemented on top of Tendermint
    consensus.''
Permissionless
    blockchain protocols based on the Tendermint protocol (such as
    Cosmos~\cite{kwon2019cosmos}) must also define a sybil-resistant
    way of periodically choosing validators (e.g., via a proof-of-stake
    approach along the lines of the protocol described in Section~\ref{thirteen}).}

\end{enumerate}

\begin{table}[h]
\begin{center}
\begin{tabular}{|c|c|c|}\hline
Setting & Protocol's Knowledge of Current Participants & Canonical Protocol\\ \hline
Fully permissionless & None & Bitcoin\\
Dynamically available & Subset of the currently known ID list & Ouroboros\\
Quasi-permissionless & All of the currently known ID list & Algorand\\
Permissioned & All of the a priori fixed ID list & Tendermint\\ \hline
\end{tabular}
\caption{A hierarchy of four settings, ordered according to ``degree
  of permissionlessness.'' The results in this paper demonstrate that the
  four settings are mathematically distinct.}\label{table:hierarchy}
\end{center}
\end{table}
\noindent
The hierarchy is formally defined in Sections
\ref{pswrr}--\ref{defpermyay}. Our results demonstrate that the four
levels of this hierarchy are mathematically distinct (see
Section~\ref{ss:results} and
Tables~\ref{table:sync}--\ref{table:psync}).  The first three levels
can be viewed as explicitly parameterizing the player inactivity
challenge of permissionless consensus.  The sybil challenge is present
in all three of these settings, and is addressed by a protocol through
the use of ``resources'' such as hashrate or stake.  The unknown
players challenge is also fully present in each, although stronger
assumptions on player activity effectively provide more knowledge to a
protocol about the set of (relevant) players; for example, in the
quasi-permissionless setting, a protocol can regard the ``current
player set'' as the players that control the identifiers possessing a
non-zero amount of an on-chain resource such as stake.


With the definitions of these settings in place, we can then
colloquially speak about the ``permissionlessness of a blockchain
protocol,'' meaning the most permissionless setting in which the
protocol satisfies desired properties such as safety and liveness
(under appropriate assumptions, for example on the combined power of
Byzantine players and on message delays). For example, the same way 
the term ``asynchronous protocol'' might refer to a protocol that
functions as intended even in the asynchronous model, we could describe
Bitcoin as a ``fully permissionless protocol,'' and similarly for
Ouroboros (a ``dynamically available protocol'') and Algorand (a
``quasi-permissionless protocol'').\footnote{Appendix~\ref{apC}
  classifies several more blockchain protocols
  according to this hierarchy. In particular, the Ethereum protocol
  proves the point that a fixed blockchain protocol can be fruitfully
  analyzed (and possess different sets of non-trivial guarantees) in
  different levels of our hierarchy.}  In this way, our framework makes
precise the common intuition that proof-of-stake protocols are
typically ``more permissioned'' than proof-of-work protocols
(dynamically available or quasi-permissionless, but not fully
permissionless) but also ``more permissionless'' than traditional
(permissioned) consensus protocols.\footnote{One commonly expressed
  viewpoint is that proof-of-stake protocols are ``automatically
  permissioned,'' on the grounds that you cannot participate unless
  someone agrees to give or sell you some of the protocol's native
  currency.  We disagree with this viewpoint, on two grounds.  First,
  from a pragmatic perspective, the prospect of being unable to
  acquire enough ETH (say) from an exchange like Coinbase or Uniswap
  to become a (proof-of-stake) Ethereum validator is arguably as
  far-fetched as being unable to acquire enough ASICs to participate
  meaningfully in the (proof-of-work) Bitcoin protocol.  Second, from
  a mathematical perspective, the current participants of a
  proof-of-stake protocol generally depend on the protocol's
  execution---a surprisingly subtle issue, as we'll explore in
  Section~\ref{addedsec}---while the participants of a permissioned
  protocol are fixed once and for all.
Theorem~\ref{niceneg} shows that this distinction
  materially changes the conditions under which consensus is achievable.}

\subsection{Five Representative Results}\label{ss:results}

We next highlight informal statements of five of our results that
illustrate the flexibility of our analysis framework for protocols
that operate in one of the settings described in
Section~\ref{ss:hierarchy}; see also Tables~\ref{table:sync}
and~\ref{table:psync}.  The undefined terms in the statements below
are explained in Section~\ref{ss:model} or in
Sections~\ref{pswrr}--\ref{longr}.
%
%
See
Section~\ref{ss:rw_results} for a comparison of our results to those
in previous works.


\vspace{0.2cm} \noindent \textbf{Theorem \ref{fmt}}. 
In the fully permissionless, authenticated,
  and synchronous setting, every deterministic protocol for solving Byzantine
  agreement has an infinite execution, even when Byzantine players
  control only an arbitrarily small (but non-zero) fraction of the
active players' resources at each timeslot and deviate from honest behavior only by delaying message dissemination (or crashing).

\vspace{0.1cm} \noindent 
Theorem~\ref{fmt} can be viewed as an analog of the FLP impossibility
theorem~\cite{fischer1985impossibility}, which states a similar
conclusion for permissioned protocols in the asynchronous
setting. In this sense, the fully permissionless (and synchronous)
setting poses as great a challenge as the asynchronous (and
permissioned) setting.  
Theorem~\ref{fmt}, coupled with a recent result of Losa and
Gafni~\cite{losa2023consensus} (see Theorem~\ref{sep} and
Appendix~\ref{app:lg}), formally separates the fully permissionless
and dynamically available settings.



\vspace{0.2cm} \noindent \textbf{Theorem \ref{psm}}. 
In the dynamically available, authenticated,
  and partially synchronous setting, no protocol can solve the
  Byzantine agreement problem with high probability, even if there are
  no Byzantine players at all.

\vspace{0.1cm} \noindent 
Theorem~\ref{psm} can be viewed as a CAP-type result, stemming from the
ambiguity between waning participation (a challenge of the dynamically
available setting) and unexpectedly long message delays (a challenge
of the partially synchronous setting).
  



\vspace{0.2cm} \noindent \textbf{Theorem \ref{ortheorem}}. 
In the dynamically available,
  authenticated, and synchronous setting, no optimistically responsive
  state machine replication protocol guarantees consistency and
  liveness, even when there are no Byzantine players at all.

\vspace{0.1cm} \noindent 
Informally, ``optimistic responsiveness'' strengthens liveness by
requiring that the speed of transaction confirmation depend only on
the realized message delays (as opposed to on an assumed worst-case
bound on message delays). Theorem~\ref{ortheorem} shows that this
performance goal is fundamentally at odds with the unpredictable
inactivity of players in the dynamically available setting.

\vspace{0.2cm} \noindent \textbf{Theorem \ref{posPoS}}. 
In the quasi-permissionless,
  authenticated, and partially synchronous setting, there exists a
  deterministic and optimistically responsive proof-of-stake protocol
  that solves the state machine
  replication problem (and hence also the Byzantine agreement
  problem), provided Byzantine players always control less than
  one-third of the total stake.  

\vspace{0.1cm} \noindent 
Theorem~\ref{posPoS}, together with either Theorem~\ref{psm}
or~\ref{ortheorem}, proves that
the dynamically available setting is strictly more challenging than
the quasi-permissionless setting.

As an aside, Theorem~\ref{posPoS} requires the (standard) assumption
that the confirmed transactions always allocate Byzantine players less
than a bounded fraction of the stake.  Theorems~\ref{niceneg}
and~\ref{posPoS2} 
provide possibility and
impossibility results under relaxed assumptions that concern only the
transactions that have been sent (as opposed to the transactions that
have been confirmed).
In
particular, Theorem~\ref{niceneg} provides a formal separation between
the quasi-permissionless and permissioned settings.

  


{\small
\begin{table}[h]
\begin{center}
\begin{tabular}{|c|c|c|c|c|c|}\hline
Setting & prob.\ BA & det.\ BA (delays) & det.\ BA & acct.\ SMR
  & opt.\ resp.\ SMR\\ \hline
FP (pos.) & $< \tfrac{1}{2}$ \cite{garay2018bitcoin} & 0 (trivial) & 0 (trivial) & N/A
  & N/A \\
FP (imp.) & $\ge \tfrac{1}{2}$ (folklore) & $> 0$ (Thm.~\ref{fmt}) &
                                                         $> 0$
(Thm.~\ref{fmt}) & no \cite{neu2022availability} & no (Thm.~\ref{ortheorem})\\ \hline
DA (pos.) & $< \tfrac{1}{2}$ \cite{garay2018bitcoin} & $< \tfrac{1}{2}$ \cite{losa2023consensus} & $< \tfrac{1}{2}$? (conjecture) & N/A & N/A\\
DA (imp.) & $\ge \tfrac{1}{2}$ (folklore) & $\ge \tfrac{1}{2}$ (folklore) &  $\ge \tfrac{1}{2}$ (folklore) & no \cite{neu2022availability} & no (Thm.~\ref{ortheorem})\\ \hline
QP (pos.) & $< \tfrac{1}{2}$ \cite{dolev1983authenticated} & $< \tfrac{1}{2}$
\cite{dolev1983authenticated} & $<\tfrac{1}{2}$ \cite{dolev1983authenticated} & yes (Thm.~\ref{posPoS}) & yes (Thm.~\ref{posPoS})\\ 
QP (imp.) & $\ge \tfrac{1}{2}$ (folklore) & $\ge \tfrac{1}{2}$ (folklore) &
$\ge \tfrac{1}{2}$ (folklore)& N/A & N/A \\ \hline
\end{tabular}
\caption{Possibility (pos.) and impossibility (imp.) results in the
  synchronous setting.  ``FP'' refers to the fully permissionless
  setting (with resource restrictions) defined in
  Sections~\ref{pswrr}--\ref{ps2}; ``DA'' refers to the dynamically
  available setting defined in Section~\ref{dasection}; and ``QP''
  refers to the quasi-permissionless setting (with reactive resources)
  defined in Sections~\ref{pdr} and~\ref{defpermyay}.  ``Det.\ BA''
  refers to the Byzantine agreement problem defined in
  Section~\ref{2.5}, and ``prob.\ BA'' to its probabilistic version
  defined in Section~\ref{instance}. ``Det.\ BA (delays)'' refers to a
  version of the Byzantine agreement problem in which faulty players
  can deviate from honest behavior only by delaying the dissemination
  of messages (or crashing).  ``Acct.\ SMR'' and ``opt.\ resp.\ SMR''
  concern the existence of state machine replication protocols that
  satisfy, respectively, accountability (in the sense defined in
  Section~\ref{ss:da_imp2}) and optimistic responsiveness (in the
  sense defined in Section~\ref{ss:da_imp3}). The numbers in the
  second, third, and fourth columns refer to the fraction of resources
  controlled by Byzantine players.}\label{table:sync}
\end{center}
\end{table}

\begin{table}[h]
\begin{center}
\begin{tabular}{|c|c|c|c|c|c|}\hline
Setting & prob.\ BA & det.\ BA (delays) & det.\ BA & acct.\ SMR
  & opt.\ resp.\ SMR\\ \hline
FP (pos.) & N/A & N/A & N/A & N/A
  & N/A \\
FP (imp.) & $\ge 0$ (Thm.~\ref{psm}) & $\ge 0$ (Thm.~\ref{psm}) &
$\ge 0$ (Thm.~\ref{psm})                                                                 & no \cite{neu2022availability} & no (Thm.~\ref{ortheorem})\\ \hline
DA (pos.) & N/A & N/A & N/A & N/A & N/A\\
DA (imp.) & $\ge 0$ (Thm.~\ref{psm}) & $\ge 0$ (Thm.~\ref{psm}) &  $\ge 0$ (Thm.~\ref{psm}) & no \cite{neu2022availability} & no (Thm.~\ref{ortheorem})\\ \hline
QP (pos.) & $<\tfrac{1}{3}$ & $<\tfrac{1}{3}$ & $<\tfrac{1}{3}$  & yes  & yes \\ 
& (Thm.~\ref{posPoS} for SMR) &  (Thm.~\ref{posPoS} for SMR) &  (Thm.~\ref{posPoS} for SMR) & (Thm.~\ref{posPoS})&(Thm.~\ref{posPoS}) \\
QP (imp.) & $\ge \tfrac{1}{3}$ \cite{DLS88} & $\ge \tfrac{1}{3}$ \cite{DLS88} &
$\ge \tfrac{1}{3}$ \cite{DLS88} & N/A & N/A \\ \hline
\end{tabular}
\caption{Possibility and impossibility results in the partially
  synchronous setting. Abbreviations are defined in the caption of
  Table~\ref{table:sync}. 
The positive results for the QP setting all
follow from
Theorem~\ref{posPoS}, which provides an SMR
protocol that uses only reactive resources and achieves the stated
properties (and which can be used to
solve BA).}\label{table:psync}
\end{center}
\end{table}

}


\vspace{0.2cm} \noindent \textbf{Theorem \ref{lrtheorem}}. 
In the quasi-permissionless, authenticated,
  and synchronous setting, every proof-of-stake protocol that 
  uses only time-malleable cryptographic primitives is vulnerable to
  long-range attacks.

\vspace{0.1cm} \noindent 
Informally, a ``long-range attack'' is one in which a Byzantine player
violates the security of a blockchain protocol by acquiring old
private keys (e.g., from players for whom they no longer have any
value) and using them to create an alternative blockchain history.
Informally, a time malleable cryptographic primitive is one that can
be evaluated quickly (unlike, say, a verifiable delay function) and
whose output does not depend on the time of evaluation (unlike, say,
ephemeral keys).










\subsection{Main Ingredients of the Analysis Framework}\label{ss:model}

Making the informal statements in Section~\ref{ss:results} precise
requires a formal model, both of the design space of possible
protocols and of the assumptions imposed on the setting.  This section
is meant to give the reader a quick sense of our model and how it can
be used to formalize the specific results highlighted in
Section~\ref{ss:results}.
\begin{itemize}

\item To model the unknown players challenge,
we allow an unknown-to-the-protocol player set of unbounded size.

\item To model the sybil challenge,
we allow each player to control an
  unbounded number of unknown identifiers (e.g., public keys). 

\item To model the player inactivity challenge, 
we allow each player to be ``active''
  (participating in the protocol) or
  ``inactive'' (not participating) at any given timeslot.

\item At each timeslot, an (active) player can disseminate an
  arbitrary finite set of messages.\footnote{We could alternatively
    take an approach such as that in Khanchandani and
    Wattenhoffer~\cite{khanchandani2021byzantine}, which allows for
    point-to-point communication between players once their addresses
    are known; all of this paper's results would remain the same, with
    only minor changes to the proofs. The communication model
    considered here is chosen to reflect a permissionless setting, in
    which players cannot be expected to know how to reach specific
    other players and instead rely on a gossip protocol for message
    dissemination.}  All other active players eventually hear all such
  messages, possibly after a delay.  The assumptions of a setting
  include restrictions on the possible message delays. In this paper,
  we consider only two standard restrictions on message delays, the
  synchronous and partially synchronous settings.

\item As is common in distributed computing, we model a player and its
  behavior via a state diagram.  A protocol is then defined by the
  state diagram of an honest player, and the behavior of a Byzantine
  player is defined by an arbitrary state diagram.  Each timeslot,
  each player makes a state transition, which dictates the new state
  and the messages disseminated by the player at that timeslot.  The
  state transition is a function of the player's current state and the
  new information it learns in that timeslot (messages from other
  players and, as we'll see below, responses from ``oracles'').  This
  model is purely information-theoretic, with no computational
  constraints imposed on honest or Byzantine players
 (other than the inability to simulate oracles, as described below).

\item One consequence of allowing computationally unbounded players is
  that any cryptographic primitives that might be used by a protocol
  must be incorporated outside of the state diagram in an idealized
  form; we use a notion of ``oracles'' for this purpose. By
  definition, no player (honest or Byzantine) can compute an oracle
  output without directly querying the oracle (or being told the
  answer by some other player that has already queried it).

  For example, the ``time malleable cryptographic primitives''
  mentioned in the final result highlighted in
  Section~\ref{ss:results} translate formally in our model to oracles
  that respond to queries immediately (in the same timeslot as the
  query) with an
  answer that is independent of the current timeslot.


\item To operate in a permissionless setting in a sybil-resistant
  way, blockchain protocols must restrict, implicitly or
  explicitly, the allowable actions of a player as a function of how
  much of some scarce ``resource'' the player owns.  
(Proposition~\ref{prop:noresources} in Section~\ref{2.5} makes this
assertion precise.)
The two
  most canonical examples of such resources are hashrate (for
  proof-of-work protocols) and in-escrow cryptocurrency (for
  proof-of-stake protocols).  These two examples are fundamentally
  different: the latter is ``on-chain'' and publicly known (to the
  protocol and all the participating players), while the former is
  ``external,'' meaning private and not directly observable (by the
  protocol or by other players).  In our model, the description of a
  protocol includes the list of resources that its behavior depends
  upon, and a classification of each resource as 
  ``on-chain'' or ``external.'' 

\item We use a specific type of oracle (a ``permitter'') to model the
  restrictions imposed on behavior by an external resource.\footnote{Restrictions imposed by on-chain resources like stake can be
 enforced by honest players without aid of an oracle. For example, if
  the current (publicly observable) stake distribution determines a public key responsible for
  proposing a block in
  the current round, honest players can simply ignore any proposals
  for the round that do not include an appropriate signature.}
A  permitter is an oracle whose response can depend on the amount of
  the external resource that is controlled by the querying player.
For example, one way to model
  proof-of-work in a protocol like Bitcoin is with a permitter that,
  when queried by a player with current resource level~$b$,
  returns the smallest of~$b$ 256-bit strings drawn independently
  and uniformly at random.

\item The current distribution of an on-chain resource like stake
  depends on the execution-thus-far of the blockchain protocol. (Whereas the
  evolution of the distribution of an external resource like hashrate
  can be reasonably modeled as a process independent of the protocol's
  operation.) Thus, discussing on-chain resources requires
  a formal definition of an ``environment'' that issues
  ``transactions'' to players (which may, for example, update the
  resource distribution). Further, the description of a protocol that
  depends on on-chain resources must include a ``confirmation rule''
  (which specifies which transactions an honest player should regard
  as executed, given the messages and oracle responses received by the
  player thus far) and a ``resource update rule'' (which specifies how
  the execution of a transaction affects the distribution of the
  relevant resources).

\item This completes the high-level description of our model for the
  fully permissionless setting. The first two results in
  Section~\ref{ss:results} demonstrate the difficulty of operating in
  this setting: deterministic consensus is impossible even in the
  synchronous setting, and even probabilistic consensus is impossible
  in the partially synchronous setting. (A protocol like Bitcoin
  operates in the fully permissionless setting, but it enjoys consistency
  and liveness guarantees only in the synchronous setting, and only
  with high probability.)  


\item 
The dynamically available and quasi-permissionless settings
%
impose additional assumptions on the correlation between honest players' periods of
activity and their resource balances.  (Recall that each
player can be active or inactive at each timeslot.)  For example,
consider a proof-of-stake protocol that uses a single on-chain
resource that tracks how much cryptocurrency each player has locked up
in escrow in a designated staking contract.  The dynamically available
setting makes the minimal assumption that, at each timeslot, among the
honest players that control a non-zero amount of the on-chain
resource, at least one is active. 

\item A typical PBFT-style protocol based on weighted super-majority voting could stall in the dynamically available setting,
even if on-chain resources are controlled entirely by honest players,
by failing to assemble any quorum certificates due to unexpectedly low
participation.
The quasi-permissionless setting imposes the stronger assumption that,
at every timeslot, every honest player with a non-zero resource
balance---e.g., every honest player with stake locked up at that
time---is active. (The definition of the quasi-permissionless
setting can also be extended to on-chain resources beyond
stake; see Section~\ref{defpermyay} for details.)

\end{itemize}
























\subsection{Possibility vs.\ Impossibility Results}\label{ss:generous}

The primary focus of this paper is impossibility results
(Proposition~\ref{prop:noresources} and Theorems~\ref{fmt}, 
\ref{psm}, \ref{NTT}, \ref{ortheorem}, \ref{lrtheorem}, and
\ref{niceneg}), and our model is, accordingly, optimized to make such
results as strong as possible.  For example:
\begin{itemize}

\item We allow protocols to specify arbitrary oracles, without regard
  to their realizability.

\item We impose no restrictions on the computational or communication
complexity of a protocol.

\item We assume that newly active honest players learn about all
  previously disseminated messages (perhaps after a delay).

\end{itemize}
All of our impossibility results hold {\em even after} granting these
powers to protocols and the honest players that run them.

By the same token, any possibility result proved in our model should
be accompanied by a discussion of the extent to which its generosity
has been abused.  The two main possibility results in this paper
(Theorems~\ref{posPoS} and~\ref{posPoS2}) rely only on relatively
standard assumptions: that secure digital signature schemes exist,
that players can carry out a polynomial amount of computation and
communication, and that newly active honest players can somehow learn
about the blocks that have already been finalized.







\subsection{Organization of this Paper}

The following section summaries are meant to help the reader navigate the
rest of the paper.
\begin{itemize}

\item Section~\ref{rw} provides an overview of related work. 
  Sections~\ref{ss:rw_unknown}--\ref{ss:rw_sybils} and
  Section~\ref{ss:rw_results} compare our model and results,
  respectively, to those in the existing literature.

\item 
Section~\ref{pswrr} formally defines many of the basic
ingredients of the fully permissionless setting listed in
Section~\ref{ss:model}, up to and including oracles.
Section~\ref{2.5} defines the Byzantine agreement problem and
(in Proposition~\ref{prop:noresources})
formally shows that it is unsolvable in the fully permissionless setting
without some notion of ``resource restrictions'' to diffuse the sybil
challenge.

\item 
Section~\ref{fp} formally defines external resources and
permitter oracles (in Section~\ref{ss:external}), and illustrates
how they can model the Bitcoin protocol (in
Section~\ref{ss:bitcoin}).

\item
Section~\ref{belch}
formally defines blockchain
protocols, transactions, environments, stake, and confirmation rules
(in Section~\ref{bps}), 
as well as consistency and
liveness for blockchain protocols (in Section~\ref{ss:liveness}).

\item Section~\ref{ps2} considers the fully permissionless setting
  with external resource restrictions. Section~\ref{ss:fp} gives
  the formal statement of the first result
described in Section~\ref{ss:results} (Theorem~\ref{fmt}), which rules out deterministic
protocols for the Byzantine agreement problem, even in the synchronous
setting.  (This statement is proved in Section~\ref{proofoffmt}.)
Section~\ref{instance} introduces terminology for reasoning
about probabilistic protocols.

\item Section~\ref{dasection} concerns the dynamically
  available setting. Section~\ref{ss:da} formally defines the setting;
  Section~\ref{ss:sep} uses a result of Losa and
  Gafni~\cite{losa2023consensus} to separate the fully permissionless
  and dynamically available settings (Theorem~\ref{sep});
  Section~\ref{ss:da_imp1} formally states the second result highlighted in
  Section~\ref{ss:results} (Theorem~\ref{psm}, which is proved in
  Section~\ref{proofofpsm}); Section~\ref{ss:da_imp2} notes that a
  result of Neu, Tas, and Tse~\cite{neu2022availability} rules out
  accountability in the dynamically available setting
  (Theorem~\ref{NTT}); and Section~\ref{ss:da_imp3} formally states
  the third result highlighted in
  Section~\ref{ss:results} (Theorem~\ref{ortheorem}, which is proved in
  Section~\ref{proofofortheorem}).

\item Section~\ref{pdr} defines a version of the quasi-permissionless
  setting specifically for proof-of-stake protocols
  (Section~\ref{lfqp}) and, as a generalization of stake, the notion
  of reactive protocol-defined resources (Sections~\ref{blockrewards}--\ref{react}).

\item Section~\ref{defpermyay} formally defines the
  quasi-permissionless setting with arbitrary on-chain resources
  (Section~\ref{ss:qp}) and, in Section~\ref{posres}, gives the formal
  statement of the fourth result described in Section~\ref{ss:results}
  (Theorem~\ref{posPoS}, proved in Section~\ref{thirteen}).

\item
Section~\ref{longr} formally
defines long-range attacks and proves the last result highlighted in
Section~\ref{ss:results}.  

\item Section~\ref{addedsec} investigates relaxing restrictions on the
  (execution-dependent) transactions confirmed by a proof-of-stake
  protocol to restrictions only on the transactions issued by the
  environment.  Sections~\ref{ss:circle} and~\ref{ss:circle2} focus on
  impossibility results, culminating in Theorem~\ref{niceneg} (which
  is proved in Section~\ref{twelve}). Section~\ref{envrho} shows (via
  Theorem~\ref{posPoS2}, proved in Section~\ref{thirteen}) that the
  guarantees of Theorem~\ref{posPoS} can be recovered under additional
  natural assumptions on the environment.

\item As noted above, Sections~\ref{proofoffmt}--\ref{twelve} provide
  proofs and additional discussion for results in Sections~\ref{ps2},
  \ref{dasection}, \ref{defpermyay}, and~\ref{addedsec}.

\item Appendix~\ref{app:chia} supplements the proof-of-work discussion
  in Section~\ref{ss:bitcoin} with a second application of the
  permitter formalism in Section~\ref{ss:external}, to proof-of-space
  protocols.

\item Appendix~\ref{apC} supplements Table~\ref{table:hierarchy}
by describing how a number of other blockchain protocols fit into the
hierarchy of settings described
in Section~\ref{ss:hierarchy}.

\item Appendix~\ref{app:lg} presents, for completeness, a full proof
  of Theorem~\ref{sep} (adapted from Losa and
  Gafni~\cite{losa2023consensus}).

\end{itemize}

\section{Related work} \label{rw}

\subsection{Previous work on consensus without an a priori fixed and
  known player set}\label{ss:rw_unknown}

The Byzantine agreement problem was introduced by Lamport, Pease, and
Shostak~\cite{pease1980reaching,lamport1982byzantine} and has long
been a central topic in distributed computing.  Even before the advent
of Bitcoin and other blockchain protocols, a number of papers have
studied the problem in models that do not assume a fixed and known set
of players (but without considering the player inactivity or sybil
challenges).


\vspace{0.2cm}
\noindent \textbf{The CUP framework}.  Several papers, including
Alchieri et al.~\cite{alchieri2008byzantine} and Cavin, Sasson, and
Schiper~\cite{cavin2004consensus}, study the CUP framework (for
``Consensus amongst Unknown Participants'').  In this
model, the player set is fixed and finite but unknown, with
each (always-active and sybil-free) player aware of, and with
authenticated channels to, a subset of the other players that may
increase over time.  These works assume a known bound~$f$ on the
number (rather than the fraction) of faulty players.


Khanchandani and Wattenhoffer~\cite{khanchandani2021byzantine},
motivated in part by blockchain protocols, consider a variant of the
CUP framework that takes a different approach to ``player discovery.''
In their model, players can broadcast messages to all other players
(whether known or not) or send individual messages to players from
whom they have previously received messages.  (A Byzantine player is
allowed to immediately send individual messages to any player of its
choice). Each (always-active) player has a unique identifier and can
sign messages under that identifier only. The assumption is now that
less than one third of the players are Byzantine. Khanchandani and
Wattenhoffer~\cite{khanchandani2021byzantine} prove that consensus is
possible in this setting with synchronous communication, 
but impossible in the partially synchronous model.

\vspace{0.2cm} 
\noindent \textbf{Reconfiguration protocols}. 
There is
a substantial line of work on ``player reconfiguration protocols,''
meaning protocols for which the
set of players is allowed to change mid-execution, as dictated by
certain ``player-change'' instructions; see Duan and
Zhang~\cite{duan2022foundations} for a recent formal treatment of this
idea, and earlier papers by Lamport, Malkhi, and
Zhou~\cite{lamport2009vertical} and Ongaro and
Ousterhout~\cite{ongaro2014search} (for crash faults) and by Abraham
and Malkhi \cite{abraham2016bvp} and Bessani, Sousa, and Alchieri
\cite{bessani2014state} (for Byzantine faults).
The current players are assumed active and can participate under only
one identifier.
The specific case of proof-of-stake PBFT-style protocols in the
quasi-permissionless setting (Section~\ref{lfqp}) can be regarded as
player reconfiguration protocols in which the ``player-change''
instructions have a specific form (e.g., stake-transferring
transactions) and are issued by an adversarial environment (subject to
restrictions such as those proposed in Section~\ref{addedsec}).

\vspace{0.2cm} 
\noindent \textbf{Protocols tolerating unknown participants and crash
  faults}.  A number of papers (see~\cite{aguilera2004pleasant} for
a brief summary) consider distributed algorithms designed for a system
in which the number of participants may be unbounded and unknown, but
which evade the sybil challenge by restricting attention to crash
(rather than Byzantine) faults. For example, Gafni and Koutsoupias
\cite{gafni1998uniform} consider \emph{uniform protocols} that operate
in a setting in which each execution has a finite but unknown set of
participants. In a further generalization, Merritt and Taubenfeld
\cite{merritt2013computing} and Aguilera \cite{aguilera2004pleasant}
consider a setting in which infinitely many participants may join over
time during a single execution (and in which the set of participants is
again unknown).

\vspace{0.2cm} 
\noindent \textbf{Further related work}. 
Other works consider settings that differ both from the classical
permissioned setting and from the settings considered in this paper
(but again, without considering the player inactivity or sybil
challenges).  For example, Okun \cite{okun2005distributed} considers
a setting in which a fixed number of players 
communicate by private channels, 
but might be ``port unaware,'' meaning that they
are unable to determine from which private channel a message has
arrived. Okun~\cite{okun2005distributed} shows that deterministic
consensus is not possible in the absence of a signature scheme 
when players are port unaware, even in the synchronous setting.
Bocherding~\cite{borcherding1996levels} considers a setting in which
a fixed set of $n$ players communicate by private channels (without
the ``port unaware'' condition of Okun~\cite{okun2005distributed}),
and in which a signature scheme is available. Here, however, each
player is unaware of the others' public keys at the start of the
protocol's execution. In this setting,
Bocherding~\cite{borcherding1996levels} shows that Byzantine
agreement is impossible when at least one-third of the players may
be Byzantine (even in the synchronous model).

 
\subsection{Previous work on consensus with intermittently active players}\label{ss:rw_inactive}

\noindent \textbf{Sleepy consensus}.  The ``sleepy consensus'' model
of Pass and Shi \cite{pass2017sleepy} focuses squarely on the player
inactivity challenge (but without addressing the unknown players or
sybil challenges). In this model, there is a fixed and known finite
set of players (with no sybils). At each timeslot, however, an unknown
subset of these players may be inactive. Pass and
Shi~\cite{pass2017sleepy} design a Nakamoto-type consensus protocol
that is live and consistent in the synchronous setting provided
that, at each timeslot, a majority of the active players are honest. 
A number of follow-up works
(e.g.~\cite{momose2022constant,malkhi2022byzantine,d2023streamlining,d2023recent})
present further results for the sleepy consensus model and extensions
of it (without incorporating the unknown players or sybil challenges).
Momose and Ren~\cite{momose2022constant}, for example, extend the
classic PBFT-style approach from static quorum size to dynamic quorum
size to give a protocol with constant latency (i.e., with
time-to-confirmation bounded by a constant factor times the worst-case
message delay).



\vspace{0.2cm}
\noindent \textbf{Precursors to the dynamically available setting}.
There are several works that, as a byproduct of analyzing the
mathematical properties of a specific blockchain protocol, implicitly define
special cases of our dynamically available setting.  For example, Neu,
Tas, and Tse~\cite{neu2021ebb} describe a protocol that, in the
terminology of this paper, satisfies consistency and liveness in both
the synchronous and dynamically available setting and in the partially
synchronous and quasi-permissionless setting (under appropriate bounds
on the resources controlled by Byzantine players).
Similarly, the rigorous analyses of the Snow
White~\cite{daian2019snow} and Ouroboros~\cite{kiayias2017ouroboros}
proof-of-stake longest-chain protocols can be interpreted as working
in the special case of the dynamically available setting in which
there are no
external resources, no on-chain resources other than stake, and, in
the terminology of Section~\ref{bps}, a $\tfrac{1}{2}$-bounded
adversary.



\subsection{Previous work on consensus with sybils}\label{ss:rw_sybils}

\noindent \textbf{Frameworks for analyzing proof-of-work (PoW)
  protocols}.  Many of the previous works on permissionless consensus,
such as Garay, Kiayias, and Leonardos~\cite{garay2018bitcoin} and a
number of follow-up papers, focus specifically on the Bitcoin
protocol~\cite{nakamoto2008bitcoin} and its close relatives.  Pass,
Seeman, and shelat~\cite{WHGSW16} and Pass and
Shi~\cite{pass2017rethinking} consider more general models of
proof-of-work protocols. (Indeed, Pass and
Shi~\cite{pass2017rethinking} explicitly discuss analogs of the
unknown players, player inactivity, and sybil challenges studied here,
albeit in a narrower proof-of-work setting.)  This line of work
generally uses the UC framework of Canetti
\cite{canetti2001universally}, as opposed to the state diagram (with
oracles) approach taken here, and models proof-of-work using a
random-oracle 
functionality (which, in the terminology of Section~\ref{ss:external},
plays the role of a single-use permitter).  For example, Pass and
Shi~\cite{pass2017rethinking} describe a model for analyzing
proof-of-work protocols with a fixed set of participants, where the
number of participants controlled by the adversary depends on its
hashing power.  They show that, if a protocol is unsure of the number
of participants (up to a factor of~2) and must provide availability
whenever the number of participants is at least $n$ and at most $2n$,
the protocol cannot guarantee consistency in the event of a network
partition.
More recent papers in this line of work include Garay et
al.~\cite{garay2020resource}, in which a player's ability to broadcast
is limited by access to a restricted resource, and
Terner~\cite{terner2020permissionless}, in which the model is
reminiscent of that in~\cite{garay2018bitcoin} but with all
probabilistic elements of player selection ``black-boxed.''  

\vspace{0.2cm} 
\noindent \textbf{Permitters}.  The idea of black-boxing the process
of participant selection as an oracle, akin to our notion of
permitters (Section~\ref{ss:external}), was also explored by Abraham
and Malkhi~\cite{abraham2017blockchain}.  Azouvi et
al.~\cite{azouvi2022modeling}, motivated in part by a preliminary
version of this paper~\cite{lewis2021byzantine}, identify precisely
how the properties of a permitter (called a ``resource allocator''
in~\cite{azouvi2022modeling}) govern the properties of a longest-chain
blockchain protocol that uses it, and prove possibility and
impossibility results for the construction of permitters with various
properties. The updated notion of a single-use permitter (see
Section~\ref{ss:external}) in this paper, relative
to~\cite{lewis2021byzantine}, is inspired
by~\cite{azouvi2022modeling}.

\vspace{0.2cm} 
\noindent \textbf{Aside: going beyond proof-of-work}.  
There are several reasons why it is important
to study permissionless consensus in a model that, like the one in
this paper, also accommodates protocols that use on-chain resources such as
stake for sybil-resistance.  First, with the obvious exception of
Bitcoin, most of the major blockchain protocols deployed today (and
most notably Ethereum) are proof-of-stake protocols.  Second,
restricting attention to pure proof-of-work protocols forces other
design decisions that rule out still larger swaths of the protocol
design space. For example, there is a formal sense in which
proof-of-work protocols cannot work with PBFT-style quorum
certificates (see Theorem 5.1 of~\cite{lewis2021does}) and therefore
typically default to a longest-chain rule, while proof-of-stake
protocols can (as shown formally by Theorem~\ref{posPoS}, or
informally by a number of major deployed proof-of-stake
protocols). Finally, confining attention to external resources like
hashrate, which evolve independently of a blockchain's execution,
completely evades the key technical challenges of analyzing protocols
with on-chain resources: resource balances are subjective (that is, a
function of the messages received), and inextricably linked to the
blockchain's execution (specifically, the confirmed transactions).
The results in this paper show that widening the scope of analysis from
proof-of-work protocols to general blockchain protocols reveals a
remarkably rich and nuanced picture (e.g., Tables~\ref{table:sync}
and~\ref{table:psync} and Theorem~\ref{lrtheorem}).

\vspace{0.2cm} 
\noindent \textbf{Proof-of-stake longest-chain protocols}.  As
discussed in Section~\ref{ss:rw_inactive}, in formally analyzing the
Snow White and Ouroboros proof-of-stake longest-chain protocols,
respectively, Daian, Pass, and Shi~\cite{daian2019snow} and Kiayias et
al.~\cite{kiayias2017ouroboros} implicitly define special cases of our
general framework that accommodate both stake and stake-changing
transactions. The main positive results
in~\cite{daian2019snow,kiayias2017ouroboros}, which effectively match
the consistency and liveness guarantees of the Bitcoin
protocol~\cite{garay2018bitcoin} with proof-of-stake protocols, can be
viewed as analogs of our positive results (Theorems~\ref{posPoS}
and~\ref{posPoS2}) for the dynamically available (rather than
quasi-permissionless) setting (and with a $\tfrac{1}{2}$-bounded
adversary rather than a~$\tfrac{1}{3}$-bounded one).  Consistent with
Theorems~\ref{psm}--\ref{ortheorem}, the positive results
in~\cite{daian2019snow,kiayias2017ouroboros} hold only in the
synchronous setting, and the protocols defined in those papers satisfy
neither accountability nor optimistic responsiveness.

Brown-Cohen et al.~\cite{brown2019formal} describe a model for the
analysis of proof-of-stake longest-chain protocols, with a novel focus
on the incentives faced by the players running the protocol. They do
not attempt to model protocols outside of this class. (Of the major
deployed proof-of-stake blockchain protocols, only Cardano, based on
Ouroboros~\cite{kiayias2017ouroboros}, is a longest-chain protocol.)

Dembo et al.~\cite{dembo2020everything} describe a model for 
analyzing the consistency and liveness properties of longest-chain
protocols (proof-of-stake in particular, but also proof-of-work and proof-of-space).
Again, the goals of this paper are more specific and fine-grained than
ours here; the main results in~\cite{dembo2020everything} provide a
precise analysis of the security properties of certain canonical
longest-chain protocols in the synchronous setting and with static
resource balances.


\vspace{0.2cm} 
\noindent \textbf{Ripple, Stellar, and asymmetric trust
  assumptions}. Classical consensus protocols typically operate with
\emph{symmetric} trust assumptions, meaning that all honest players
operate under the same assumption on the largest-possible fraction of
faulty players.  However, a number of papers
(e.g. \cite{damgaard2007secure,cachin2019brief,cachin2021asymmetric})
have considered protocols that operate in the permissioned setting,
but which allow for \emph{asymmetric} trust assumptions. For example,
each player could specify which other players they believe to be honest,
and a protocol might then aim to provide liveness and consistency for
players whose assumptions are correct.  To function as intended, such
protocols typically require the trust assumptions of honest players to
be strongly aligned. For example, Cachin and
Tackmann~\cite{cachin2019brief} describe a protocol that achieves a
form of reliable broadcast assuming, among other things, the existence
of a ``guild,'' meaning
a large group of influential and mutually
trustful players that all make correct trust assumptions.
Cachin and Zanolini~\cite{cachin2021asymmetric} extend these methods
to define a protocol that solves the Byzantine agreement problem under
similar conditions.

Developing this idea further, a number of works
\cite{schwartz2014ripple,amores2020security,mazieres2015stellar,garcia2018federated,losa2019stellar,cachin2022quorum},
including the Stellar and Ripple white papers, have explored
asymmetric trust assumptions beyond the permissioned setting. At a
high level, the basic idea is to ``outsource'' the problem of
sybil-resistance to the protocol participants and then
deliver liveness and consistency in the event that participants
(somehow) make good choices as to whom to trust.  Mirroring the
permissioned setting, positive results require strongly aligned trust
assumptions (see e.g.~\cite{amores2020security,cachin2019brief}).
Protocols based on this approach would not qualify as
``permissionless'' in the senses defined in this paper, as in our
model players have no knowledge about each other (beyond the
disseminated messages that have been received) and thus the blockchain
protocol itself must shoulder the burden of sybil-resistance.

\subsection{Comparison with existing results}\label{ss:rw_results}

Some of the results in this paper have precursors in the existing
literature.  We discuss each of our results in turn.

\vspace{0.2cm} 
\noindent \textbf{Proposition~\ref{prop:noresources}}.
Proposition~\ref{prop:noresources} demonstrates that the sybil
challenge cannot be addressed without supplementing the protocol design
space with some notion of ``resources'' to prevent a Byzantine player
from simulating an arbitrarily large number of honest players. This
result is similar in spirit to Theorem~1 in Pass and
Shi~\cite{pass2017rethinking}. The proof given here must accommodate
the full generality of our framework (e.g., arbitrary oracles),
however.  It is also, in our view, substantially simpler than the
proof in \cite{pass2017rethinking}.

\vspace{0.2cm} 
\noindent \textbf{Theorem \ref{fmt}}. Theorem \ref{fmt} shows that in
the fully permissionless, authenticated, and synchronous setting,
every deterministic protocol for solving Byzantine agreement has an
infinite execution in which honest players never terminate, even when
Byzantine players control only an arbitrarily small (but non-zero)
fraction of the active players' resources at each timeslot. The
theorem holds even when faulty players can deviate from honest
behavior only by crashing or delaying message dissemination.  Pu,
Alvisi, and Eyal \cite{pu2022safe}, motivated in part by a preliminary
version of this paper~\cite{lewis2021byzantine}, show the existence of
a protocol that guarantees deterministic agreement and termination
with probability 1 under general omission failures.
The two results do not contradict each other because the protocol of
Pu, Alvisi, and Eyal is probabilistic and only guarantees termination
with probability~1, i.e., there remain infinite executions in which
honest players never terminate.

The proof of Theorem \ref{fmt} given in Section \ref{proofoffmt} is
similar in form to the proof of Moses and Rajsbaum
\cite{moses2002layered} that deterministic consensus is not possible
in the synchronous and permissioned setting in the case of a mobile
adversary (with both proofs being bivalency arguments in the style of
Fischer, Lynch, and
Paterson~\cite{fischer1985impossibility}).\footnote{See Attiya and
  Ellen~\cite[Chapter 7]{attiya2022impossibility} for an introduction to valency arguments.} We note
that impossibility results for the permissioned setting with a mobile
adversary (but without resource restrictions) or the asynchronous
permissioned setting do not generally
carry over to the fully permissionless setting (with resource
restrictions and synchronous communication); see also
footnotes~\ref{foot:bitcoin} and~\ref{foot:flm} in
Section~\ref{instance}. We leave it for future work to determine
whether Theorem~\ref{fmt} can alternatively be established through
some kind of reduction to an already-known impossibility result
in a different model (perhaps by building on the ideas
in~\cite{moses2002layered,santoro2007agreement,gafni2023time}).



\vspace{0.2cm} 
\noindent \textbf{Theorems~\ref{sep}--\ref{ortheorem}}.  The four
results stated in Section~\ref{dasection} map out what is possible and
impossible in the dynamically available setting.  Theorem~\ref{sep} is
a positive result showing that Theorem~\ref{fmt} does not hold in this
setting---thus proving a formal separation between the fully
permissionless and dynamically available settings---and it is
essentially due to Losa and Gafni~\cite{losa2023consensus}.
Theorem~\ref{NTT}, starting that accountability is impossible in the
dynamically available setting, is a result of Neu, Tas, and
Tse~\cite{neu2022availability}.  

Theorem~\ref{psm} is new to this paper; its proof resembles that of
the ``CAP Theorem''~\cite{brewer2000towards,gilbert2002brewer} and
other results in the same lineage.  For example, as mentioned in
Section~\ref{ss:rw_sybils}, Pass and Shi~\cite{pass2017rethinking}
prove an analogous result specifically for proof-of-work protocols
with a fixed set of participants.  Khanchandani and
Wattenhoffer~\cite{khanchandani2021byzantine} prove an analogous
result in their variant of the CUP framework described in
Section~\ref{ss:rw_unknown}.

Theorem~\ref{ortheorem} is, to the best of our knowledge, the first
result that identifies a fundamental conflict between dynamic
availability and optimistic responsiveness. The closest analog is a
result of Pass and Shi~\cite{pass2018thunderella} that concerns their
stronger notion of \emph{responsiveness}, which insists that
confirmation times scale with the realized message delays
even when there are (a limited number of) Byzantine players.
Pass and Shi~\cite{pass2018thunderella} show that, even in a
synchronous, authenticated, and permissioned setting, and even if one
augments that setting with resources, responsive state machine
replication protocols cannot tolerate a $\rho$-bounded adversary for
any $\rho\geq 1/3$.  To maximize the strength of our impossibility
result in Theorem~\ref{ortheorem}, our definition of optimistic
responsiveness in Section~\ref{ss:da_imp3} is specified to be as weak
as possible.  Theorem~\ref{ortheorem} holds even for $0$-bounded
adversaries, and as such cannot extend to the permissioned (or even
quasi-permissionless) setting; the result is fundamentally driven by
the player inactivity challenge present in the dynamically available
setting.


\vspace{0.2cm} 
\noindent \textbf{Theorem~\ref{posPoS}}. Theorems \ref{posPoS} and
\ref{posPoS2} formally establish that proof-of-stake protocols can be
(deterministic and) live and consistent in the partially synchronous
setting.  A result along the lines of Theorem~\ref{posPoS} is strongly
hinted at in several previous works.  For example, the analysis of the
Algorand protocol by Chen and Micali~\cite{chen2016algorand} provides
a rigorous proof of the protocol's liveness and consistency properties
in the synchronous setting, but does not deal formally with partial
synchrony.  The analysis of the Algorand protocol in Chen et
al.~\cite{chen2018algorand}, meanwhile, establishes liveness and
consistency in the partially synchronous setting, but only for a
permissioned version of the protocol.  Chen et
al.~\cite{chen2018algorand} suggest, without formal analysis, that in
a permissionless setting, one can implement a proof-of-stake version
of their permissioned protocol. A recent paper by Benhamouda et
al.~\cite{benhamouda2023analyzing} gives a probabilistic analysis of
consistency and liveness for the presently deployed instantiation of
the Algorand protocol (which differs substantially from the protocol
presented in the older papers cited above), but does not consider
changing players. Similar comments apply to the analysis of the Casper
protocol by Buterin and Griffith~\cite{buterin2017casper}, which
formally treats a permissioned version of the protocol and,
without any further formal analysis,
suggests how it could be extended to a proof-of-stake protocol
suitable for a permissionless setting.

Theorems \ref{posPoS} and \ref{posPoS2} also formally establish the
existence of a proof-of-stake protocol for the authenticated and
partially synchronous setting that is $(1/3,1)$-accountable (as
defined in Section~\ref{ss:da_imp2}) and optimistically responsive (as
defined in Section \ref{ss:da_imp3}). A number of papers
(e.g.,~\cite{sheng2021bft,ranchal2020zlb,civit2021polygraph,buterin2017casper,shamis2022ia})
have previously established accountability for permissioned
protocols. 
Similarly, several variations of optimistic responsiveness have been
achieved by permissioned protocols (e.g.,~\cite{yin2019hotstuff,castro1999practical,shrestha2020optimality}). 


\vspace{0.2cm} 
\noindent \textbf{Theorem~\ref{lrtheorem}}.  In~\cite{daian2019snow},
Daian, Pass, and Shi assert that, in the absence of ``additional trust
assumptions,'' proof-of-stake protocols are vulnerable to long-range
attacks. As shown by the Algorand protocol, however, cryptographic
primitives such as ephemeral keys can be used to prevent long-range
attacks. Theorem~\ref{lrtheorem} formally establishes precisely what
form of cryptographic primitive is required to avoid vulnerability to
long-range attacks.

In more detail, Theorem \ref{lrtheorem} shows that in the
quasi-permissionless, authenticated, and synchronous setting, every
proof-of-stake protocol that uses only time-malleable cryptographic
primitives is vulnerable to long-range attacks (see
Section~\ref{longr} for the definition of time-malleable oracles).
This result holds even if the adversary obtains stake only through
posterior corruption, meaning the acquisition of the private keys of
players who no longer own stake.
Some previous
works~\cite{chen2016algorand,badertscher2018ouroboros} can be
interpreted as exploring, in our terminology, how non-time-malleable
oracles might be used
to prevent long-range attacks. 
In a different but related direction, Tas et al.~\cite{tas2023bitcoin,tas2022babylon}
show that
\emph{slashable-safety} and liveness cannot be achieved for
proof-of-stake protocols (even those with non-time-malleable oracles)
if the adversary can, even without any posterior corruption, acquire a
majority of the stake 
(and without extra trust assumptions such as those considered in
\cite{wsbut,daian2019snow}). 
Tas et al.~\cite{tas2023bitcoin,tas2022babylon} also show how
long-range attacks on proof-of-stake protocols can be mitigated by
using a system of checkpoints recorded in the Bitcoin blockchain (the
security of which is, in effect, an extra trust assumption).

Finally, Azouvi, Danezis, and Nikolaenko~\cite{azouvi2020winkle}
describe an approach to mitigating long-range attacks on
proof-of-stake protocols by having clients sign off on checkpoints,
where a ``client'' refers to an entity that does not participate
directly in a blockchain protocol but may monitor and submit
transactions to the protocol. The motivation for this idea is that the
potentially large set of clients may be
less vulnerable to widespread posterior corruption than the
(presumably smaller) set of validators.

\vspace{0.2cm} 
\noindent \textbf{Theorems~\ref{niceneg} and~\ref{posPoS2}}.  All
existing analyses of the consistency and liveness properties of
proof-of-stake protocols in the partially synchronous setting have
restricted attention to protocol executions in which Byzantine players
always control less than a certain fraction of the total stake.
%
In a proof-of-stake protocol, however, the amount of stake controlled
by Byzantine players------and, thus, whether an execution meets the
hypothesis of such a guarantee---can depend on the realized network
delays, the strategies of the Byzantine players, and even the
decisions of the protocol itself.  Theorems~\ref{niceneg}
and~\ref{posPoS2} are, to our knowledge, the first results in the
literature to address this issue by investigating what can be achieved
in the partially synchronous setting
under assumptions only on the environment (e.g., on the set of
transactions issued at each timeslot) rather than on the specific
execution.

\section{The fully permissionless setting without resource restrictions }  \label{pswrr}

This section describes a basic version of our model, without
any notion of ``resources''   (which will be introduced in
Section~\ref{fp}).  The most important concept in this section is
``oracles,'' which we use to model various cryptographic primitives.
%
%
Some of the impossibility results presented in this paper hold
independent of the assumed cryptographic primitives, while others 
depend on which cryptographic primitives are assumed to exist.  

\subsection{The set of players and the means of communication}\label{ss:inputs}

\vspace{0.2cm} 
\noindent \textbf{The set of players}. We consider a potentially
infinite set of players $\mathcal{P}$. Each player $p\in \mathcal{P}$
is allocated a non-empty and potentially infinite set of
\emph{identifiers} $\mathtt{id}(p)$.  One can think of
$\mathtt{id}(p)$ as an arbitrarily large pre-generated set of public
keys for which~$p$ knows the corresponding private key; a player~$p$
can use its identifiers to create an 
arbitrarily large number of sybils.
Identifier sets are disjoint, meaning
$\mathtt{id}(p)\cap \mathtt{id}(p')=\emptyset$ when $p\neq p'$;
intuitively, no player knows the private keys that are held by other
players.

\vspace{0.2cm} 
\noindent \textbf{Permissionless entry and exit}. Time is divided into
discrete timeslots $t=1,2,\dots, d$, where
$d\in \mathbb{N}_{\geq 1} \cup \{ \infty \}$ is referred to as the
\emph{duration}. Each player may or may not be \emph{active} at each
timeslot.  A \emph{player allocation} is a function specifying
$\mathtt{id}(p)$ for each $p\in \mathcal{P}$ and the timeslots at
which each player is active.  Because protocols generally require
\emph{some} active players to achieve any non-trivial functionality,
unless explicitly stated otherwise we assume that a non-zero but finite number of players is active at
each timeslot $\leq d$.

\vspace{0.2cm} 
\noindent \textbf{Inputs}. Each player is given a finite set of
\emph{inputs}, which capture its knowledge at the beginning of the
execution of a protocol.  If a variable is specified as one of $p$'s
inputs, we refer to it as \emph{determined for} $p$, otherwise it is
\emph{undetermined for} $p$. If a variable is determined/undetermined
for all $p$ then we refer to it as \emph{determined/undetermined}. For
example, to model a permissionless environment with sybils, we
generally assume that the player set and the player allocation
are undetermined.  The duration~$d$, on the other hand, we typically
treat as determined.

\vspace{0.2cm} 
\noindent \textbf{Message sending}.  At each timeslot, each active
player may \emph{disseminate} a finite set of messages (each of finite
length) to the other players, and will receive a (possibly empty)
multiset of messages that have been disseminated by other players at
previous timeslots.  (Inactive players do not disseminate or receive
any messages.)  Each time $p$ disseminates a message at some timeslot,
we wish to stipulate that it is received at most once by each of the
other players.
Formally, for each dissemination of the message $m$,
i.e., for each triple $(p,m,t)$ such that $p$ disseminates $m$ at $t$,
and for each $p'\neq p$, we allow that $p'$ may \emph{receive that
  dissemination} at a single timeslot $t'>t$.  If $p'$ receives $k$
disseminations of~$m$ at $t'$, then $k$ is the multiplicity with which
$m$ appears in the multiset of messages received by~$p'$ at~$t'$.
If $p$ disseminates $m$ at $t$, then it is convenient to suppose that
$p$ regards that dissemination as received (by itself) at timeslot $t+1$ (if active). 

\vspace{0.1cm} We consider two of the most common models of message reliability,
the synchronous and partially synchronous
models. 
These models have to be adapted, however, to deal with the
fact that players may not be active at all timeslots.\footnote{As
  discussed in Section~\ref{ss:generous}, the fact that honest active
  players eventually receive all messages disseminated in the past
  (even after a long period of inactivity) only makes our impossibility
  results stronger. In the protocol used to prove our possibility
  results (Theorems~\ref{posPoS} and~\ref{posPoS2}), newly active
  players need to know only the set of already-finalized blocks.}

\noindent \textbf{Synchronous model.}  There exists some determined
$\Delta \in \mathbb{N}_{>0}$ such that if $p$ disseminates $m$ at $t$,
and if $p'\neq p$ is active at $t'\geq t+\Delta$, then $p'$ receives
that dissemination at a timeslot $\leq t'$.\footnote{We typically
  assume that $\Delta \ge 2$.  (With our dissemination model, if
  $\Delta=1$, active honest players would automatically proceed in lockstep,
  with each receiving the same set of disseminated messages in each
  timeslot.)}

\noindent \textbf{Partially synchronous model.}  There exists some determined $\Delta \in \mathbb{N}_{>0}$, and some undetermined timeslot called GST (with GST$\leq d$), such that if $p$ disseminates $m$ at $t$, and if $p'\neq p$ is active at $t'\geq \text{max} \{ \text{GST},  t \}  +\Delta$, then $p'$ receives that dissemination at a timeslot $\leq t'$.

\vspace{0.2cm} 
\noindent \textbf{Timing rules}.  It is useful to consider the notion
of a \emph{timing rule}, which specifies how long messages take to be
received by each player after they are disseminated.
Formally, a timing rule is
 a partial function $T$ that maps tuples of the form $(p,p',m,t)$
to timeslots. An execution of the protocol follows the timing rule $T$
if  the following holds for all players $p$ and $p'$:  We have that
$p'$ receives $m$ at $t'$  if and only if there exists some $p$ and $t<t'$ such
that $p$ disseminates the message $m$ at  $t$ and
$T(p,p',m,t)\downarrow =t'$.\footnote{Generally, we write
  $x\downarrow$ to indicate that the variable $x$ (e.g., the output of
  a partial function) is defined, and we
  write $x\uparrow$ to indicate that $x$ is undefined.} Of course, we restrict attention to
timing rules that are consistent with the communication model
(synchronous or partially synchronous).

\subsection{Players and oracles}\label{ss:st}

\vspace{0.2cm} 
\noindent \textbf{The oracles}. 
The description of a protocol will specify
 a (possibly empty) set of \emph{oracles}
$\mathcal{O}= \{ O_1,\dots, O_z \}$ that are generally used to
capture, within our information-theoretic model, idealized
cryptographic primitives (e.g., 
VDFs)
and (in Section~\ref{fp}) the use of external resources.\footnote{As
  discussed in Section~\ref{ss:generous}, allowing idealized oracles
  in place of concrete implementations only makes our impossibility
  results stronger. 
Our possibility results (Theorems~\ref{posPoS}
  and~\ref{posPoS2}) rely on only an idealized version of a standard secure
  digital signature scheme (as described in Section~\ref{ss:permitted}).}  At
each timeslot, each active player may send a set of \emph{queries} to
each oracle~$O_i$. Each query is a binary string of finite length. 

If $p$ sends a query $q$ to $O_i$ at timeslot $t$,
it receives a \emph{response} to that query from $O_i$ at a
timeslot $t' \ge t$.  Unlike message delays, which are determined by the
communication network and therefore unpredictable, the time at which a
querying player receives an oracle response is determined by the oracle
specification. (Intuitively, typical oracle computations are carried
out ``locally'' by a player.)
For some oracles, the response may be received instantly (i.e., 
$t'=t$). For others, such as those modeling VDFs, the response may be
received at a timeslot later than $t$.
Oracles are ``non-adaptive,'' in that their response (or
distribution over responses) to a query is independent of any previous
queries asked of them.
    
\vspace{0.2cm} 
\noindent \textbf{Oracles in the deterministic model}. To prove
impossibility results that hold for all possible sets of oracles, we
must be fully precise about the way in which oracles function.
Oracles may be given some inputs from amongst the determined inputs at
the start of the protocol execution (e.g., the duration
$d$). During the protocol execution, they are then sent oracle queries
by the players. 
In the ``deterministic model,'' we assume that oracles
are deterministic. This means that an oracle's inputs at the start of
the protocol execution determine a \emph{response function} $f$
that specifies, for every possible query~$q$ and timeslot~$t$, a
response~$r$ to~$q$ (if sent at~$t$ by some~$p$) and a delivery timeslot $t' \ge
t$ for that response; we write $f(q,t)=(r,t')$. 
We call the response \emph{undelivered} until received by $p$. 

%


\vspace{0.2cm} 
\noindent \textbf{Oracles in the probabilistic model}. In the
``probabilistic model,'' each oracle $O$ specifies, given its inputs,
a distribution over response functions. At the start of the protocol
execution, a response function is sampled from the distribution
specified by $O$, and that function then specifies the responses of
the oracle in that protocol execution.  (Intuitively, the oracle flips
all its coins ``up front.'')

\vspace{0.2cm} 
\noindent \textbf{Modeling players as state diagrams}. Each player is
modeled as a processor specified by a state diagram, for which the
set of states $St$ may be infinite. 
A player's state transition at a given timeslot can depend on the
oracle responses received at that time (in addition to the player's
current state, the messages received from other players, and the
number of the current timeslot).\footnote{Thus, in effect, players
  automatically know the number of the current timeslot (even after
  waking up from a long period of inactivity); this modeling choice
  serves to make our impossibility results as strong as possible.  For
  our positive results in partial synchrony (Theorems~\ref{posPoS}
  and~\ref{posPoS2}), where such knowledge arguably goes against the
  spirit of the setting, we use a protocol for which the behavior of
  honest players depends only on the set of messages received (and not
  on the number of the current timeslot).\label{foot:know_timeslot}}

One subtle point, which can be ignored on a first reading, is that
some oracles may respond ``instantaneously'' (i.e., $f(q,t)=(r,t')$
with $t'=t$) and players may wish to query such oracles repeatedly in
the same timeslot. (By contrast, messages disseminated by a player in
a timeslot~$t$ cannot arrive at their destinations prior to
timeslot~$t+1$.)
%
For example, $p$ may query an oracle that models a signature scheme (perhaps with specific properties, such as are satisfied by ephemeral keys \cite{chen2016algorand}), and then make further queries at the same timeslot depending
on the response. 

For this reason, we use state diagrams that specify a special subset
$St^{\ast}$ of the set $St$ of all states, indicating the states from
which the player will proceed to the next timeslot without further
oracle queries.  At the beginning of the protocol's execution, a
player $p$ begins in an \emph{initial state} in $St^{\ast}$, which is
determined by $\mathtt{id}(p)$ and by the other inputs given to $p$.
At each timeslot $t$ for which $p$ is active, $p$ begins in some state
$s^{\ast}\in St^{\ast}$ and receives a multiset of messages
$M$. Player $p$ then enters a state $s$ (which may or may not belong
to $St^{\ast}$).  In the deterministic model, $s$ is determined by
$s^{\ast}$, $M$, and~$t$. In the probabilistic model, $s$ is a random
variable with distribution determined by $s^{\ast}$, $M$,
and~$t$.
Player
$p$ then repeatedly carries out the following steps, until instructed
to proceed to the next timeslot:

\begin{enumerate} 

\item[(1)] Player $p$ receives  $R$, which is the set of all undelivered oracle responses (to queries sent by $p$) with delivery timeslots $\leq t$. 
\item[(2)] If $s\in St^{\ast}$, then $p$ disseminates a set of messages $M'$ and stops instructions for timeslot $t$.   The set of messages $M'$ is determined by $s$ and $R$. Player $p$ will begin its next active timeslot in state $s$. 
\item[(3)] If $ s\notin  St^{\ast}$, then $p$  sends a set of oracle queries $Q$ and enters a new state $s'$, where $Q$ and $s'$  are determined by $s$ and $R$. Set $s:=s'$ and return to (1).

\end{enumerate}


For certain state diagrams, the instructions above might not
terminate. In this case, $p$ disseminates no messages at timeslot $t$
and is regarded as inactive at all subsequent timeslots. If $p$ is
inactive at timeslot $t$, then it does not disseminate or receive any messages, does not send any
queries at $t$, and does not change state.

\vspace{0.2cm}
\noindent \textbf{Byzantine and honest players}. 
To ensure that our impossibility results are as strong
as possible, we focus primarily on a \emph{static adversary}.
(Section~\ref{longr} considers a weak type of dynamic adversary, as is
necessary to discuss long-range attacks.)
In the static adversary model, each player is either \emph{Byzantine}
or \emph{honest}, and an arbitrary and undetermined subset of the
players may be Byzantine. The difference between Byzantine and honest
players is that honest players must have the state diagram $\Sigma$
specified by the protocol, while Byzantine players may have arbitrary
state diagrams. To model a perfectly coordinated adversary, we also
allow that the instructions carried out by each Byzantine player can
depend on the messages and responses received by other Byzantine
players.  That is, if $p$ is Byzantine, the 
messages $p$ disseminates, the queries $p$ sends, and $p$'s
state transition at a timeslot~$t$ are a function not only of $p$'s
state and the multiset of messages and oracle responses received by
$p$ at $t$, but also of the corresponding values for all the other
Byzantine players.


\vspace{0.2cm} 
\noindent \textbf{Protocols and executions}. A \emph{protocol} is a
pair $(\Sigma,\mathcal{O})$, where $\Sigma$ is the state diagram of
honest players and $\mathcal{O}$ is a finite set of oracles. An
\emph{execution} of the protocol $(\Sigma,\mathcal{O})$ is a
specification of the set of players $\mathcal{P}$, the player
allocation, the state diagram of each player and their inputs, and the
following values for each player $p$ at each timeslot:
\begin{enumerate} 
\item[(i)]  $p$'s state at the beginning of the timeslot.  
\item[(ii)]  The multiset of messages received by $p$. 
\item[(iii)]  The sequence of sets of queries sent by $p$. 
\item[(iv)] The sequence of sets of responses received by $p$. 
\item[(v)]  The messages disseminated by $p$.  
\end{enumerate} 
As described above, messages disseminated by a player at the end of a
timeslot (item~(v)) can depend on the messages and all oracles
responses received during the timeslot (items~(ii) and~(iv)).  Oracle
queries sent during a timeslot can depend on the messages received in
that timeslot, as well as any oracle responses received previously during
the timeslot.  A player's (ultimate) state transition in a timeslot can be
interpreted as occurring simultaneously with item~(v).

\vspace{0.2cm} 
\noindent \textbf{Executions ``consistent with the setting''}.  Part
of the definition of a ``setting'' is the class of executions to be
considered; we say that such executions are ``consistent with the
setting.''  For example, if we are assuming the synchronous
communication model as part of the setting, an execution consistent
with the setting will, in particular, follow a timing rule that is
consistent with the synchronous communication model.

\subsection{Permitted messages}\label{ss:permitted}

To model certain cryptographic primitives, and to model resources later on, it
is convenient to restrict the set of messages and queries that a
player can send at any given timeslot. The approach we'll describe
here may initially seem a little heavy-handed, but in exchange it
gives a general methodology that can be used to model signature
schemes, 
VDFs, proof-of-work (PoW), proof-of-space
(PoSp), and proof-of-stake (PoS).

\vspace{0.2cm} 
\noindent \textbf{Message entries and types}. We assume that all
messages and queries are ordered tuples $(e_1,\dots,e_k)$, where $k$
may vary for different messages and queries, and where each entry
$e_i$ may be of a specific \emph{type}. This does not lose any
generality, as tuples may be of length 1 and
of 
\emph{general} type. The type of each entry in the tuple could be
specified explicitly, but to keep things simple we suppress type
information in our notation.

\vspace{0.2cm} 
\noindent \textbf{Permitted messages/queries}. At a given timeslot,
and for a given player $p$, each entry $e_i$ in a message or query
$(e_1,\dots,e_k)$ may or may not be \emph{permitted} for $p$. The
message/query is permitted for $p$ if and only if all entries are permitted for
$p$. The messages that $p$ disseminates at a given timeslot, and the
set of queries that $p$ sends, are those specified by the state
diagram of $p$ which are also permitted.  (Importantly, Byzantine
players cannot send oracle queries or messages that are not permitted
for them.)

In addition to the \emph{general} entry type (which is
permitted for all $p$), entries may be of type
\emph{signed} (used below to model a signature
scheme) or of a type $O_i$ corresponding to an oracle (intuitively,
representing a signature by the oracle~$O_i$).

\vspace{0.2cm}
\noindent \textbf{A simple model of a signature scheme}. 
To model 
a signature scheme, we suppose that an entry of signed type must be of
the form $(id,m)$, where $id$ is an identifier and $m$ is
arbitrary. In the \emph{authenticated} setting, the entry $e=(id,m)$
of signed type is permitted for $p$ at $t$ if and only if
$id\in \mathtt{id}(p)$ or $p$ has previously received a message which
contained the entry $e$ (in which case~$p$ is free to repeat it to
others).  We say that $e$ is $m$ ``signed by'' $id$.\footnote{Modeled
  this way, an idealized signature scheme requires no oracle,
only restrictions on the messages that players are permitted
to send.  (Intuitively, we grant idealized standard signatures ``native
status'' in our framework.)} In the
unauthenticated setting, entries of signed type are always
permitted.\footnote{The practical deployment of a
  signature scheme requires probabilistic processes for key
  generation. Following a traditional approach in distributed
  computing (see, e.g., Lynch~\cite{lynch1996distributed}), we
  treat a protocol that uses a signature scheme as deterministic if
  its state diagram and oracles are deterministic (in effect,
  ``fast-forwarding'' past any probabilistic processes that might
  have been required for the setup).}

\vspace{0.1cm} In general, the ``signed'' type is used when different
entries are initially permitted for different players.  Meanwhile, if
an entry $e$ is of type $O$, then $e$ is permitted for $p$ if $p$ has
previously received a message or response which contained the entry
$e$ of type $O$. (Intuitively, such an entry can be obtained directly
by querying the oracle~$O$, or can be repeated after hearing it from
others.)

\vspace{0.2cm} 
\noindent \textbf{Modeling a perfectly coordinated adversary}. If $p$
and $p'$ are both Byzantine, then the set of entries permitted for $p$
is modified to include those permitted for $p'$. (E.g., Byzantine
players might pool their private keys.)

\vspace{0,2cm} 
\noindent \textbf{Further terminology}. If a message is a tuple of
length 1, then the \emph{type of the message} is the type of its
single entry.

\subsection{Modeling VDFs and ephemeral keys} \label{modhash}

Many of the impossibility results in this paper
(Proposition~\ref{prop:noresources} and Theorems~\ref{fmt}, \ref{psm},
\ref{NTT}, \ref{ortheorem}, and \ref{niceneg}) apply no matter what oracles a protocol might use
(be they practically realizable or not).  Accordingly, the proofs of
these results refer to a protocol's oracles only in a ``generic'' way,
typically in the context of a simulation of honest players carried out
by Byzantine players.  The exception is Theorem~\ref{lrtheorem}, which
investigates how the choice of oracles affects a protocol's
vulnerability to long-range attacks.  In preparation for that result,
and to furnish additional examples of our oracle formalism in action,
we next show how to model via oracles two different cryptographic
primitives that are useful for repelling long-range attacks.
 
\vspace{0,2cm} 
\noindent \textbf{Modeling a VDF}. To model VDFs, we suppose queries
to the VDF oracle $O_v$ must be of the form $(m,k)$, where
$k\in \mathbb{N}$ and $m$ is arbitrary. If $p$ submits the query
$(m,k)$ at timeslot $t$, then it receives the response $(m,k)$ at
timeslot $t+k$, where the response is of type $O_v$ (in effect, signed
by the oracle).  A player can send a message that includes such a
type-$O_v$ entry only after receiving it as a response from the oracle
or in a message from another player.


\vspace{0,2cm} 
\noindent \textbf{Modeling ephemeral keys}.  For our purposes, we can
model (the signatures produced by) ephemeral keys using an oracle $O$
to which player $p$ may submit any finite number of requests of the
form $((id,m),t')$ at any timeslot~$t$, where the entry $(id,m)$ is of
signed type (and must be permitted for $p$). If $t'> t$, then $O$
sends $p$ the response $((id,m),t')$ at timeslot~$t$, which is an
entry of type $O$ (``signed by $O$,'' in effect). If $t'\leq t$, then
$O$ sends some predefined ``null'' response at~$t$ (intuitively, the
key for timeslot~$t'$ has already expired).

\subsection{Consensus is not possible in the fully permissionless setting without resource restrictions}
\label{2.5} 

We conclude this section by showing that the Byzantine Agreement
problem is unsolvable without some notion of ``resources'' to limit
the collective power of the Byzantine players.

 \vspace{0,2cm} 
 \noindent \textbf{Defining the Byzantine Agreement Problem (BA)}. We
 must modify the standard definition of the problem to accommodate
 players that may not be active at all timeslots.  Suppose $d=\infty$
 and that each player is given an input in $\{ 0,1 \}$. We refer to
 the latter input as the \emph{protocol} input for $p$, as
 distinguished from other \emph{parameter} inputs, such as the
 duration~$d$ and $\mathtt{id}(p)$. 
For simplicity, in this section we consider deterministic solutions to
the problem; probabilistic solutions are taken up in Section~\ref{ps2}.
We say that a deterministic
 protocol ``solves Byzantine Agreement'' if, for every protocol
 execution consistent with the setting, there exists a timeslot
 $t^{\ast}$ for which the following three conditions are satisfied:

\textbf{Termination}. All honest players active at any timeslot $\geq t^{\ast}$ terminate and give an output in $\{ 0,1 \}$. 

\textbf{Agreement}. All honest players that output give the same output. 

\textbf{Validity}.  If all players are given the same input
$i$,  then every honest player that outputs gives  $i$ as their
output.

 \vspace{0,2cm} 
 \noindent \textbf{Solving BA is impossible in this setting.}  Because
 we have imposed no bounds on the ``combined power'' of Byzantine
 players, no protocol (even with the use of oracles) can solve
 Byzantine Agreement in the model that we have described thus
 far.

\begin{proposition}\label{prop:noresources}
Consider the fully permissionless setting (with oracles but without
resources) and suppose that the player set~$\mathcal{P}$ is
finite. For every $\rho > 0$, no protocol solves Byzantine
agreement when up to a $\rho$ fraction of the players may be
Byzantine. This result holds even in the synchronous setting, with a
known player set, and with all players active at all
timeslots.
\end{proposition}

\begin{proof}
We give a proof for the deterministic case, but the argument extends
easily to probabilistic protocols and the probabilistic version of
Byzantine Agreement defined in Section~\ref{instance}; see also footnote~\ref{foot:noresources}.

Towards a contradiction, suppose that such a protocol
exists. Fix a set of parameter inputs with $\Delta=1$. Choose $n\geq 3$ such that $2/n < \rho$
and consider a set of players  $\mathcal{P}:=\{ p_0,p_1,\dots,p_{n-1}
\}$.
Let $\mathtt{id}_1,\mathtt{id}_2,\ldots,\mathtt{id}_{n-1},
\mathtt{id}_1',\mathtt{id}_2',\ldots,\mathtt{id}_{n-1}'$ denote $2n-2$
disjoint sets of identifiers. We consider three possible executions of
the given protocol:
\begin{enumerate}

\item Execution $\mathtt{E}_1$: 
\begin{itemize}

\item Player~$p_0$ is Byzantine, players~$p_1,p_2,\ldots,p_{n-1}$ are honest.  

\item All players are active in all timeslots.

\item For each honest player~$p_i$, $\mathtt{id}(p_i)=\mathtt{id}_i$.

\item All players receive an input of~0.

\item For~$p_0$, $\mathtt{id}(p_0)= \cup_{i=1}^{n-1} \mathtt{id}_i'$.

\item $p_0$ ignores its input and simulates the behavior (including
  messages sent and oracle queries made) of~$n-1$
  honest (and fictitious) players~$p'_1,p'_2,\ldots,p'_{n-1}$ that all
  have an input of~1 and in which
  $\mathtt{id}(p_i')=\mathtt{id}_i'$.\footnote{Player~$p_0$ is in a
    position to carry out this simulation because neither honest
    players nor oracles know whether the player to which an identifier belongs
 is Byzantine, and because there is no limit
    on the number of messages or oracle queries and responses that a
    player can send
    or receive in a single timeslot.}

\end{itemize}

\item Execution $\mathtt{E}_2$: 

\begin{itemize}

\item Player~$p_0$ is Byzantine,
  players~$p_1,p_2,\ldots,p_{n-1}$ are honest. 

\item All players are active in all timeslots.

\item For~$p_1$, $\mathtt{id}(p_1)=\mathtt{id}_1$. For every other
honest player~$p_i$, $\mathtt{id}(p_i)=\mathtt{id}_i'$.

\item Player~$p_1$ receives
  an input of~0 and all other  players receive an input of~1.

\item For~$p_0$, $\mathtt{id}(p_0)= \mathtt{id}_1' \cup \bigcup_{i=2}^{n-1} \mathtt{id}_i$.

\item $p_0$ simulates~$n-1$ honest players~$p'_1,p'_2,\ldots,p'_{n-1}$
under the conditions that~$p'_1$ has an input of~1,
$p_2',p_3',\ldots,p_{n-1}'$ each has an input of~0, 
$\mathtt{id}(p'_1)=\mathtt{id}_1'$, and
$\mathtt{id}(p_i')=\mathtt{id}_i$ for each~$i=2,3,\ldots,n-1$.

\end{itemize}

\item Execution $\mathtt{E}_3$: 

\begin{itemize}

\item Players~$p_0$ and~$p_1$ are Byzantine,
  players~$p_2,p_3,\ldots,p_{n-1}$ are honest. 

\item All players are active in all timeslots.

\item All players receive an input of~1.

\item Identifier assignments for all players are the same as in
  execution $\mathtt{E}_2$.

\item $p_0$ and $p_1$ behave identically as in
  execution $\mathtt{E}_2$.  (In particular, player~$p_1$ acts as if
  its input is~0.)

\end{itemize}

\end{enumerate}
The following observations complete the contradiction:
\begin{enumerate}

\item By validity and termination, in $\mathtt{E}_1$, all honest
  players must eventually output~0.

\item By validity and termination, in $\mathtt{E}_3$, all honest
  players must eventually output~1.

\item Because the executions $\mathtt{E}_1$ and $\mathtt{E}_2$ are
  indistinguishable for~$p_1$---that is, $p_1$ receives the same
  inputs and the same messages and oracle responses at each timeslot
  in both executions---$p_1$ must eventually output~0 in $\mathtt{E}_2$.

\item By agreement and termination, in $\mathtt{E}_2$, all honest
  players must eventually output~0.

\item Because the executions $\mathtt{E}_2$ and $\mathtt{E}_3$ are
  indistinguishable for all the players~$p_2,p_3,\ldots,p_{n-1}$, they
  must output the same value in both executions.
\end{enumerate}
\end{proof}

\noindent \textbf{Discussion}. The proof of
Proposition~\ref{prop:noresources} is driven by the fact that
identifier sets are undetermined and of unbounded
size (i.e., that Byzantine players can employ
sybils).\footnote{This aspect of the model will not feature in any of
  our other impossibility results as, from Section~\ref{fp} onward,
  protocols will be allowed to make use of ``resources'' that,
  in many scenarios, suffice to nullify the sybil challenge.}
In particular, the unbounded size assumption is crucial:
The positive result of Khanchandani and
Wattenhofer~\cite{khanchandani2021byzantine} shows that
Proposition~\ref{prop:noresources} no longer holds with a finite
undetermined player set and undetermined player allocation when $|\mathtt{id}(p)|=1$
for every player~$p$.

The statement of Proposition~\ref{prop:noresources} assumes that the
player set is finite. A similar result holds quite trivially in the
case that the player set is infinite after dropping the (no longer
meaningful) assumption that at most a $\rho$ fraction of the players
are Byzantine.
Consider two infinite sets of (always-active) players $A$ and $B$. Let
players in $A$ have input 0, while those in $B$ have input 1. If all
players are honest, then all must terminate and give a common
output. Without loss of generality, suppose they all output 0 in this
case.  Then validity is violated in the case that all players receive input 1, but the players in $A$
are Byzantine and behave exactly like honest players with
input~0.



\vspace{0.2cm} 

Proposition~\ref{prop:noresources} shows that, to make progress, we
must somehow bound the combined power of Byzantine players, and in a
way different from the traditional approach of bounding the fraction
of players that are Byzantine.  How should we phrase such a bound?
The proof of Proposition~\ref{prop:noresources} fundamentally depends
on the ability of Byzantine players to costlessly generate an
arbitrarily large number of sybils, which in turn enables the
simultaneous simulation of an arbitrarily large number of honest
players.  The bound must therefore concern something that is scarce
and costly to acquire---something that prevents a Byzantine player
from carrying out a simulation of many honest players.  These ideas
bring us to the notion of {\em resources} that limit the number of
times that a player can query an oracle in a single timeslot; these
are discussed next.


\section{External resources} \label{fp} 
  

This section introduces {\em external} resources, which typically
represent the hardware (such as ASICs or memory chips) that is
used by validators in proof-of-work (PoW) or proof-of-space (PoSp)
protocols.  Sections~\ref{belch} and \ref{pdr}  consider on-chain resources such as
staked cryptocurrency.


\subsection{External resources and permitters}\label{ss:external}
      
\emph{Permitter oracles}, or simply {\em permitters}, are
required for modeling external resources (but not for on-chain
resources).  In the fully permissionless setting (with resource
restrictions), a protocol may use a finite set of permitters.
These are listed as part of the protocol 
amongst the oracles in~$\mathcal{O}$, but have some distinguished
features. Like other oracles, permitters may be deterministic or
probabilistic.
  
\vspace{0,2cm} 
\noindent \textbf{The resource allocation}. For each execution of the
protocol, and for each permitter oracle $O$, we are given a
corresponding \emph{resource allocation}, denoted $\mathcal{R}^O$.  We
assume that $\mathcal{R}^O$ is undetermined, as befits an ``external''
resource.  The resource allocation can be thought of as assigning
each player some amount of the external resource at each
timeslot.  That is, for all $p\in \mathcal{P}$ and $t\leq d$, we have
$\mathcal{R}^O(p,t)\in \mathbb{N}$.  We refer to $\mathcal{R}^O(p,t)$
as $p$'s {\em resource balance} at $t$.  Because resource balances are
unknown to a protocol, an inactive player might as well have a zero
resource balance: $\mathcal{R}^O(p,t)=0$ whenever $p$ is not active at
$t$.  (For example, a protocol can't know whether a player owns no
ASICs, or owns a large number of ASICs but currently has them all turned
off.)  For each $t$, we also define
$\mathcal{R}^O(t):= \sum_p \mathcal{R}^O(p,t)$.  

\vspace{0,2cm} 
\noindent \textbf{Restricting the adversary}. An arbitrary value
$R^O_{\text{max}}$ is given as a protocol input. This value is
determined, but the protocol must function for any given value of
$R^O_{\text{max}}\geq 1$.\footnote{Our assumption that the (possibly
  very loose) upper bound $R^O_{\text{max}}$ on the total resource
  usage is given as an input reflects the fact that protocols such as
  Bitcoin are not expected to function properly for arbitrarily
  changing total hashrates. For example, the Bitcoin protocol need not
  be secure (even in the synchronous setting) when the total hashrate
  can increase extremely quickly. (One could alternatively bound the
  maximum rate of change of $\mathcal{R}^O(t)$, but we'll work here
  with the simpler assumption.)
} Let $B$ denote the set of Byzantine players and define
$\mathcal{R}^O_B(t):= \sum_{p\in B} \mathcal{R}^O(p,t)$.  For
$\rho \in [0,1]$, we say that $\mathcal{R}^O$ is {\em $\rho$-bounded}
if the following conditions are satisfied for all $t\leq d$:
  \begin{itemize} 
  \item   $\mathcal{R}^O(t)\in [1, R^O_{\text{max}}]$. 
  \item $\mathcal{R}^O_B(t)/\mathcal{R}^O(t) \leq \rho$.
  \end{itemize}  
The smaller the value of~$\rho$, the more severe the restriction on
the combined ``power'' of the Byzantine players.  
When $\rho < 1$, the assumption that
$\mathcal{R}^O(t) \ge 1$ for all $t\leq d$ avoids the degenerate case
in which there are no active honest players. If all resource allocations corresponding to permitters in~$\mathcal{O}$ are $\rho$-bounded, then we say the adversary is \emph{externally} $\rho$-\emph{bounded}.

Generally, protocols in the fully permissionless setting will be
expected to function for arbitrary and undetermined resource
allocations, so long as they are externally $\rho$-bounded for a
certain specified $\rho$.\footnote{Further restrictions on resource
  allocations can be imposed when convenient, as when we model the
  (proof-of-space) Chia protocol in Appendix~\ref{app:chia}.}

    
\vspace{0,2cm} 
\noindent \textbf{Permitter oracles}. The difference between permitter
oracles and other oracles is that the queries that a player $p$ can
send to a permitter oracle $O$ at timeslot $t$ depend on
$\mathcal{R}^O(p,t)$.  If a player $p$ sends a query to the permitter
oracle $O$ at timeslot $t$, the query must be of the form $(b,\sigma)$ such
that $b\in \mathbb{N}$ and $b\leq \mathcal{R}^O(p,t)$. Importantly,
this constraint applies also to Byzantine players.

A player may make multiple queries~$(b_1,\sigma_1),\ldots,(b_k,\sigma_k)$ to a
permitter in a single timeslot~$t$---possibly across multiple
iterations within the timeslot, as described in
Section~\ref{ss:st}---subject to the following constraints.  With a
{\em single-use} permitter---the appropriate version for modeling the
Bitcoin protocol (Section~\ref{ss:bitcoin})---these queries are
restricted to satisfy $\sum_{i=1}^k b_i \le \mathcal{R}^O(p,t)$.
(Interpreting the~$b_i$'s as hashrates devoted to the queries, this
constraint asserts that none of a player's overall hashrate can be
``reused.'')  With a {\em multi-use} permitter---the more convenient
version for modeling the Chia protocol in
Appendix~\ref{app:chia}---the only restriction is
that~$b_i \le \mathcal{R}^O(p,t)$ for each~$i=1,2,\ldots,k$.
(Intuitively, space devoted to storing lookup tables can be reused
across different ``challenges.'')




\subsection{Modeling PoW in Bitcoin}\label{ss:bitcoin}

We next show how to model proof-of-work, as used in the Bitcoin
protocol, using our permitter formalism.  The basic idea is that, at
each timeslot, a player should be able to request permission to
publish a certain block, with the permitter response depending on the
player's resource balance (i.e., hashrate).  To integrate properly
with the ``difficulty adjustment'' algorithm used by the Bitcoin
protocol to regulate the average rate of block production (even as the
total hashrate devoted to the protocol fluctuates), we define the
permitter oracle so that its responses can have varying ``quality.''
Each PoW oracle response will be a 256-bit string~$\tau$, and we can
regard the quality of $\tau$ as the number of 0s that it begins with.



\vspace{0.2cm} 
\noindent \textbf{An example}. As an example, suppose that
$\mathcal{R}^O(p,t)=5$. Then $p$ might send a single query
$(5,\sigma)$ at timeslot $t$, and this should be thought of as a
request for a PoW for $\sigma$ (which could represent a block
proposal, for example).  If $p$ does not get the response that it's
looking for, it might,
at a later timeslot $t'$ with $\mathcal{R}^O(p,t')=5$, 
send another query $(5,\sigma)$ to the PoW oracle.

\vspace{0.2cm} 
\noindent \textbf{Formal details}.  We consider a single-use
permitter. All queries to the permitter must be of the form
$(b,\sigma)$, where $\sigma$ is any finite string.  As already
stipulated above, the fact that $O$ is a single-use permitter means
that $p$ can send the queries
$(b_1,\sigma_1),\dots (b_k,\sigma_k)$ at the same timeslot $t$ only if
$\sum_{i=1}^k b_i\leq \mathcal{R}^O(p,t)$.

\vspace{0.1cm} 
\noindent \textbf{Sampling the response function}. At the beginning of
the protocol execution, we choose a response function $f $, as
follows. For each triple $(b,\sigma,t)$, we independently sample
$b$-many 256-bit strings uniformly at random and let $\tau$ be the
lexicographically smallest of these sampled strings.
We define $f((b,\sigma),t)$ to be $(r,t)$, where $r:=(\sigma,\tau)$ is
an entry of type $O$.\footnote{Recall from Section \ref{pswrr} that
  players cannot send messages with a specific entry of type $O$ until
  they have received a message or a response including that same entry
  of type $O$.  For example, a Byzantine player cannot acquire an
  original PoW without querying the PoW oracle directly.}
%
That is, if $p$ submits the query $(b,\sigma)$ at timeslot $t$, it
receives the above type-$O$ response $r=(\sigma,\tau)$ at timeslot
$t$.


\section{Blockchain protocols and stake} \label{belch} 
  
For an external resource (Section~\ref{fp}), players' resource
balances evolve according to an exogenous function $\mathcal{R}^O$
that does not depend on the execution of the protocol under
consideration.  On-chain resources (such as staked cryptocurrency)
evolve in lockstep with a protocol's execution, and as such require a
more sophisticated model.  This section focuses on the specific case
of stake, which is relevant, in particular, to the fully
permissionless setting (with resource restrictions)
defined in Section~\ref{ps2}.
More general protocol-defined resources, which are defined in Section~\ref{pdr},
are not relevant to the fully permissionless setting but play an
important role in the quasi-permissionless setting studied in
Sections~\ref{defpermyay}--\ref{addedsec}.

%

\subsection{Stake and transaction confirmation} \label{bps} 

On-chain resources, and stake in particular, make sense in the context
of blockchain protocols, meaning protocols that \emph{confirm} certain
messages such as \emph{transactions}.\footnote{Many blockchain protocols
decide on a total ordering on the set of confirmed transactions. 
The general form of ``payment system'' considered here also captures,
for example, DAG-based (``directed acyclic graph'') protocols that do
not specify such a total ordering (e.g., \cite{sompolinsky2016spectre}).}
  
 \vspace{0.2cm} 
 \noindent \textbf{The initial stake distribution}. Players are
 allocated an \emph{initial stake distribution}, denoted $S_0$, as
 input.\footnote{``Stake'' could be interpreted as the cryptocurrency
   balance of an account, or more restrictively as the amount of an
   account's cryptocurrency that is held in escrow in a designated
   staking contract.  Except where necessary, we won't distinguish
   between these two interpretations.\label{foot:stake}}  This distribution allocates
a positive integer amount of stake to each of a 
finite number of identifiers,
 and can be thought of as chosen by an adversary, subject to any
 constraints imposed on the fraction of stake controlled by Byzantine
 players.
  
\vspace{0.2cm} 
\noindent \textbf{The environment}.  For each execution of the
blockchain protocol, there exists an undetermined \emph{environment},
denoted $\mathtt{En}$, which sends messages of \emph{transaction type}
to players---messages with a single entry, which is of
\emph{transaction type}. An entry/message $\mathtt{tr}$ of transaction
type, henceforth referred to as a \emph{transaction}, is not permitted
for $p$ until received by $p$ (directly from the environment, or from
another player).
 
If $\mathtt{En}$ sends $\mathtt{tr}$ to $p$ at timeslot $t$, then $p$
\emph{receives} $\mathtt{tr}$ at $t$ as a member of its multiset of
messages received at that timeslot. Formally, the environment
$\mathtt{En}$ is simply a set of triples of the form
$(p,\mathtt{tr},t)$ such that $p$ is active at $t$. We stipulate that,
if $(p,\mathtt{tr},t)\in \mathtt{En}$, then $p$ receives the
transaction $\mathtt{tr}$ at $t$, in addition to the other
disseminations that it receives at $t$. We assume that, for each
$p\in \mathcal{P}$ and each $t$, there exist at most finitely many
triples $(p,\mathtt{tr},t)$ in $\mathtt{En}$.

One may think of the environment as being chosen by an adversary,
subject to constraints that are detailed below.

\vspace{0.2cm} 
\noindent \textbf{The stake allocation function}.  We assume some
fixed notion of validity for transactions: Any set of transactions
$\mathtt{T}$ may or may not be \emph{valid} relative to $S_0$.  We say
that $\mathtt{tr}$ is {\em valid relative to $S_0$ and~$\mathtt{T}$} if
$\mathtt{T}\cup \{ \mathtt{tr} \}$ is valid relative to $S_0$.  If
$\mathtt{T}$ is a set of transactions that is valid relative to $S_0$,
we let $\mathtt{S}(S_0,\mathtt{T}, id)$ ($\in \mathbb{N}$) be the
stake owned by identifier $id$ after execution of the transactions in
$\mathtt{T}$. It will also be notationally convenient to let
$\mathtt{S}(S_0,\mathtt{T})$ denote the function which on input $id$
gives output $\mathtt{S}(S_0,\mathtt{T}, id)$.%
\footnote{We are
  assuming that the stake distribution depends only on the subset of
  transactions that have been processed, and not on the order in which
  they were processed. (Similarly, validity is assumed to be a
  property of an unordered set of transactions, not of an ordering of
  transactions.)  This property holds for payment systems, but does
  not generally hold in Turing-complete blockchain protocols. For our
  impossibility results, this restriction on the possible transactions
  only makes them stronger.  Theorem~\ref{posPoS}, a positive result,
  holds more generally when the transaction order matters. The role of
  this restriction in Theorem~\ref{posPoS2} is discussed at length in
  Sections~\ref{addedsec} and~\ref{thirteen}.}

We assume that stake can be transferred and that ``negative balances''
are not allowed.\footnote{If stake represents cryptocurrency in a
  designated staking contract (see footnote~\ref{foot:stake}), then we are
  additionally assuming that each account is free to stake or unstake
  any amount of the cryptocurrency owned by that account.}
Formally,
suppose $\mathtt{T}$ is a valid set of transactions relative to $S_0$.
Suppose 
$\mathtt{S}(S_0,\mathtt{T},id_1)=x_1$ and
$\mathtt{S}(S_0,\mathtt{T},id_2)=x_2$.
Then, for each
$x\in \mathbb{N}$ with $x\leq x_1$, there exists a transaction
$\mathtt{tr}$ that is valid relative to $\mathtt{T}$ and $S_0$ and
such that
$ \mathtt{S}(S_0,\mathtt{T} \cup \{ \mathtt{tr} \},id_1)=x_1-x$,
$ \mathtt{S}(S_0,\mathtt{T} \cup \{ \mathtt{tr} \},id_2)=x_2+x$, and  $\mathtt{S}(S_0,\mathtt{T} \cup \{ \mathtt{tr} \},id)=  \mathtt{S}(S_0,\mathtt{T} ,id)$ for $id\notin \{ id_1, id_2 \}$. If
$\mathtt{tr}$ is as specified and $\mathtt{T}'$ is a set of
transactions which is valid relative to $S_0$ with
$\mathtt{S}(S_0,\mathtt{T}',id_1)<x$, then
$\mathtt{T}' \cup \{ \mathtt{tr} \}$ is not a valid set of
transactions.

We say transactions $\mathtt{tr}_1$ and $\mathtt{tr}_2$ are
\emph{conflicting} relative to $\mathtt{T}$ and $S_0$ if
$\mathtt{T}\cup \{ \mathtt{tr}_1 \}$ and
$\mathtt{T}\cup \{ \mathtt{tr}_2 \}$ are both valid relative to $S_0$,
but $\mathtt{T}\cup \{ \mathtt{tr}_1 \} \cup \{ \mathtt{tr}_2 \}$ is
not valid relative to $S_0$.
   
\vspace{0.2cm}
\noindent \textbf{Confirmed transactions}. Each \emph{blockchain protocol}
specifies a \emph{confirmation rule} $\mathcal{C}$, which
is a function that takes as input an arbitrary set of messages $M$ and
returns a subset $\mathtt{T} \subseteq M$ of the transactions among
those messages; the confirmed transactions $\mathtt{T}$
must be valid with respect to the initial stake distribution~$S_0$.
At timeslot $t$, if $M$ is
the set of all messages received by an honest player~$p$ at timeslots
$\leq t$, then $p$ regards the transactions in $\mathcal{C}(M)$ as
\emph{confirmed}.\footnote{In principle, a confirmation rule could
  depend on the timeslots at which an honest player 
  received its messages---necessarily timeslots at which the player
  was active---in addition to the messages themselves.  This
  distinction plays an interesting role in our treatment of long-range
  attacks in Section~\ref{longr}. (All other results in this paper
  hold whether or not a confirmation rule is allowed to depend on the
  timeslots in which messages were received.)\label{foot:conf_rule}}
For a set of messages~$M$, 
define $\mathtt{S}(S_0,M,id):=\mathtt{S}(S_0,\mathcal{C}(M),id)$ and
$\mathtt{S}(S_0,M):=\mathtt{S}(S_0,\mathcal{C}(M))$.

The confirmation rule $\mathcal{C}$ is separate from the state
diagram~$\Sigma$ (i.e., the protocol instructions), and one may
therefore consider different confirmation rules for the same protocol
instructions.  In the fully permissionless setting, a blockchain
protocol is thus a triple $(\Sigma, \mathcal{O},\mathcal{C})$. For
example, the Bitcoin protocol instructs players to extend the longest
chain.
One confirmation rule might extract the ``longest chain'' from $M$ and
regard a transaction as confirmed if it belongs to a block that is at
least 6 deep in that chain. Another rule might regard a transaction as
confirmed once it belongs to a block that is at least 10 deep in the
longest chain.

\vspace{0.2cm} 
\noindent \textbf{Defining $\rho$-bounded adversaries.}  We say that an
execution of a  protocol is {\em $\rho$-bounded} if:
\begin{itemize}

\item The adversary is externally $\rho$-bounded (in the sense of Section~\ref{ss:external}).


\item Among active players, Byzantine players never control more than
  a $\rho$ fraction of the stake.  Formally,
for every honest player $p$ at timeslot $t$,
if~$\mathtt{T}$ is the set of
transactions that are confirmed for $p$ at $t$ in this execution,
then at most a $\rho$ fraction of the stake allocated to 
players active at $t$ by $\mathtt{S}(S_0,\mathtt{T})$ is allocated to
Byzantine players.

\end{itemize}
 
\noindent When we say that ``the adversary is $\rho$-bounded,'' we mean that we
restrict attention to $\rho$-bounded executions.  In a slight abuse of
language, we say that a protocol is {\em $\rho$-resilient} when it
satisfies a certain claimed functionality under the assumption that
the adversary is $\rho$-bounded. So, for example, we might say that a
protocol for solving BA is $\rho$-resilient if it solves BA for
$\rho$-bounded adversaries.\footnote{The restriction to executions
  that are $\rho$-bounded may seem natural---what could
  one hope to
  prove for an execution of a proof-of-stake protocol in which Byzantine players amass enough stake
  to exceed a critical threshold?  But it is arguably preferable and
  more interpretable to instead impose a protocol-independent restriction on
  the {\em environment}; this point is explored at length in
  Section~\ref{addedsec}.}



\subsection{Liveness and Consistency}\label{ss:liveness}

With our notions of transactions and confirmation in place, we can
proceed to defining \emph{liveness} and \emph{consistency} for blockchain
protocols.
These definitions are complicated by the fact that
probabilistic protocols such as Bitcoin and
Algorand can offer consistency guarantees only for protocol executions
of finite duration.
For this reason, our liveness definition requires
that transactions are confirmed sufficiently quickly (i.e., bounded by
a parameter sublinear in the duration).
Our definitions must also accommodate players that are not always active.
   
\vspace{0.2cm} 
\noindent \textbf{Defining liveness}. Blockchain protocols are run
relative to a determined input $\epsilon \in [0,1)$, which is called the
\emph{security parameter}.
(In the deterministic model, $\epsilon=0$.)
Liveness requires the existence of a value $\ell$, which can depend on
the determined inputs (such as~$\epsilon$ and~$d$) but must be
sublinear in~$d$ (i.e., with $\ell=o(d)$ as $d \rightarrow \infty$), 
such that with probability at least $1-\epsilon$, the following
condition is satisfied for every~$t$.
Suppose that:
\begin{itemize}

  \item $t^{\ast}:=\text{max}\{ \text{GST}, t \} + \ell$ is at most $d$.
  \item
    The transaction $\mathtt{tr}$ is 
received by an honest player at some timeslot $\leq t$,
and $\mathtt{tr}$ is
     valid for all honest players through timeslot~$t^*$. (Formally,
 for
     every honest $p$ and every $t'\in [t,t^{\ast}]$, if $\mathtt{T}$
     is the set of transactions confirmed for $p$ at $t'$, then
     $\mathtt{T}\cup \{ \mathtt{tr} \}$ is a valid set of transactions
     relative to $S_0$.)
   \end{itemize} 
   \noindent
     Then, $\mathtt{tr}$ is confirmed for all
   honest players active at any timeslot after $t^{\ast}$, and is
   confirmed for those players at the first timeslot $\geq t^{\ast}$
   at which they are active.\footnote{This liveness requirement would  become
     vacuous if~$\ell$ were allowed to depend linearly on~$d$.
We work with the weak liveness requirement considered here
(with~$\ell$ required
only to be sublinear in~$d$) in order to make our impossibility
results as strong as possible.}
   
    \vspace{0.2cm} 
    \noindent \textbf{Defining consistency}. Consistency requires that, with
    probability at least $1-\epsilon$, the following two conditions
    always hold:
\begin{itemize} 
\item[(i)] No roll-backs: If $\mathtt{tr}$ is confirmed for honest $p$ at $t$,
  then $\mathtt{tr}$ is confirmed for $p$ at all $t'\geq t$. 
\item[(ii)] Confirmed transactions never conflict: If $\mathtt{T}$ and $\mathtt{T}'$ are the sets of transactions confirmed for honest $p$ and $p'$ at $t$ and $t'$, respectively, then $\mathtt{T} \cup \mathtt{T}'$ is a valid set of transactions relative to $S_0$. 
\end{itemize}

\section{The fully permissionless setting (with resource restrictions)} \label{ps2}

\subsection{Defining the setting and the main result}\label{ss:fp}

The fully permissionless setting (with
resource restrictions) is the same as the 
setting without resources
described in Section \ref{pswrr}, but with the following changes:
  
  \begin{enumerate} 
  \item[(a)] Protocols may now  make use of a finite set of permitters
    (and their corresponding resource allocations), as described in Section \ref{fp},  
   and may now run relative to an environment that sends transactions
   (which might, for example, update a stake distribution), as described in Section \ref{belch}.%
\footnote{In the fully permissionless setting, we do not make any
assumptions about the activity status of players with non-zero
stake. For example, having stake locked up
in a staking contract does not prevent a player from going offline for
any number of reasons.  For an external resource, by contrast,
inactive players might as well possess a zero resource balance (a
protocol cannot distinguish between different external resource
balances that might be possessed by an inactive player).}

   \item[(b)] One may assume that the adversary is $\rho$-bounded (as
     defined in Section~\ref{belch}), for some $\rho\in [0,1]$.
 \end{enumerate} 
  

\noindent In the fully permissionless setting, a protocol is a triple
$(\Sigma, \mathcal{O}, \mathcal{C})$. We will sometimes consider
protocols that are not blockchain protocols (i.e., that do not receive
any transactions from an environment, only the initial protocol
inputs) and do not need a confirmation rule; in this case,
$\mathcal{C}$ can be interpreted as, for example, a function that always
returns the empty set.

\vspace{0.2cm} Now that we have defined the fully permissionless
setting (with resource restrictions), we can state our first main
result---the analog of the FLP
impossibility result for the fully permissionless (and synchronous)
setting mentioned in Section~\ref{ss:results}. (Recall the definition
of the Byzantine Agreement Problem from
Section~\ref{2.5}.)


  
\begin{theorem} \label{fmt} 
Consider the fully permissionless,
authenticated, and synchronous setting, and 
suppose that, for some~$\rho > 0$, 
the adversary is $\rho$-bounded. Suppose that dishonest players can
only deviate from honest behavior 
by crashing or delaying message dissemination by
an arbitrary number of timeslots. 
Every
deterministic protocol for the Byzantine Agreement Problem has
an infinite execution in which honest players never terminate.
\end{theorem}  

We present the proof of Theorem~\ref{fmt} in Section \ref{proofoffmt}, 
along with a discussion of which aspects of the fully permissionless
setting matter for the proof.\footnote{Could a 
  simple proof of Theorem \ref{fmt} be achieved by adapting the
  proof of Dolev and Strong \cite{dolev1983authenticated} that $f+1$
  rounds are necessary to solve Byzantine Agreement with $f$ faulty
  players? (The intuition being that, if a new Byzantine player can enter
  in each round, then no finite number of rounds should suffice to solve
  Byzantine Agreement.) The difficulty with this approach is that the
  non-existence of a protocol with bounded termination time does
  \emph{not} imply the non-existence of a protocol that always
(eventually) terminates, unless one can establish an appropriate compactness
  condition on the execution space. Such an approach can be made
  to work (as in the conference version of this paper
  \cite{lewis2021byzantine}), but the alternative argument we present in
Section~\ref{proofoffmt} is arguably somewhat simpler.}

\vspace{0.2cm}
Can the impossibility results in Proposition~\ref{prop:noresources} or
Theorem~\ref{fmt} be circumvented by a
probabilistic solution to BA?  
The next section introduces the definitions necessary to discuss such
questions.


  \subsection{Defining 
the probabilistic setting} \label{instance}

  \textbf{The meaning of probabilistic statements}. For a given
 (possibly randomized) protocol, one way to completely specify an
 execution is via the following breakdown:
  
  \begin{enumerate} 
\item[1.]  The set of players $\mathcal{P}$ and their inputs;
\item[2.]  The player allocation; 
\item[3.]  The inputs and resource allocation for each permitter; 
\item[4.] The set of Byzantine players 
and their state transition diagrams;
\item[5.] The timing rule;
\item[6.] The transactions sent by the environment (and to which players and when); 
\item[7.] The (possibly probabilistic) state transitions of each player at each timeslot; 
\item[8.] The (possibly probabilistic) responses of permitter oracles.
\end{enumerate}  

When we say that a protocol satisfies a certain condition (such as
solving the Byzantine Agreement Problem), we mean that this holds for
all values of (1)-(8) above that are consistent with the setting. We
call a set of values for (1)-(6) above a \emph{protocol
  instance}. When we make a probabilistic statement to the effect that
a certain condition holds with at most/least a certain probability,
this means that the probabilistic bound holds for all protocol
instances that are consistent with the setting (with the randomization
over the realizations in~(7) and~(8)).\footnote{This approach
  implicitly restricts attention to protocol instances that
  can be defined independently of the realizations in~(7) and~(8); for
  example, we do not allow the transactions issued by the environment
  to depend on the outcomes of players' coin flips. The proofs of our
  impossibility results for probabilistic protocols
  (Proposition~\ref{prop:noresources} and Theorems~%
  \ref{psm}, \ref{NTT}, \ref{ortheorem}, \ref{lrtheorem}, and \ref{niceneg}) all use only
  instances of this restricted type; this restriction to ``oblivious''
adversaries
  and environments only makes those results stronger.  
This discussion is not relevant for the other results in the paper
(Theorems~\ref{fmt}, \ref{posPoS}, and \ref{posPoS2}), which concern
deterministic protocols.\label{foot:adaptive}}

\vspace{0.2cm} 
 \noindent \textbf{Indistinguishability of instances}.  We say that two protocol instances are \emph{indistinguishable for $p$ at timeslots $\leq t$} if both of the following hold:  (a) Player $p$ receives the same inputs for both instances, and; (b)  The distributions on messages and responses  received by $p$ at each timeslot $t'\leq t$ are the same for the two instances.  
  We say that two protocol instances are \emph{indistinguishable for
    $p$} if they are indistinguishable at all timeslots
  $<\infty$.\footnote{For example, with this definition of
    indistinguishability, steps~(3) and~(5) of the proof of
    Proposition~\ref{prop:noresources} apply more generally to
    probabilistic protocols. (For a probabilistic protocol with security
    parameter~$\epsilon$, steps~(1), (2), and~(4) of that proof hold
    except with at
    most $\epsilon$ probability each, leading to a contradiction
    provided $\epsilon < 1/3$.)\label{foot:noresources}}

\vspace{0.2cm}
\noindent \textbf{Discussion}.
Some results that are impossible in the deterministic setting become
possible in the probabilistic setting.
%
%
For example,
Garay, Kiayias, and Leonardos~\cite{garay2018bitcoin} prove that, for
every $\rho< 0.5$ and $\epsilon > 0$, the Bitcoin protocol
can be used to define a $\rho$-resilient protocol (in the sense of Section~\ref{bps})
that solves probabilistic BA with security parameter~$\epsilon$ in the
synchronous and fully permissionless setting (with
resource restrictions).\footnote{More precisely, for every~$\rho <
  1/2$, $\epsilon > 0$, and value $\Delta > 0$ for the delay bound
  parameter, there exists a finite~$T$ such
  that the variant of the Bitcoin protocol with an average interblock
  time of~$T$ can be used to define such a $\rho$-resilient protocol.\label{foot:bitcoin}}
That is, the impossibility result in Theorem~\ref{fmt} does not extend
to the probabilistic setting, even when arbitrary Byzantine behavior
is allowed.\footnote{An easy argument---similar to
  the argument in the discussion following
  Proposition~\ref{prop:noresources}---shows that no protocol solves
  probabilistic BA with a security parameter $\epsilon < 1/2$ when $\rho \ge 1/2$, even in the permissioned and
  authenticated setting (with resource restrictions).}
Interestingly, this positive result holds even in the
unauthenticated setting.%
%
%
\footnote{Pease,
  Shostak, and Lamport~\cite{pease1980reaching} and Fischer, Lynch,
  and Merritt~\cite{fischer1986easy} show that, already in the permissioned
  setting, probabilistic BA cannot be solved without authentication
  when $\rho \ge 1/3$, even with synchronous communication.  Why
  doesn't this impossibility result contradict the positive result
  (for all~$\rho < 1/2$) in~\cite{garay2018bitcoin}? Because the
  Bitcoin protocol uses a permitter (see Section~\ref{ss:bitcoin}) and
  the negative results in~\cite{pease1980reaching,fischer1986easy} consider only
  protocols that do not use any resources or permitters. Intuitively,
  resources and permitters can partially substitute for authentication
  in limiting the ability of a Byzantine player to simulate 
  honest players.\label{foot:flm}}

\section{The dynamically available setting} \label{dasection} 

Theorem~\ref{fmt} 
establishes limitations on what is possible in the fully
permissionless setting; are there palatable additional assumptions
under which stronger positive results are possible?  Relatedly, the
fully permissionless setting studied in Section~\ref{ps2} does not
meaningfully differentiate between protocols that use only external
resources (like typical proof-of-work protocols) and those that use
on-chain resources (like typical proof-of-stake protocols) as, in the
fully permissionless setting, all the players with on-chain resources
might well be inactive.  The way forward, then, would seem to be
additional assumptions on how player activity correlates with on-chain
resource balances.  The dynamically available setting, defined in
Section~\ref{ss:da}, is an instantiation of this idea.  (As is the
quasi-permissionless setting, which is defined in Section
\ref{defpermyay}.) Section~\ref{ss:sep} notes that a recent result of
Losa and Gafni~\cite{losa2023consensus} can be adapted to show that
the dynamically available setting is indeed meaningfully less
demanding than the fully permissionless setting (with resource
restrictions).

The primary goal of this section is to demonstrate fundamental
limitations on what protocols can achieve in the dynamically available
(and hence also the fully permissionless) setting: they cannot solve the
Byzantine Agreement problem in partial synchrony (Theorem~\ref{psm} in
Section~\ref{ss:da_imp1}); they cannot provide accountability, even in
the synchronous setting (Theorem~\ref{NTT} in
Section~\ref{ss:da_imp2}, which follows directly from work of Neu,
Tas, and Tse~\cite{neu2022availability}); and they cannot provide
optimistic responsiveness, even in the synchronous setting
(Theorem~\ref{ortheorem} in Section~\ref{ss:da_imp3}).  (With the
stronger assumptions of the quasi-permissionless setting, on the other
hand, Theorem~\ref{posPoS} shows that there is a protocol that attains
all three of these properties.)

\subsection{Defining the setting}\label{ss:da}

  
  


In the dynamically available setting, a protocol is a tuple
$(\Sigma,\mathcal{O},\mathcal{C})$.  No assumptions are made about
participation by honest players, other than the minimal assumption
that, if any honest player owns a non-zero amount of stake, then at
least one such player is active. (Additional assumptions about the
fraction of stake controlled by active honest players are, as usual,
phrased using the notion of $\rho$-bounded adversaries from
Section~\ref{bps}.)

\vspace{0.2cm} 
\noindent \textbf{Dynamically available setting.} Consider the protocol $(\Sigma,\mathcal{O},\mathcal{C})$.  By
definition, an execution of the protocol is {\em
  consistent with the dynamically available setting} if:
\begin{itemize}

\item [] Whenever~$p$ is honest and active at~$t$, with~$\mathtt{T}$
  the set of transactions confirmed for~$p$ at $t$ in this execution, if
  there exists an honest player assigned a non-zero amount of stake in
  $\mathtt{S}(S_0,\mathtt{T})$, then at least one such player is
  active at~$t$.

\end{itemize}

\subsection{Separating the dynamically available and fully
  permissionless settings}\label{ss:sep}

The additional assumptions imposed by the dynamically available
setting, relative to the fully permissionless setting, may appear
modest. But the following result, due (essentially) to Losa and
Gafni~\cite{losa2023consensus}, provides a formal separation
between the two settings: the impossibility result in
Theorem~\ref{fmt} no longer holds in the dynamically available
setting, even in the unauthenticated case.


\begin{theorem}[Losa and Gafni \cite{losa2023consensus}]\label{sep} 
Consider the dynamically available, unauthenticated, and synchronous
setting, and suppose that, for some~$\rho < 1/2$, the adversary is $\rho$-bounded. 
There is a deterministic protocol that solves the Byzantine Agreement
problem provided dishonest players can deviate from honest
behavior only by crashing or delaying message dissemination by
an arbitrary number of timeslots.  
\end{theorem}  

The protocol used in the proof of Theorem~\ref{sep} is similar in
spirit to the permissioned Dolev-Strong
protocol~\cite{dolev1983authenticated}.  Losa and
Gafni~\cite{losa2023consensus} phrase their results in a model that
differs from ours in several respects; for completeness, we include a
full proof of Theorem~\ref{sep} in Appendix~\ref{app:lg}.  The
additional assumptions of the dynamically available setting circumvent
the impossibility result in Theorem~\ref{fmt} by allowing a protocol
to ignore, without suffering an immediate loss of liveness, all
identifiers with no initial stake.  Restricting attention to
identifiers with non-zero initial stake has the effect of bounding the
maximum number of identifiers that could be controlled by faulty
players, which in turn opens the door to a Dolev-Strong-style
approach.

%



\subsection{The dynamically available and partially synchronous setting}\label{ss:da_imp1}

Theorem~\ref{fmt} established that deterministic protocols cannot
solve BA in the fully permissionless setting, even with synchronous
communication. The following result shows that, in the dynamically
available and partially synchronous setting (and therefore also in the fully permissionless and partially synchronous setting), BA cannot be solved, even
in the probabilistic sense and with no Byzantine players.\footnote{Recall that
a protocol is $\rho$-resilient if it achieves the desired
functionality---such as solving BA, or satisfying consistency and
liveness with a security parameter~$\epsilon$---for adversaries that are
$\rho$-bounded in the sense of Section~\ref{bps}. Thus, a
``0-resilient'' protocol is required to function correctly only when
all external resources and stake are always controlled by honest
players.}

\begin{theorem} \label{psm} 
For every $\epsilon < \tfrac{1}{3}$, there is no 0-resilient protocol
solving probabilistic BA with security parameter~$\epsilon$
in the dynamically available, authenticated, and partially synchronous setting.
\end{theorem} 


We present the proof of Theorem~\ref{psm} in Section \ref{proofofpsm}, 
along with a discussion of which aspects of the dynamically available
setting matter for the proof.  As noted in
Section~\ref{ss:rw_results}, the argument resembles the logic
behind the ``CAP Theorem.''

\vspace{0.2cm} Theorem \ref{posPoS} in Section \ref{defpermyay} shows
that, in the quasi-permissionless and partially synchronous setting,
there is a (deterministic) proof-of-stake protocol that solves BA.
Theorems~\ref{psm} and~\ref{posPoS} thus provide a formal separation
between the dynamically available and quasi-permissionless settings,
with the former setting strictly more difficult than the
latter.  


%


\subsection{Accountability in the dynamically available setting}\label{ss:da_imp2}

The remaining results in this section concern blockchain protocols
in the sense of Section~\ref{belch}.
Intuitively, such a protocol is
``accountable'' if there is always a ``smoking gun'' for consistency
violations, such as two signatures by a common identifier that support
conflicting transactions. Formally, we adopt the definition of
``forensic support'' by Sheng et\ al. \cite{sheng2021bft}---who
studied only the traditional permissioned setting without external or
on-chain resources---to our framework.


\vspace{0.2cm} 
\noindent \textbf{Defining accountability}.  Consider a blockchain
protocol $(\Sigma,\mathcal{O},\mathcal{C})$ and
security parameter $\epsilon$.
For a player set
$\mathcal{P}$ and set of identifiers $\mathtt{id}^*$, let
$P(\mathtt{id}^*):= \{ p\in \mathcal{P}: \ \mathtt{id}^*\cap
\mathtt{id}(p) \neq \emptyset \}$
denote the players that control at least one identifier
in~$\mathtt{id}^*$. For $\rho_1\in [0,1]$, we say that $\mathtt{id}^*$
is of \emph{weight at least} $\rho_1$ in a given execution if either:
\begin{itemize} 

\item there exists $t$ and a permitter $O\in \mathcal{O}$ such that
  $\sum_{p\in P(\mathtt{id}^*)} \mathcal{R}^O(p,t)\geq \rho_1\cdot \mathcal{R}^O(t)$; or

\item there exists $t$ and honest $p$, with $\mathtt{T}$ the set of
  confirmed transactions for $p$ at $t$, such that, for some $\mathtt{T}'\subseteq \mathtt{T}$,
  $\mathtt{S}(S_0,\mathtt{T}')$ allocates the identifiers in  $\mathtt{id}^*$ at least a
  $\rho_1$-fraction of the total stake.

\end{itemize}  
We say that a \emph{consistency failure} occurs in an execution if
there exist honest $p$ and $p'$, with $\mathtt{T}$ and~$\mathtt{T}'$
the sets of transactions confirmed for $p$ and $p'$ at timeslots $t$
and $t'$, respectively, such that $\mathtt{T} \cup \mathtt{T}'$ is not
a valid set of transactions relative to $S_0$.  We say that a protocol
is $(\rho_1,\rho_2)$-\emph{accountable} (with security
parameter~$\epsilon$) if there exists a \emph{blame}
function $F$ which maps message sets to identifier sets and satisfies:
\begin{enumerate} 

\item only Byzantine players are ever blamed, meaning that in any
  execution, and   for any subset of messages $M$
  disseminated in that execution,
  $F(M)$ is a (possibly empty) set of identifiers belonging to Byzantine players; and

\item for each protocol instance in which the adversary is
  $\rho_2$-bounded, the following holds with probability
  at least $1-\epsilon$: if a consistency failure occurs, then there exists
  some finite set of disseminated messages $M$ with $F(M)$ of weight
  at least $\rho_1$.\footnote{To make the impossibility result in
    Theorem~\ref{NTT} as strong as possible, this definition deems a
    set of blamed identifiers ``high-weight'' in an
    execution if there exists any subset of transactions confirmed in
    that execution that allocate them a large fraction of the
    stake. The protocol used to prove our positive results (Theorems
    \ref{posPoS} and \ref{posPoS2}) offers a stronger form of
    accountability, with the set of blamed identifiers 
    allocated a large amount of stake specifically at an ``epoch
    boundary''; see Section~\ref{secprof} for details.}

\end{enumerate} 
A $\rho$-resilient protocol is vacuously
$(\rho_1,\rho)$-accountable for all~$\rho_1$ (with $F(M)=\emptyset$
for all~$M$), as with a $\rho$-bounded adversary consistency
violations can occur only with probability at most~$\epsilon$. Thus,
in the interesting parameter regime, $\rho_2$ exceeds the values
of $\rho$ for which $\rho$-resilient protocols exist.
In this regime, one cannot expect $(\rho_1,\rho_2)$-accountability
unless~$\rho_1 \le \rho_2$.  For fixed~$\rho_2$, then, the goal is to
design protocols with~$\rho_1$ as close to~$\rho_2$ as possible.

\vspace{0.2cm} 
\noindent \textbf{An impossibility result}. The proof of Theorem
\ref{psm} implies that, in the dynamically available, authenticated,
and partially synchronous setting, no blockchain protocol can be
$(\rho_1,\rho_2)$-accountable for any $\rho_1>0$ and $\rho_2 \ge 0$.
In the dynamically available, authenticated, and {\em synchronous}
setting, the interesting case is when $\rho_2\geq 1/2$.  (As discussed
in Section~\ref{instance} and footnote~\ref{foot:bitcoin}, for every
$\rho < 1/2$, there is a Bitcoin-style protocol that is
$\rho$-resilient in this setting.)  Neu, Tas, and
Tse~\cite{neu2022availability} have shown that accountability is not
possible in this case, and their proof is easily adapted to our formal
framework:

\begin{theorem}[Neu, Tas, and Tse \cite{neu2022availability}] \label{NTT}
For every $\epsilon < 1/4$, $\rho_1 > 0$, and $\rho_2 \ge 1/2$,
there is
no 0-resilient blockchain protocol that is $(\rho_1,\rho_2)$-accountable
with security parameter~$\epsilon$ in the dynamically available,
authenticated, and synchronous setting. 
\end{theorem} 


\noindent \textbf{Discussion}. The proof of Theorem~\ref{NTT}
in~\cite{neu2022availability}, once adapted to our model, is driven by
the possibility of inactive players in the dynamically available
setting.  This feature of the setting is necessary for
Theorem~\ref{NTT} to hold, as otherwise the set of players allocated
non-zero stake by the initial stake distribution could carry out a
permissioned and accountable state machine replication protocol (e.g.,
as described in~\cite{sheng2021bft}).  Inactive players serve two
purposes in the proof: (i) to force active honest players to confirm
transactions when the other players might plausibly be inactive;
and (ii) to allow for ``late-arriving'' players, who cannot
disambiguate between conflicting sets of allegedly confirmed transactions.

\vspace{0.2cm}

Theorem \ref{posPoS} in Section \ref{defpermyay} shows that, in the
quasi-permissionless, authenticated, and partially synchronous
setting, there is a proof-of-stake blockchain protocol that is
$(1/3,1)$-accountable. Together, Theorems \ref{NTT} and \ref{posPoS}
therefore establish another formal separation between what is possible
in the dynamically available and quasi-permissionless settings.

\subsection{Optimistic responsiveness in the dynamically available setting}\label{ss:da_imp3}

Ideally, a blockchain protocol would not only be consistent and live
in ``worst-case'' executions that are consistent with a given setting,
but would also automatically offer stronger guarantees (such as faster transaction confirmations) for ``non-worst-case'' executions.
Pass and Shi~\cite{pass2018thunderella} introduce the  notion of
\emph{optimistic responsiveness} to capture this idea.
%


\vspace{0.2cm} 
\noindent \textbf{Defining optimistic responsiveness}.  
Roughly speaking, for our purposes, an optimistically responsive
protocol is one that confirms transactions at network speed whenever
all players are honest.
The formal definition resembles that of liveness
(Section~\ref{ss:liveness}), with a liveness parameter~$\ell$ that
scales with the realized (post-GST) maximum message delay.

More formally, for a given protocol execution in the partially synchronous
setting with no Byzantine players, define $\delta$ as the realized
maximum message delay following GST (which will be at most the
worst-case delay bound~$\Delta$, but may be much smaller). Said
differently,~$\delta$ is the smallest number
such that whenever $p$ disseminates $m$ at $t$, every $p'\neq p$ that
is active at $t'\geq \text{max} \{ \text{GST}, t \} +\delta $ receives
that dissemination at a timeslot $\leq t'$. 
(In the special case of the synchronous setting, GST$=0$.)

A protocol is then \emph{optimistically responsive} (with
security parameter~$\epsilon$) if there exists a liveness
parameter~$\ell = O(\delta)$ and a ``grace period'' parameter
$\Delta^* = O(\Delta)$ such that, in every instance consistent with
the setting, with probability at least $1-\epsilon$,
the following condition is satisfied for every~$t$.
Suppose that:
\begin{itemize}

  \item $t^{\ast}:=\text{max}\{ \text{GST}+\Delta^*, t \} + \ell$ is at most $d$.

  \item
    The transaction $\mathtt{tr}$ is 
received by an honest player at some timeslot $\leq t$,
and $\mathtt{tr}$ is
valid for all honest players through timeslot~$t^*$
(in the same sense as in the definition of liveness
in
Section~\ref{ss:liveness}).

\end{itemize}
Then, $\mathtt{tr} $ is confirmed for all
honest players active at any timeslot after $t^{\ast}$, and is
confirmed for those players at the first timeslot $\geq t^{\ast}$ at
which they are active.\footnote{Intuitively, the ``grace period''
  $\Delta^*$ grants a protocol time to ``switch over'' to ``fast
  confirmation mode.''  Allowing a grace period that scales
  with~$\Delta$ is natural (and similar to the definition in Pass and
  Shi~\cite{pass2018thunderella}), and the impossibility result
  in Theorem~\ref{ortheorem} holds for this definition of optimistic
  responsiveness. The PoS-HotStuff protocol used to prove our positive results
  (Theorems~\ref{posPoS} and~\ref{posPoS2}), meanwhile, is in fact optimistically
  responsive with $\Delta^* = 0$, and also satisfies stronger
  definitions of optimistic responsiveness (as in, e.g.,~\cite{lewis2023fever}).}

%
%


\vspace{0.2cm} 
\noindent \textbf{An impossibility result}.  The next result asserts
that blockchain protocols cannot be optimistically responsive in the
dynamically available setting.

\begin{theorem} \label{ortheorem} 
For every $\epsilon < \tfrac{1}{3}$,  there is no $0$-resilient
blockchain protocol that is optimistically responsive
with security parameter~$\epsilon$
in the dynamically available, authenticated, and synchronous
setting.
\end{theorem} 

We present a proof of Theorem \ref{ortheorem}  in Section \ref{proofofortheorem}. 

Theorem \ref{posPoS} in Section \ref{defpermyay} shows that, in the
quasi-permissionless and partially synchronous setting, there is a
proof-of-stake protocol that is optimistically responsive. Thus,
Theorems~\ref{ortheorem} and \ref{posPoS} provide yet another formal
separation between the dynamically available and quasi-permissionless
settings.

\section{General on-chain resources}  \label{pdr} 
  
The discussion in Section~\ref{ss:da} noted that, in the fully
permissionless setting, protocols cannot generally take advantage of
stake or other on-chain resources because all the honest players with
a non-zero amount of some such resource might well be inactive.  The
dynamically available setting defined there insists on a very weak
correlation between player activity and on-chain resource balances: at
each timeslot, if there is an honest player with a non-zero amount of
stake, some such player must be active.
While Theorem~\ref{sep} demonstrates that this additional assumption
can lead to strictly stronger positive results,
Theorems~\ref{psm}--\ref{ortheorem} illustrate the persisting
limitations of the dynamically available setting. All three
of these impossibility results take advantage of the fact that, in this
setting, network delays and some types of Byzantine player behavior
can be indistinguishable from the inactivity of honest players.
Making further progress requires resolving this ambiguity via stronger
player activity assumptions; this is the idea behind the
``quasi-permissionless'' setting.

The quasi-permissionless setting can be
defined narrowly for ``pure proof-of-stake protocols,'' as is outlined
in Section~\ref{lfqp}, but its full versatility is unlocked once we allow
protocols to define arbitrary on-chain resources that are expected to
correlate with player activity.  Section~\ref{blockrewards} defines
such ``protocol-defined resources'' and notes examples of existing
protocols that make use of them.  Section~\ref{react} defines the
subclass of ``reactive'' protocol-defined resources, which rules out
certain ``essentially permissioned'' protocols and describes more
precisely the nature of the on-chain resources used by the protocols
surveyed in Section~\ref{blockrewards}.  The general
quasi-permissionless setting is then defined in
Section~\ref{defpermyay} with respect to arbitrary on-chain resources.

\subsection{Looking forward to the quasi-permissionless setting} \label{lfqp} 
  

For the purposes of this section, by a ``proof-of-stake (PoS)''
protocol, we mean a blockchain protocol $(\Sigma, \mathcal{O},
\mathcal{C})$ in 
which $\mathcal{O}$ contains no permitters.\footnote{Technically, a
  blockchain protocol with no stake distribution (and no permitters) qualifies as
  a PoS protocol according to this definition. To offer any
  non-trivial guarantees in a permissionless setting, however,
  blockchain protocols without permitters must make use of a
  non-trivial stake distribution and player activity assumptions with
  respect to it (otherwise, all honest players might well be
  inactive at all timeslots).}  (The full definition of PoS protocols,
in the more general setting that allows for protocol-defined
resources, is stated in Section~\ref{ss:qp}.)
Typically, PBFT-style PoS
protocols such as Algorand require stronger assumptions on the set of
active players than are provided by the dynamically available setting:
 
\vspace{0.2cm} 
\noindent \textbf{Quasi-permissionless setting (specifically for
  PoS protocols).}  An execution of a proof-of-stake protocol is {\em consistent
  with the quasi-permissionless setting} if:
\begin{itemize}

\item [] Whenever~$p$ is honest and active at~$t$, with~$\mathtt{T}$
  the set of transactions confirmed for~$p$ at $t$ in this execution,
  every honest player that is assigned a non-zero amount of stake in
  $\mathtt{S}(S_0,\mathtt{T})$ is active at~$t$.

\end{itemize}
By definition, every execution of a PoS protocol consistent with the
quasi-permissionless setting is also consistent with the dynamically
available setting.

Sections~\ref{blockrewards} and~\ref{react} describe terminology for
discussing forms of on-chain resources other than stake, as are used
by, for example,
%
Byzcoin \cite{kokoris2016enhancing}, Hybrid
\cite{pass2016hybrid}, and Solida \cite{abraham2016solida}.
Readers 
uninterested in on-chain resources other than stake can keep in mind
the simplified version of the quasi-permissionless setting defined
above and 
skip to Section \ref{posres}.

\subsection{Protocol-defined resources}  \label{blockrewards} 

  
What would be an on-chain resource other than stake?  Consider, for
example, the two-stage protocol of Andrychowicz and
Dziembowski~\cite{andrychowicz2015pow}.  The basic idea is to, in the
first stage, employ a proof-of-work protocol in which players compete
to establish identities that are recorded as having provided PoW.
(For example, one could interpret a public key as ``established'' in a
Bitcoin-style protocol if it contributed
one of the first 100 blocks of the longest chain.)  After a set of
identities has been established, in the second stage, those identities
then carry out a version of the (permissioned) Dolev-Strong protocol
\cite{dolev1983authenticated} to execute an instance of single-shot
consensus.
(The Andrychowicz-Dziembowski protocol, like the Dolev-Strong
protocol, is intended for the synchronous setting.)  This idea
translates to a protocol that relies on one external resource (for the
PoW) and one on-chain resource (e.g., a 0-1 function encoding which
identities have permission to participate in the second-stage
permissioned protocol).  Protocols of this type are typically assumed
to operate in a generalization of the simplified quasi-permissionless
setting
described in Section~\ref{lfqp}, with all honest
players with on-chain resources assumed to be active.

\vspace{0.2cm} Taking this idea further, protocols such as Byzcoin
\cite{kokoris2016enhancing}, Hybrid \cite{pass2016hybrid}, and Solida
\cite{abraham2016solida} implement state machine replication
\cite{schneider1990implementing} by using PoW to select a rolling
sequence of committees. The basic idea is that, once a committee is
selected, it should carry out a PBFT-style protocol to implement the
next consensus decision. This consensus decision includes the next
sequence of transactions to be committed to the blockchain and also
determines which identifiers have provided sufficient PoW in the
meantime for inclusion in the next committee. Again, such protocols
are assumed to operate in a generalization of the simplified
quasi-permissionless setting: all honest players with on-chain
resources (i.e., those belonging to the current committee) are
required to be active for the PBFT-style protocol to function
correctly.

\vspace{0.2cm} 
\noindent
\textbf{Defining protocol-defined resources}.  A
\emph{protocol-defined resource} is a function $\mathtt{S}^{\ast}$
such that $\mathtt{S}^{\ast}(M,id) \in \mathbb{N}$ for every identifier
$id$ and set of messages $M$, and such that
$\mathtt{S}^{\ast}(M,id) =0$ for identifiers~$id$ outside of $\cup_{p
  \in \mathcal{P}} \mathtt{id}(p)$.
Let $\mathtt{S}^{\ast}(M)$ denote the function which on input $id$ gives
output $\mathtt{S}^{\ast}(M,id)$.

\vspace{0.1cm} 
Thus, an honest player that has received the messages~$M$ by
timeslot~$t$ regards $\mathtt{S}^{\ast}(M)$ as the
balances of the protocol-defined resource at that time. (In effect,
each protocol-defined resource specifies its own
message semantics and confirmation rule.)
By an \emph{on-chain} resource, we mean a resource that is either
stake or protocol-defined.

\subsection{Reactive resources: not all protocol-defined resources are created equal} \label{react} 
  
Because the definition of protocol-defined resources is so general,
generalizing the simplified quasi-permissionless setting to
arbitrary such resources can lead to the consideration of
``essentially permissioned'' protocols.
%
%
For example, while the two-stage protocol by Andrychowicz and
Dziembowski~\cite{andrychowicz2015pow} only asks the selected
committee to carry out a single consensus decision,
a variant of that protocol could select a single static committee that
is then tasked with carrying out state machine replication (e.g.,
using a PBFT-style permissioned protocol) for an unbounded
duration. This protocol can be viewed as operating in a version of the
quasi-permissionless setting with unrestricted protocol-defined
resources, with one such resource indicating membership in the static
committee and the quasi-permissionless assumption then requiring that
committee members be active for the remainder of the protocol's
execution.
Can one call such a protocol ``(quasi-)permissionless'' with a
straight face?  


\vspace{0.2cm} Why aren't
the Byzcoin, Hybrid, and Solida protocols ``essentially permissioned''
in a similar sense?
And why didn't this issue come up in Section~\ref{belch} for the
specific case of stake as an on-chain resource?
Because stake, and the protocol-defined resources used by the Byzcoin,
Hybrid, and Solida protocols, can be forced to change through the
confirmation of certain messages and transactions that might be issued
by the environment or by players.  (E.g., recall from
Section~\ref{bps} our assumption that stake can always be transferred
via an appropriate transaction.)  To rule out ``essentially
permissioned'' protocols, then, it must be possible for players to
eventually lose their on-chain resources.  This leads to the notion of
``reactive'' resources, which cannot be held indefinitely by any
players that own no external resources or stake.


%

\vspace{0.2cm} 
\noindent \textbf{Generalizing the sense in which stake can be forced
  to change}. 
%
Formally, consider a protocol that uses $k$ permitters with
corresponding resources $\mathcal{R}_1,\dots, \mathcal{R}_{k}$.
Consider an execution of this protocol and a time
interval~$I=[t_1,t_2]$.  We say that a player~$p$:
\begin{itemize}

\item {\em owns no external resources in~$I$} if 
$\mathcal{R}_i(p,t)=0$ for all $t\in I$ and $i\in      
\{1,\ldots,k\}$; 

\item {\em owns no stake in~$I$} 
if, for every $id \in
  \texttt{id}(p)$, 
$\mathtt{S}(S_0,\mathtt{T},id)=0$ whenever $\mathtt{T}$ is the set of
transactions confirmed for an honest player~$p'$ at a timeslot~$t \in I$;




\item {\em owns none of a protocol-defined resource
    $\mathtt{S}^{\ast}$ at~$t$} if, for every $id \in \texttt{id}(p)$,
  $\mathtt{S}^{\ast}(M,id)=0$ when\-ever~$M$ is the set of messages
  received by an honest player $p'$ at timeslots $\leq t$.

\end{itemize}
%
%
%
%
  
  
\vspace{0.2cm} 
\noindent \textbf{Reactive sets of resources}.  Let $\epsilon \ge 0$ be a
determined security parameter.
A set $\mathcal{S} = \{ \mathtt{S}^*_1,\ldots,\mathtt{S}^*_j \}$
of on-chain resources used by a protocol
is \emph{reactive} if
there exists a value $\ell$, which
may depend on the determined inputs but must be sublinear in
$d$, such that the following condition holds with probability
at least $1-\epsilon$ for each protocol instance: 
\begin{itemize}

\item [] If $t_1\geq$ GST, $t_2-t_1\geq \ell$, $t_2\leq d$, and each
  of the players of some set~$P$ owns no external resources or stake
  in $I:=[t_1,t_2]$, then for every $p \in P$ and $i \in \{1,\ldots,j\}$,
$p$ owns none of $\mathtt{S}^{\ast}_i$ at $t_2$. 


\end{itemize}


\vspace{0.2cm} 
\noindent \textbf{Restricting to reactive sets of on-chain resources}. The
quasi-permissionless setting, defined in the next section, requires a
protocol to explicitly list the 
on-chain resources to which the setting's player activity assumption
applies.
Unless one wants to also consider permissioned protocols, one
generally restricts 
attention to protocols that use a reactive set of on-chain resources
(or some variation thereof).



\section{The quasi-permissionless setting}\label{defpermyay} 
  
\subsection{Defining the setting} \label{ss:qp}
  
For the purposes of the general quasi-permissionless setting,
a protocol is a tuple
$(\Sigma,\mathcal{O},\mathcal{C},\mathcal{S})$, where
$\Sigma, \mathcal{O}$, and $\mathcal{C}$ are as in the fully
permissionless and dynamically available settings, and where
$\mathcal{S}=\{ \mathtt{S}^{\ast}_1,\dots, \mathtt{S}^{\ast}_k \}$ is
a set of on-chain resources (possibly including stake~$\mathtt{S}$).
A {\em proof-of-stake (PoS)} protocol is then, for the purposes of
this paper, one in which
$\mathcal{O}$ contains no permitters and $\mathcal{S}$ contains 
no on-chain resources other than stake.

 

\vspace{0.2cm} 
\noindent \textbf{Quasi-permissionless setting.} Consider the protocol
$(\Sigma,\mathcal{O},\mathcal{C},\mathcal{S})$, where
$\mathcal{S}=\{ \mathtt{S}^{\ast}_1,\dots, \mathtt{S}^{\ast}_k \}$.
By definition, an execution of the protocol is {\em consistent with
  the quasi-permissionless setting} if the following hold for all
$i\in \{1,\ldots,k\}$:
\begin{itemize}

\item [(i)] If $\mathtt{S}_i^{\ast}=\mathtt{S}$ then, whenever~$p$ is
  honest and active at~$t$, with~$M$ the set of messages received
  by~$p$ at timeslots $\le t$ in this execution, every honest player that is
  assigned a non-zero balance by~$\mathtt{S}(S_0,M)$ is active
  at~$t$. (Recall that $\mathtt{S}(S_0,M)$ is shorthand for $\mathtt{S}(S_0,\mathcal{C}(M))$.)
  
\item [(ii)] If $\mathtt{S}_i^{\ast}\neq \mathtt{S}$ then, whenever~$p$
  is honest and active at~$t$, with~$M$ the set of messages received
  by~$p$ at timeslots $\leq t$ in this execution, every honest player
  that is assigned a non-zero balance by~$\mathtt{S}^{\ast}_i(M)$ is
  active at~$t$.

\end{itemize}
Thus, the quasi-permissionless setting insists on activity from every
honest player that possesses any amount of any of the on-chain
resources listed in the protocol description.\footnote{An interesting
  direction for future work is to investigate analogs of the
  quasi-permissionless setting for protocols that use only external
  resources.  For example, what becomes possible when the sum of the
  resource balances of each external resource is fixed and known?}
Every execution consistent with the quasi-permissionless setting is
also consistent with the dynamically available setting.\footnote{A
  protocol $(\Sigma,\mathcal{O},\mathcal{C},\mathcal{S})$ can be
  treated in the dynamically available setting by ignoring
  $\mathcal{S}$.}
%
The definition of $\rho$-bounded adversaries remains the same as in
Section \ref{bps}, and refers only to external resources and stake.

\vspace{0.2cm} 
\noindent \textbf{Proving impossibility results for the
  quasi-permissionless setting}.  
As noted in Section~\ref{react}, coupling the activity assumptions of
the quasi-permissionless setting with unrestricted protocol-defined
resources allows for ``essentially permissioned'' protocols, such as
protocols that treat the initial stake distribution as a
(non-reactive) protocol-defined resource encoding the identifiers
allowed to participate in some permissioned consensus protocol. Thus,
proving impossibility results for the quasi-permissionless setting
(that do not already apply to the permissioned setting) requires
imposing constraints on the protocol-defined resources allowed, such
as the reactivity restriction from
Section~\ref{react}. Theorem~\ref{niceneg} in Section~\ref{addedsec}
is an example of such an impossibility result.

\subsection{A positive result for the quasi-permissionless setting} \label{posres} 

The next result provides a blockchain protocol that, in the
quasi-permissionless setting, overcomes the impossibility results in
Theorems~\ref{psm}--\ref{ortheorem} for the dynamically available
setting. Specifically, in the quasi-permissionless,
authenticated, and partially synchronous setting:
\begin{itemize}
\item For every~$\rho < 1/3$, the blockchain protocol is
  $\rho$-resilient (with respect to consistency and liveness) and
  induces a $\rho$-resilient
  protocol for solving BA. (Cf.,
    Theorem~\ref{psm}.)\footnote{As discussed at length in
Section~\ref{ss:rw_results}, several previous works strongly
hint at such a result while stopping short of providing a full proof.}
  \item The protocol is $(1/3,1)$-accountable. (Cf.,
    Theorem~\ref{NTT}.)\footnote{Sheng et al.~\cite{sheng2021bft}
      show how to convert several permissioned protocols
      into~$(1/3,2/3)$-accountable protocols. Our protocol
      achieves~$(1/3,1)$-accountability
%
through more stringent requirements for block confirmation. Roughly,
an honest player will regard a 
block of
transactions as confirmed
only after seeing a full ``plausible history'' (at least
within each ``epoch'') leading to confirmation of that block.}

\item The protocol is optimistically responsive with
  $\Delta^*=0$. (Cf., Theorem~\ref{ortheorem}.)
\end{itemize}
Moreover, the protocol can be made deterministic.
        


\begin{theorem} \label{posPoS} 
Consider the quasi-permissionless, authenticated, and partially
synchronous setting. For every
$\rho<1/3$, there exists a deterministic and $\rho$-resilient PoS
blockchain protocol.
Moreover, the protocol can be made $(1/3,1)$-accountable
and optimistically responsive. 
\end{theorem} 

We present the proof of Theorem~\ref{posPoS} in Section~\ref{thirteen}.
The protocol used in the proof, which we call the PoS-HotStuff
protocol, can be viewed as an extension of the (permissioned) HotStuff
protocol~\cite{yin2019hotstuff} to a permissionless PoS protocol. (One
could similarly adapt a protocol such as
PBFT~\cite{castro1999practical}.)  Although the protocol is designed
to make the proof of Theorem~\ref{posPoS} as straightforward as
possible, the proof is complicated by a number of difficulties that
may not be immediately apparent.  
For example, one issue revolves
around the heights at which the protocol ensures that players receive
\emph{direct} confirmation of a block (e.g., for HotStuff, receiving
three quorum certificates (QCs) for a block as opposed to only for one
of its descendant blocks).
Player changes at an indirectly confirmed height 
can lead to consistency violations. One approach to 
addressing this
issue is to ensure direct confirmation of blocks at all heights, 
but honest leaders must then
re-propose blocks initially proposed by Byzantine leaders, which in
turn 
threatens liveness (which requires the confirmation of honestly
proposed blocks). Our proof instead
divides the execution into epochs, each of which contains potentially
infinitely many views. Each epoch is concerned with producing
confirmed blocks up to a certain height, but may also produce blocks
of greater height that receive three QCs in three rounds of voting but 
are not considered confirmed.  More generally, the proof of
consistency must now address the fact
that different players may see different versions of the
blockchain (especially during asynchrony) and therefore possess
different beliefs
as to who should be producing and voting on the block in each
round. In particular, the definition of a ``quorum certificate'' is
now both player- and time-relative.  
See Section~\ref{overv} for a more detailed overview of some of the
technical challenges involved in the proof of Theorem~\ref{posPoS}.

It would be interesting to extend versions of Theorem~\ref{posPoS} to
one or more of the major PoS blockchain protocols that have been
deployed in practice.

\section{Long-range attacks} \label{longr} 

    
\vspace{0.2cm} 
\noindent \textbf{A weakly dynamic adversary}.  Long-range attacks on
PoS protocols fundamentally involve the corruption of honest players
(i.e., a non-static adversary).  To make our impossibility result as
strong as possible, we consider a particularly restricted type of
dynamic adversary.  Given a protocol with security parameter~$\epsilon
\ge 0$ (in the sense of Section~\ref{ss:liveness}), we consider only
protocol instances in which:
\begin{itemize}

\item Byzantine players have no initial stake.

\item Except with probability at most~$\epsilon$, no stake is transferred from
  honest to Byzantine players. 

\item 
%
At exactly one timeslot
$t_c$ in a protocol execution, there may be a \emph{corruption event}
for some set~$P$ of players that have ``cashed out,'' in the sense
that:
for every honest player $p$, if $\mathtt{T}$ is the set of
transactions confirmed for $p$ at $t_c$, 
then
$\mathtt{S}(S_0,\mathtt{T})$ allocates zero stake to all identifiers 
belonging to players in $P$.

If a corruption event occurs at $t_c$, then: 
\begin{itemize} 

\item The state diagram of each player $p\in P$ may be replaced by an
  arbitrary state diagram, and~$p$ may begin in any state of that
  diagram at $t_c$.

\item If $P_B$ was the set of initially Byzantine players, then at
  $t_c$ the set of Byzantine players becomes $P_B^{\ast}:=P_B\cup P$.

\item 
As in Section \ref{ss:permitted}, message entries permitted for one Byzantine
player are permitted for all Byzantine players: if $p\in P_B^{\ast}$
and $t \ge t_c$, then all entries 
permitted for $p$ at $t$ are also permitted for all $p'\in P^{\ast}_B$.
  
\end{itemize} 

\end{itemize}
The first restriction ensures that an adversary cannot  use the
initial stakes of the Byzantine players to cause violations of
consistency or liveness. The
second ensures that the environment cannot conspire with Byzantine
players to cause such violations by issuing transactions that would
force the transfer of stake to them. The third restriction (to
``cashed-out'' honest players) prevents an adversary from causing such
violations by using a corruption event to acquire stake
from stake-holding honest players.

\vspace{0.2cm} 

\noindent \textbf{Desired property.} 
A blockchain protocol is \emph{invulnerable to simple long-range
  attacks} if it satisfies consistency and liveness
(Section~\ref{ss:liveness}) with a dynamic adversary of the
restricted type above.
  
\vspace{0.2cm} Theorem \ref{lrtheorem} below establishes that PoS
protocols cannot be invulnerable to simple long-range attacks unless
they make use of cryptographic primitives with certain time-dependent
features, such as VDFs or ephemeral
keys.\footnote{Section~\ref{modhash} shows how VDFs and ephemeral keys
  can be modeled
  within our framework via suitably defined oracles. 
Neither of these oracles is 
time malleable in the sense defined in this section (with the VDF
oracle violating the first condition and the ephemeral keys oracle the
second condition). One-shot signatures \cite{amos2020one} could also
potentially play a role similar to ephemeral keys in defending against
long-range attacks.}


\vspace{0.2cm} 
\noindent \textbf{Time malleable oracles}. We call an oracle $O$ 
\emph{time malleable} if it satisfies the following  conditions
for each response function $f$ in the support of $O$ and for
  every pair $(q,t)$:

\begin{itemize} 

\item  $f(q,t)=(r,t)$ for some $r$, and; 

\item If  $f(q,t)=(r,t)$ then  $f(q,t')=(r,t')$ for all $t'$. 

\end{itemize}

%
\noindent These conditions assert that every oracle response is
delivered instantaneously and is independent of the timeslot in which
the query is made.
 
\begin{theorem} \label{lrtheorem}
For every $\epsilon < \tfrac{1}{4}$, no PoS blockchain protocol that
uses only time malleable oracles is 
invulnerable to simple long-range attacks with security
parameter~$\epsilon$ in the quasi-permissionless, authenticated, and synchronous setting.
\end{theorem} 
    
\begin{proof} 
Towards a
contradiction, suppose there is a PoS blockchain protocol
$(\Sigma,\mathcal{O},\mathcal{C},\mathcal{S})$ (with $\mathcal{S}=\{ \mathtt{S} \}$) that is invulnerable to 
simple long-range attacks in the setting described,
with security parameter $\epsilon < \tfrac{1}{4}$ and 
liveness parameter~$\ell$ (see Section~\ref{ss:liveness}).  
Fix all parameter inputs 
(see Section~\ref{ss:inputs}), which will be the same for all protocol
instances considered.  Assume that the duration~$d$ is larger
than~$\ell$.  Choose non-empty and pairwise disjoint sets of
players~$P_0$, $P_1$, and $P_2$. 
Let~$S_0$ denote the protocol's initial stake distribution, which we
assume is not the all-zero allocation.  Assign identifiers to players
such that every identifier assigned non-zero stake in~$S_0$ belongs to
$\mathtt{Id}(p)$ for some $p \in P_0$.  Set
$\mathtt{Id}_i:= \cup_{p\in P_i} \mathtt{Id}(p)$ for
$i \in \{0,1,2\}$.
Define~$\mathtt{T}_1$ and~$\mathtt{T}_2$ to be disjoint sets
of transactions, each valid with respect to~$S_0$, such that:
\begin{enumerate} 

\item[(i)] For $i\in \{ 1,2 \}$, each $\mathtt{tr}\in \mathtt{T}_{i}$
  takes some $id\in \mathtt{Id}_0$ such that
  $\mathtt{S}(S_0,\emptyset,id)>0$ and transfers all the
  corresponding stake to an identifier in $\mathtt{Id}_{i}$.
    
\item[(ii)] For $i\in \{ 1,2 \}$, $\mathtt{S}(S_0,\mathtt{T}_{i})$
  allocates all stake to identifiers in $\mathtt{Id}_{i}$.
\end{enumerate} 
%
The sets $\mathtt{T}_1$ and $\mathtt{T}_2$ exist because we assume
that stake is transferable (see Section~\ref{bps}).  Note that the set
$\mathtt{T}_1 \cup \mathtt{T}_2$ is not valid with respect to~$S_0$.
This completes the setup.

We first consider two instances of the protocol (in the sense of Section~\ref{instance}), $\mathtt{I}_1$
and $\mathtt{I}_2$.  For $i\in \{ 1,2 \}$, instance $\mathtt{I}_i$ is
specified as follows:
\begin{itemize} 

\item $\mathcal{P}:=P_0 \cup P_{i}$, with all players honest.

\item All players are active at all timeslots $\leq d$. 

\item All disseminations are received by all other players at the next
  timeslot.

\item The environment sends all transactions in $\mathtt{T}_i$ to all
  players at timeslot 1.

\end{itemize} 
%
Choose $t^{\ast}>\ell$.  
Generate ``coupled executions'' of the instances~$I_1$ and~$I_2$ by
first sampling a single response function
(shared across the two executions)
for each oracle $O \in
\mathcal{O}$ according to the distribution specified by~$O$
and then sampling state transitions
(timeslot-by-timeslot) independently in each of the instances.
For~$i=1,2$, let~$M_i$, a random variable, denote the set of all
messages disseminated by players at timeslots $\le t^*$ in 
such an execution of the protocol instance~$I_i$.

Liveness implies that, for each $i=1,2$,
$\mathbb{P}(\mathcal{C}(M_i))=\mathtt{T}_i) \ge 1 -
\epsilon$.
Because the output of a confirmation rule is a valid set of
transactions (see Section~\ref{bps}) and
$\mathtt{T}_1 \cup \mathtt{T}_2$ is not valid with respect to~$S_0$
(by construction),
there must be a value of $i^* \in \{1,2\}$
such that, with probability at least
$\tfrac{1}{2}(1-2\epsilon) = \tfrac{1}{2}-\epsilon$ over the
messages~$(M_1,M_2)$ disseminated in the coupled executions:
\begin{itemize}

\item [(1)] $\mathcal{C}(M_i)=\mathtt{T}_i$ for $i=1,2$;

\item [(2)] $\mathcal{C}(M_1 \cup M_2)$ does not contain all of
  $\mathtt{T}_{i^*}$.

\end{itemize}
Without loss of generality, suppose $i^{\ast}=1$. 

We then consider a third protocol instance $\mathtt{I}_3$, specified as follows: 
\begin{itemize} 

\item Players in $P_0$ and $P_1$ are initially honest, while players
  in $P_2$ are initially Byzantine.

\item All players are active at all timeslots $\leq d$. 

\item All disseminations are received by all other players at the next
  timeslot.

\item The environment sends all transactions in $\mathtt{T}_1$ to all
  players at timeslot 1.

\item Players in $P_2$ perform no actions until $t^{\ast}$, and
  perform no actions at any timeslot unless a corruption event occurs.

\item If all transactions in $\mathtt{T}_1$ are confirmed for all
  players in $P_0\cup P_1$ by $t^{\ast}$,
  then a corruption event for
  players in $P_0$ occurs at $t_c=t^{\ast}$.
(Thus, a corruption event occurs with probability at least~$ 1 - \epsilon$.)
The (Byzantine) players of~$P_2$ are in a position to detect this
event because they receive the same messages at the same timeslots as
the players of~$P_0 \cup P_1$.

\item 
The environment sends the players in $P_0 \cup P_2$ all transactions
in $\mathtt{T}_2$ at $t_c$.  (Provided a corruption event occurs,
these transactions do not transfer stake from honest to Byzantine
players.)\footnote{If we allowed the environment to issue transactions
  as a function of players' and oracles' coin flips (see
  footnote~\ref{foot:adaptive} in Section~\ref{instance}), then these transactions could be
  issued only when there is a corruption event, guaranteeing that no
  stake is ever transferred from honest to Byzantine players.}

\item If a corruption event occurs at timeslot $t_c$, then at this
  timeslot the players in $P_0\cup P_2$ simulate the first $t_c$ many
  timeslots of the protocol instance $\mathtt{I}_2$, disseminating all
  messages as dictated by the protocol instructions. 

\end{itemize} 
%
Crucially, when a corruption event occurs at $t_c$, the players in
$P_0\cup P_2$ are capable of simulating the first $t_c$ timeslots of
the protocol instance $\mathtt{I}_2$.  This is true because:
\begin{enumerate} 

\item[(i)] All oracles in $\mathcal{O}$ are time
  malleable. (Otherwise, the players of~$P_0 \cup P_2$ either wouldn't
  be able to reconstruct oracle queries from earlier timeslots
  $t' < t_c$ at the timeslot~$t_c$, or wouldn't receive the oracle
  responses quickly enough to simulate many timeslots in a single
  one.)

\item[(ii)] The instructions carried out by each Byzantine player can
  depend on the messages and oracle responses received by other Byzantine
  players, and if $p\in P_{B}^{\ast}$, then any entry/message
  permitted for $p$ is permitted for all $p'\in P_B^{\ast}$.  (If
  Byzantine players couldn't instantaneously simulate each other, they
  would need to communicate explicitly to carry out the simulation,
  and message delays would then preclude the simulation of many
  timeslots in a single one.)

\end{enumerate} 
%
%
Let $M$ denote the set of all messages disseminated by the end of
timeslot $t_c$ in instance $\mathtt{I}_3$. By construction,
conditioned on a corruption event, $M$ has the same distribution as
the union $M_1 \cup M_2$ of messages disseminated in the coupled
executions of~$I_1$ and~$I_2$ (conditioned on the event
that~$\mathcal{C}(M_1)=\mathtt{T}_1$).  Intuitively, in this case, an
execution of~$I_3$ generates a ``superposition'' of executions
of~$I_1$ and~$I_2$---one from the ``honest execution'' of~$I_1$ up to
timeslot~$t_c$, the other from the simulation of~$I_2$ that the
Byzantine players carry out at~$t_c$---which share a common set of
oracle response functions.  Thus, whenever conditions~(1) and~(2)
above hold for~$M_1$ and~$M_2$ (which in particular implies
$\mathcal{C}(M_1)=\mathtt{T}_1$ and a corruption event), instance
$\mathtt{I}_3$ sees a consistency violation, with at least one
transaction of $\mathtt{T}_1$ flipping from confirmed (at~$t_c$) to
unconfirmed (at~$t_c+1$).  Because~(1) and~(2) hold with probability
at least $1/2-\epsilon$ and $\epsilon < \tfrac{1}{4}$, the given
protocol does not satisfy consistency (with security
parameter~$\epsilon$), a contradiction.
\end{proof}

\noindent \textbf{Discussion}.  The proof of Theorem~\ref{lrtheorem}
uses the fact that, in the quasi-permissionless setting, the
set~$\mathcal{P}$ of players may be unknown ($P_0 \cup P_1$,
$P_0 \cup P_2$, or $P_0 \cup P_1 \cup P_2$).  The proof holds when
there is a finite and determined upper bound on the number of players
(namely, $|P_0 \cup P_1 \cup P_2|$).  Neither sybils nor the
possibility of player inactivity are used in the proof (players are
always active and can use unique identifiers).  Alternatively, the
proof can be rephrased as working with a known player set
($\mathcal{P} = P_0 \cup P_1 \cup P_2$) but relying instead on the
possibility that the players of~$P_1$ or~$P_2$ might be permanently
inactive.

Given that all players are always active and all messages are received
promptly after dissemination, the proof of Theorem~\ref{lrtheorem}
also relies on our assumption that the confirmation rule~$\mathcal{C}$
is a function only of the set of messages received by an honest
player, and not on the timeslots in which the player received those
messages (see footnote~\ref{foot:conf_rule} in Section~\ref{bps}).
Intuitively, a timing-dependent confirmation rule could, in instance
$\mathtt{I}_3$, resolve the conflicts between $\mathtt{T}_1$ and
$\mathtt{T}_2$ according to the timeslots at which a player received
the messages in~$M_1$ and~$M_2$ (as opposed to using the
timing-independent output $\mathcal{C}(M_1 \cup M_2)$ assumed in the
proof).  The proof can be modified to apply also to timing-dependent
confirmation rules by using players that become active only after the
dissemination of all the messages in~$M_1$ and~$M_2$; such players
would receive all those messages at the same time and would not know
which way to resolve the conflicts between the transactions
in~$\mathtt{T}_1$ and $\mathtt{T}_2$.\footnote{One of the simpler
  practical approaches to repelling long-range attacks is to require
  that newly active players communicate out-of-protocol with other
  (hopefully honest) players that {\em have} been active in order to
  learn which older messages should be regarded as confirmed; this
  idea relates to the ``weak subjectivity'' notion proposed by
  Buterin~\cite{wsbut}, and is also explored by Daian, Pass, and Shi~\cite{daian2019snow}.}  For similar reasons, the proof
can also be extended to timing-dependent confirmation rules with
always-active players in the partially synchronous setting (with
unbounded message delays substituting for unbounded player
inactivity).

  

\section{Restricting environments instead of executions} \label{addedsec}

Thus far, we have focused on possibility and impossibility results in
various settings with respect to a $\rho$-bounded adversary (for
some~$\rho \in [0,1]$), meaning for executions of a protocol in which
Byzantine players never control more than a $\rho$ fraction of an
external resource or of the stake controlled by active players (see
Section~\ref{bps}). The restriction on external resources, which
evolve independently of a protocol's execution, is straightforward to
interpret and mirrors the restrictions on the fraction of Byzantine
players that are familiar from the study of permissioned consensus
protocols.

Stake, on the other hand, evolves in tandem with a protocol's
execution. (The closest analog in a traditional permissioned state machine
replication model would be if the set of Byzantine players was somehow
a function of the set of transactions confirmed by the protocol.)  The
amount of stake controlled by Byzantine players---and, thus, whether
an execution qualifies as one for which a $\rho$-resilient protocol
must guarantee consistency and liveness---is determined by the set of
confirmed transactions.
Even after fixing the initial stake distribution and the transactions
issued by the environment, this set can depend 
on the realized network delays, the strategies of the
Byzantine players, and (somewhat circularly) on the decisions of the
protocol itself.

Section~\ref{ss:circle} sharpens this point with a concrete example
that shows how network delays can result in all stake being
transferred to a single Byzantine player, despite a seemingly
reasonable initial stake distribution and set of issued
transactions. Section~\ref{ss:circle2} builds on this idea and, in
Theorem~\ref{niceneg}, shows that no blockchain protocol can guarantee
consistency and liveness (for all executions)
with respect to any non-trivial adversary in the quasi-permissionless
and partially synchronous setting.  Section~\ref{envrho} introduces a
natural restriction on the transactions issued by the
environment---intuitively, that a transaction is issued only once its
``prerequisite transactions'' have been confirmed by some honest
player---under which the guarantees of Theorem~\ref{posPoS} can be
extended (with the same protocol) to apply no matter what the realized
network delays and Byzantine player strategies.


\subsection{The Payment Circle}\label{ss:circle}

A simple example illustrates how, even with an initial distribution
that allocates Byzantine players at most a $\rho$-fraction of the stake, and
even with a seemingly benign set of transactions, specific executions
may lead to Byzantine players controlling more than a~$\rho$ fraction
of the stake.
Consider a set $\{p_0, p_1,\ldots,p_{n-1}\}$ of~$n$ players, each with
a single unit of initial stake.  Suppose~$p_0$ is Byzantine and the
other $n-1$ players are honest.  Imagine that, in each of~$n$
consecutive rounds, for each player~$p_i$,
the environment issues a new transaction that
transfers a single unit of stake from~$p_i$
to~$p_{i+1 \bmod n}$. Let $\mathtt{T}$ denote the set of these~$n^2$
transactions.
%
%
Each round of~$n$ transactions leaves everyone's stake
unaffected---each~$p_i$ still controls one unit of stake---so what
could possibly go wrong?


Suppose we are working in the partially synchronous setting, and that
GST is well after the~$n$ rounds in which the transactions
$\mathtt{T}$ are issued by the environment.
Perhaps, for some subset $\mathtt{T}' \subseteq \mathtt{T}$,
only one honest player~$p$ receives the transactions in $\mathtt{T}
\setminus \mathtt{T}'$ from the environment.
Further, prior to GST, network delays may prevent~$p$ from
disseminating these transactions to other honest players.
%
If $\mathtt{T}'$ is a valid set of transactions and GST is
sufficiently large, liveness then forces the honest players other
than~$p$ to confirm all the
transactions in $\mathtt{T}'$.  (For these players,
the execution is indistinguishable from a synchronous
execution in which the transactions in $\mathtt{T}'$ are the only
ones ever issued by the environment.)

Now choose $\mathtt{T}'$ to contain the first~$i$ of the~$n$ transfers
from~$p_i$ to~$p_{i+1 \bmod n}$ (for each~$i=0,1,2,\ldots$).  This is
a valid set of transactions, with the (Byzantine) player~$p_0$ ending up
with all $n$ units of stake!



\subsection{A General Impossibility Result}\label{ss:circle2}

The payment circle example demonstrates that the following assumption
does not suffice to prevent Byzantine players from acquiring all the
stake (even when the environment never issues conflicting pairs of transactions): ``the execution of all transactions issued  prior to a given timeslot always results in Byzantine
players controlling at most a $\rho$ fraction of the stake.''
Any pure proof-of-stake protocol would fail to satisfy
liveness or consistency in the payment circle example.
Next, in Theorem~\ref{niceneg},
we prove a general impossibility result, which shows that analogous
examples plague every protocol that uses an arbitrary collection of
external resources and on-chain resources that, like
stake, are reactive in the sense defined in Section~\ref{react}.

For simplicity, and to make our impossibility result as strong as
possible, we assume that the transactions issued by the (adversarial)
environment are restricted to ``simple payments'' from one
identifier to another (like those in the payment circle example).  For
example, we can work with the UTXO model of Bitcoin.  We can interpret
the initial stake distribution~$S_0$ as an initial set of UTXOs (each
associated with some identifier).  A
transaction then spends one or more existing UTXOs (all associated
with the same identifier) and creates one or more new ones (each
associated with an arbitrary identifier), and two transactions
conflict if and only if they spend a common UTXO.  
The sum of the values of the UTXOs created by a transaction must be no
larger than the sum of the values of the UTXOs that it spends.
A set $\mathtt{T}$
of transactions is valid relative to~$S_0$ if and only if it contains
no conflicting transactions and every transaction spends only UTXOs
that are either in~$S_0$ or created by some other transaction
of~$\mathtt{T}$.

\vspace{0.2cm} 
\noindent \textbf{Maximal valid sets of transactions}. For the
statement of Theorem \ref{niceneg} below, we
cannot directly impose requirements on the stake distribution
resulting from ``the set of transactions
issued by the environment,'' because that set may not be valid
relative to the initial stake distribution~$S_0$. Instead, we impose a
requirement on every maximal valid subset of the transactions issued
by the environment up until some timeslot.

Formally, we say that
$\mathtt{T}$ is a \emph{maximal valid set of transactions at~$t$} if:
\begin{itemize}
\item [(i)] $\mathtt{T}$ is valid relative to $S_0$;
  \item [(ii)] Every transaction in
$\mathtt{T}$ is received by
at least one honest player from the environment 
at some timeslot $\leq t$, and;
\item [(iii)] No proper superset of $\mathtt{T}$
satisfying (i) and (ii) is valid relative to $S_0$.
\end{itemize}
Call an
environment {\em maximally $\rho$-bounded} with respect to the initial
stake distribution~$S_0$ if:
\begin{itemize}
\item whenever the environment sends 
a transaction  $\mathtt{tr}$ to some player at timeslot~$t$, 
it sends $\mathtt{tr}$ to some honest player at a timeslot $\le t$; and
\item for every~$t$ and every $\mathtt{T}$ which
is a maximal valid set of transactions at $t$,
$\mathtt{S}(S_0,\mathtt{T})$ allocates Byzantine players at most a
$\rho$-fraction of the stake. 
\end{itemize}
In particular, taking~$t=0$ in the second condition,
its conclusion must hold for the initial distribution~$S_0$.

\begin{theorem} \label{niceneg} 
Consider the quasi-permissionless, authenticated, and partially
synchronous setting.
For every~$\rho > 0$ and $\epsilon < \tfrac{1}{5}$, there is no
blockchain protocol $(\Sigma,\mathcal{O}, \mathcal{C}, \mathcal{S})$
that uses a reactive set $\mathcal{S}$ of on-chain resources and
satisfies consistency and liveness with security parameter~$\epsilon$
with respect to an externally $\rho$-bounded adversary and a maximally
$\rho$-bounded environment.
%
%
%
\end{theorem} 

\begin{proof} 
See Section \ref{twelve}. 
\end{proof} 

As foreshadowed in Section~\ref{ss:qp}, the proof of
Theorem~\ref{niceneg} must take advantage of the 
restriction to reactive
sets of on-chain resources.
Theorem~\ref{niceneg} does not hold without the reactivity assumption: 
Otherwise, provided Byzantine players control less than one-third of
the initial stake, a protocol could include the (fixed) initial stake
distribution as a (non-reactive) on-chain resource in $\mathcal{S}$
and, taking advantage of the player activity assumptions of the
general quasi-permissionless model (see Section~\ref{ss:qp}), use the
players allocated non-zero initial stake to carry out a 
permissioned state machine replication protocol such as
PBFT~\cite{castro1999practical}. 

As Theorem~\ref{niceneg} does not hold in the permissioned setting,
its proof also takes advantage of the limited non-permissioned
features of the quasi-permissionless setting: that a protocol does not
have advance knowledge of the player set, and that players without
resources may be inactive.


\subsection{A Sufficient Condition on the Environment} \label{envrho}

\vspace{0.2cm} 
\noindent \textbf{Defining $\rho$-bounded environments (for the UTXO
  model)}.  To define $\rho$-bounded environments, we again consider
the UTXO model of Bitcoin.  There is then, for any
transaction $\mathtt{tr}$, a smallest set~$\mathtt{T}$ of
transactions such that $\mathtt{T} \cup \{ \mathtt{tr} \}$ is valid
relative to $S_0$; we say that $\mathtt{T}$ is the set of transactions
{\em required for $\mathtt{tr}$ to be valid relative to $S_0$.}

The trouble in the payment circle example stemmed from the multiple
rounds of transactions, with each round produced without waiting for
confirmation of the transactions in previous rounds.  In response, we
now assume that an honestly produced transaction $\mathtt{tr}$ is
issued only once all the transactions required for $\mathtt{tr}$ to be
valid are already confirmed.  Because we do not wish to limit the
issuing of dishonestly produced transactions in any way, we
distinguish between {\em Byzantine} and {\em honest} transactions,
according to the type of player that controls the identifier
associated with the transaction's inputs.  Intuitively, we will assume
that at any moment in time, the immediate confirmation of all the
pending honest transactions and no new Byzantine transactions would
result in the Byzantine players controlling at most a $\rho$-fraction
of the total stake. (The remaining concern, then, is that the delayed
confirmation of pending honest transactions or the immediate
confirmation of some Byzantine transactions might nevertheless cause
Byzantine players to be allocated more than a $\rho$-fraction of the
total stake.)

Precisely, we say that {\em the environment is $\rho$-bounded}
if $\mathtt{S}(S_0,\emptyset)$ allocates the adversary at most a
$\rho$-fraction of the total stake and if the following holds for all
$p$ and $t$. Suppose that:
\begin{enumerate} 
\item[(i)]   Player $p$ is honest and  $\mathtt{T}_1$, which is the set of transactions confirmed for $p$ at $t$, is valid relative to $S_0$;
\item[(ii)]   $\mathtt{S}(S_0,\mathtt{T}_1)$ allocates the adversary at most a $\rho$-fraction of the total stake;
\item[(iii)] $\mathtt{T}_2$ is the set of all honest transactions sent by the environment at timeslots $\leq t$. 
\end{enumerate} 
Then: $\mathtt{T}_1$ contains all transactions required for any
transaction in $\mathtt{T}_2$ to be valid relative to $S_0$;
$\mathtt{T}_1\cup \mathtt{T}_2$ is valid relative to $S_0$; and
$\mathtt{S}(S_0,\mathtt{T}_1\cup \mathtt{T}_2)$ allocates the
adversary at most a $\rho$-fraction of the total stake.\footnote{We
  sometimes consider protocols that require players to hold stake in
  escrow to participate. In this case, it is convenient to consider a
  function $\hat{\mathtt{S}}$ (in addition to $\mathtt{S}$) that
  determines how much stake players have held in escrow.  A weakly
  $\rho$-bounded environment then, by definition, satisfies the
  following additional assumption: If $p$ is honest, if 
  the set~$\mathtt{T}$ of transactions confirmed for $p$ at $t$ is valid
  relative to $S_0$, and if $\mathtt{S}(S_0,\mathtt{T})$ allocates the
  adversary at most a $\rho$-fraction of the stake, then
  $\hat{\mathtt{S}}(S_0,\mathtt{T})$ also allocates the adversary at
  most a $\rho$-fraction of the stake in escrow.\label{foot:escrow}}

 \vspace{0.2cm} 
 \noindent \textbf{Some terminology}.  We say that the adversary is
 {\em weakly $\rho$-bounded} if all resource allocations are
 $\rho$-bounded and the environment is $\rho$-bounded. A blockchain
 protocol is {\em strongly $\rho$-resilient} if it is live and
 consistent for weakly $\rho$-bounded adversaries.

\vspace{0.2cm} 
The following result is a strengthening of Theorem~\ref{posPoS}, in
that our proof of Theorem~\ref{posPoS2} will also prove
Theorem~\ref{posPoS}.  The protocol used in the proof, PoS-HotStuff,
is the same in both cases.

\begin{theorem} \label{posPoS2} 
Consider the quasi-permissionless, authenticated, and partially
synchronous setting. For every
$\rho<1/3$, there exists a deterministic and strongly $\rho$-resilient PoS
blockchain protocol.
Moreover, the protocol can be made $(1/3,1)$-accountable
and optimistically responsive. 
\end{theorem} 

\begin{proof} 
See Section \ref{thirteen}. 
\end{proof} 

To get a feel for the role that the $\rho$-bounded environment plays
in the proof of Theorem~\ref{posPoS2}, suppose that~(i)--(iii) above
hold for sets~$\mathtt{T}_1$ and~$\mathtt{T}_2$ of transactions. We
claim that, for any $\mathtt{T}_3\subseteq \mathtt{T}_2$ and any set
of Byzantine transactions $\mathtt{T}_4$ such that
$\mathtt{T}_5:= \mathtt{T}_1\cup \mathtt{T}_3 \cup \mathtt{T}_4$ is
valid relative to $S_0$, $\mathtt{S}(S_0,\mathtt{T}_5)$ allocates the
adversary at most a $\rho$-fraction of the stake.  This claim holds
because $\mathtt{T}_1$ already contains all transactions required for
transactions in $\mathtt{T}_2$ to be valid---ruling out examples such
as the payment circle and ensuring that
$\mathtt{S}(S_0,\mathtt{T}_1\cup \mathtt{T}_3)$ allocates the
adversary at most a $\rho$-fraction of the total stake---and because
transactions in $\mathtt{T}_4$ can only transfer stake from
identifiers controlled by Byzantine players.  The significance of this
observation in the proof of Theorem~\ref{posPoS2} is that, when a new
block of transactions is confirmed by the protocol, no choice of
included transactions in that block can result in Byzantine players
being allocated more than a $\rho$-fraction of the total stake.


\section{Proof of Theorem \ref{fmt} } \label{proofoffmt} 

\subsection{The approach}

The rough idea is to adapt the bivalency proof by Aguilera and
Toueg~\cite{aguilera1999simple} of the result by Dolev and
Strong~\cite{dolev1983authenticated} that $f$-resilient consensus
requires $f+1$ timeslots in the permissioned setting. We use this
approach to produce a proof by contradiction: We assume a
deterministic protocol solving BA in the fully permissionless setting
and then show that there must be executions in which honest players
never output.  Taking this route, rather than adapting the original
proof of Dolev and Strong~\cite{dolev1983authenticated}, leads to a
proof that is somewhat simpler than the one given in the conference
version of this paper~\cite{lewis2021byzantine}.  Still, the argument
must accommodate several new features of our framework, including
message dissemination (as opposed to point-to-point communication),
failures other than crash failures, and the presence of general
oracles.


\subsection{Some conventions} 

For convenience, we assume that an active honest player disseminates exactly
one message per timeslot. Because we impose no bounds on message
length, an honest player never needs to send more than one
message in a timeslot.  Forcing honest players to send at least one
message per timeslot (as opposed to sometimes remaining silent) is
also easily seen to be without loss of generality; see Aguilera and
Toueg~\cite{aguilera1999simple} for details.
%
%

We can assume that a protocol uses one or more permitter oracles;
otherwise, there are no constraints in the fully permissionless
setting that prevent the adversary from deeming all honest players
inactive at all timeslots.  (Recall from Section~\ref{ss:external}
that the assumptions there on external resources guarantee at least
one honest active player at every timeslot.)
Because the fully permissionless setting imposes no constraints on the
activity periods of players with a non-zero amount of stake, however,
we are forced to confront the case in which none of the players included in
the initial stake distribution are ever active.


Finally, we can assume an initial stake distribution that allocates 0
stake to the Byzantine players (making the impossibility result as strong as
possible). We need to consider only the ``null'' environment that never
sends any transactions. 

%

\subsection{The setup} 

Fix $\rho>0$. Towards a contradiction, suppose that the 
$\rho$-resilient and deterministic protocol  $(\Sigma,\mathcal{O},\mathcal{C})$   solves
BA in the  fully permissionless and authenticated setting with
synchronous  communication. 

We consider an infinite set of players
$\mathcal{P}$, but, as per the assumptions described in Section
\ref{pswrr},  we suppose a finite number of players are active at each
timeslot.  The infinite set of players is used only to ensure that
there are always new players that can join at any timeslot, i.e., that
the set of players is not ultimately fixed. \\

\noindent \textbf{The rough idea}. The idea is to
restrict attention to executions in which a fresh set of players
become active at each timeslot $t$, and contribute briefly to the
protocol before becoming inactive again. So long as the number of
fresh players becoming active at $t$ is large, and so long as all
resources are allocated evenly between them and at most one of those
players is Byzantine, all resource allocations will be
$\rho$-bounded. The only deviations by Byzantine players
that are relevant for our argument are to crash or delay the 
dissemination of a message. \\

\noindent \textbf{The details}.  Choose
$s \geq \text{max} \{ 1/\rho, 3 \}$ and partition $\mathcal{P}$ into $s$
infinite and disjoint sets $P_1,\dots, P_s$.  We will be able to
derive a contradiction even after restricting attention to executions
that satisfy the conditions $(\Pi_0)-(\Pi_{10})$ below.
  
\vspace{.1in}

 \noindent First, some conditions restricting the inputs:  
 \begin{itemize} 
 \item[($\Pi_0$)]  All players in a given $P_k$ receive the same protocol input, either 0 or 1. 
 \item[($\Pi_1$)]  All players receive parameter inputs $d=\infty$ and $\Delta=2$, as well as their personal value $\mathtt{id}(p)$. 
 \end{itemize} 
 
  \noindent  Next, we specify that $s$ new players become active at each timeslot: 
 
  \begin{itemize} 
  \item[($\Pi_2$)] Each player $p$ is active at timeslots in an interval $[en_p,ex_p]$.  It may be that $en_p=ex_p=\infty$. 
 \item[($\Pi_3$)]  For each $t$, there exists precisely one player $p$
   in each $P_k$ with $en_p=t$. Let $P^{\ast}_t$ denote the set of all
   players $p$ such that $en_p=t$; thus, $|P^{\ast}_t|=s$ for each~$t$.   
  \item[($\Pi_4$)]    We consider a single fixed sequence
    $(P^{\ast}_t)_{t\geq 1}$, which defines the specific choice of
    players (from the $P_i$'s) that belong to each set~$P^*_t$.
This sequence is the same in all executions considered. 
   \end{itemize}

 \noindent  The next conditions specify the various values $ex_p$,  when disseminations may be received, and the possible behaviors of Byzantine players: 
 \begin{itemize} 
  
 \item[($\Pi_5$)] If $p$ is honest,  then $ex_p=en_p$.
 \item[($\Pi_6$)] If $p \in P^*_t$ is honest and disseminates $m$ at $t$, then all players active at $t+1$ receive that dissemination at $t+1$. 
 
  \item[($\Pi_7$)] Byzantine players can display only two forms of
     non-honest behavior: 
     
     \begin{enumerate} 
     \item[(1)] If $p\in P^{\ast}_t$ is Byzantine, it
     may crash
immediately  before disseminating a  message at $ t$ (but after
     receiving messages delivered at~$t$, sending oracle queries, and
     receiving the oracle responses delivered at~$t$).\footnote{Formally, a
  player that has ``crashed'' never again disseminates any messages or
  issues any oracle queries.}  
In this case, $ex_p=t$.  
     \item[(2)]  If $p\in P^{\ast}_t$ is Byzantine, then it may not
       crash but may fail to disseminate the instructed message $m$
       at $t$. We let $ex_p$ be the first timeslot at which $p$
       disseminates $m$ (meaning $ex_p=\infty$ is possible and that
       $ex_p=t$ if $p$ disseminates $m$ at $t$), and stipulate that
       $p$ performs no actions at stages $>t$ other than (possibly)
       disseminating $m$ at some first timeslot.  If $ex_p>t$, then we
       say that $p$ suffers dissemination failure at each $t'\in [t,ex_p]$. 
\end{enumerate} 
 
\item[($\Pi_8$)]  We say $p$ \emph{suffers delays at} $t$ if either: 
\begin{itemize} 
\item  $p$ crashes at $t$; 
\item  $p$ suffers dissemination failure at $t$, or;
\item  $p$ disseminates a message $m$ at $t$ which is not received by some active player at $t+1$.
\end{itemize} 

We stipulate that at most one (Byzantine) player can suffer delays at each timeslot.

     \end{itemize} 
     
   \noindent    Finally, we stipulate conditions that ensure resource allocations are $\rho$-bounded: 
     
   \begin{itemize}    
 \item[($\Pi_9$)]  At most one player in each $P^{\ast}_t$ is Byzantine. 
 \item[($\Pi_{10}$)] For each permitter oracle~$O \in \mathcal{O}$,
   the corresponding resource allocation $\mathcal{R}^O$ satisfies the
   following for
   every $p$ and $t$: $\mathcal{R}^O(p,t)=1$ if $p\in P^{\ast}_t$ and
   $\mathcal{R}^O(p,t)=0$ otherwise. 
(Thus, no matter how many timeslots a player might be active, it
possesses a non-zero amount of resource in at most one
timeslot.)
Together with the conditions
   above, this condition ensures that all resource allocations are
   $\rho$-bounded.
 \end{itemize} 
 
 \noindent Let $\Pi$ denote the set of all executions of the protocol
  satisfying these
 conditions.

\subsection{Bivalent $t$-runs and indistinguishable executions} 

\noindent \textbf{Bivalent $t$-runs}. Roughly, by a ``$t$-run'' we
mean a specification of the first $t$ timeslots of an execution (as
defined in Section~\ref{ss:st}). More precisely: 

By a {\em $0$-run}, we mean a specification of the inputs for each player. 

For $t\geq 1$,  by a {\em $t$-run}, we mean a specification of:
 \begin{enumerate} 
 \item[(i)]  The inputs given to each player. 
 \item[(ii)]  The messages disseminated by each player at each timeslot $\leq t$. 
 \item[(iii)] The sequence of queries sent and responses received by each
   player at each timeslot $\leq t$. 
   \item[(iv)] The players that crash at each timeslot $\leq t$. 
 \item[(v)]  The set of disseminations received by each player at each timeslot $\leq t+1$. 
 
\end{enumerate} 
 
For $t'>t$, we say that the $t'$-run $x'$ \emph{extends} the $t$-run
$x$ if: (a) Players receive the same inputs in both runs;  (b)
Players disseminate the same messages at all timeslots $\leq t$ in
both runs; (c) Players send and receive the same sequence of oracle queries and
responses at timeslots $\leq t$ in both runs; (d) The same players crash at all timeslots $\leq t$ in
both runs, and; (e) Players receive the same
disseminations at all timeslots $\leq t+1$ in both runs.  
Extending this definition in the obvious way, we can also speak about 
an execution extending a $t$-run.

For $i\in \{ 0,1 \}$, we say a $t$-run $x$ is $i$-\textit{valent} if,
for every execution extending $x$, there exists $t^{\ast}$ such that
all honest players output $i$ if they are active at any timeslot
$\geq t^{\ast}$.  We say $x$ is \textit{univalent} if it is $i$-valent
for some $i\in \{ 0, 1 \}$.  We say $x$ is \textit{bivalent} if it is
not univalent.

For each $t$, we let $X_t$ denote the set of all $t$-runs extended by
some execution in $\Pi$. \\

\noindent \textbf{Indistinguishable executions}. In the deterministic
model, we say that two executions $\mathtt{E}_0$ and $\mathtt{E}_1$ are \emph{indistinguishable} for honest $p$ if: 
\begin{itemize} 
\item Player $p$ has the same inputs in $\mathtt{E}_0$ and $\mathtt{E}_1$, and;  
\item At each timeslot $t$, player $p$ receives the same multiset of messages
  and the same sequence of oracle responses in~$\mathtt{E}_0$ as it does in
  $\mathtt{E}_1$.
\end{itemize}  
By induction on timeslots, the behavior (and eventual output) of a
player must be identical in any two executions that are
indistinguishable for it.

\subsection{Defining an execution in which honest players never output}

\noindent \textbf{There exists a bivalent $0$-run in $X_0$}. The proof
of this fact is essentially the same as the proof of the corresponding
lemma in the proof by Aguilera and Toueg~\cite{aguilera1999simple} (which, in turn, is
the same as in the proof of the FLP
Theorem~\cite{fischer1985impossibility}). Towards a contradiction,
suppose all $0$-runs in $X_0$ are univalent. For each $k$ with
$0\leq k \leq s$, let $x_k$ be the $0$-run in which players in each
$P_i$ for $1\leq i\leq k$ all get input 1, while all other players get
input 0. Because $x_0$ is 0-valent and $x_s$ is 1-valent (by Validity and Termination),
there must exist a smallest $k$ for which $x_k$ is 0-valent, while
$x_{k+1}$ is 1-valent.

For each $i\in \{ 0,1 \}$, let $\mathtt{E}_i$ be the (unique)
execution in $\Pi$ that extends $x_{k+i}$ and such that all players
are honest except the players in $P_{k+1}$, who all crash before
sending messages if they ever become active.
There exist such executions $\mathtt{E}_0$ and
$\mathtt{E}_1$ in $\Pi$ because, if only the players in $P_{k+1}$ are
Byzantine, then at most one player in each $P^{\ast}_t$ is Byzantine.
 
Because $x_k$ is 0-valent, all honest players that output must output
0 in $\mathtt{E}_0$.  Because~$x_{k+1}$ is 1-valent, all honest
players that output must output 1 in $\mathtt{E}_1$. This contradicts
the fact that $\mathtt{E}_0$ and $\mathtt{E}_1$ are indistinguishable
for players outside $P_{k+1}$.  This last statement holds because the
protocol inputs for players in $P_{k+1}$ 
are undetermined for the oracles and for players outside $P_{k+1}$,
and because oracles are non-adaptive (see
Section~\ref{ss:st}).\footnote{The players of~$P_{k+1}$ may send
  different oracle queries in $\mathtt{E}_0$ and $\mathtt{E}_1$
  immediately before crashing, but because these players send no
  messages after these oracle queries, the executions remain
  indistinguishable to all other players.}\\

\noindent \textbf{Extending the bivalent run}. Suppose $x\in X_t$  is bivalent. We show that there exists $x'\in X_{t+1}$ such that $x'$ extends $x$ and is bivalent.  Combined with the existence of a bivalent 0-run in $X_0$, this suffices to establish the existence of an execution in which honest players never output. 

Towards a contradiction, suppose all extensions of $x$ in $X_{t+1}$
are univalent. Let $x^{\ast}$ be the (unique) extension of $x$ in
$X_{t+1}$ in which all players in $P^*_{t+1}$ are honest, and in which
any Byzantine player still active at $t+1$ disseminates a message
which is received by all active players at the next timeslot. (Such a
player must be suffering delays at~$t+1$, so there is at most one such
player.)

We're done if $x^*$ is bivalent so assume, without loss of generality,
that $x^{\ast}$ is 1-valent. Because $x$ is bivalent, there must
exist $x^{\dagger}\in X_{t+1}$ which is an extension of $x$ and which
is not 1-valent. We're done if $x^{\dagger}$ is bivalent, so assume
that it is 0-valent.  From the definition of $\Pi$, it follows that there
exists a single player $p$ that is Byzantine, and
which suffers delays at $t+1$ in $x^{\dagger}$. There
are now two possibilities to consider, depending on whether these
delays are caused by the communication network or by the player itself.\\

\noindent \textbf{Case 1}: In $x^{\dagger}$, player $p$ does
disseminate $m$ at $t+1$, but there exists some non-empty subset
$\{ p_1,\dots,p_r \} $ of the players in $P^{\ast}_{t+2}$ who do not
receive the dissemination at $t+2$. In this case, define
$x_0:=x^{\dagger}$, $x_r:=x^{\ast}$, and, for each $k\in [1,r)$,
define $x_{k}$ to be the $t+1$-run that is identical to $x^{\dagger}$
except that players $p_1,\dots,p_k$ receive~$m$ at $t+2$.  Assume that every such~$x_k$ is univalent (otherwise, we
are done).

Because $x_0$ is 0-valent and $x_r$ is 1-valent, there must exist a
least $k$ such that $x_k$ is 0-valent and~$x_{k+1}$ is 1-valent. For
$i\in \{ 0,1 \}$, let $\mathtt{E}_i$ be the execution extending
$x_{k+i}$ in which players in  $P^{\ast}_{t'}$ are honest for
all $t'\geq t+2$ save $p_{k+1}$, who crashes 
immediately before sending messages
at timeslot~$t+2$. 
  Because~$x_k $ is 0-valent, all honest players that
output must output 0 in $\mathtt{E}_0$.  Because~$x_{k+1}$ is 1-valent,
all honest players that output must output 1 in $\mathtt{E}_1$. This
contradicts the fact that $\mathtt{E}_0$ and~$\mathtt{E}_1$ are
indistinguishable for honest players. \\

%

\noindent \textbf{Case 2}: In $x^{\dagger}$, player $p$ does not
disseminate $m$ at $t+1$, either because it crashes at this timeslot
or because it delays message dissemination.  We can assume the latter
scenario.  For if~$p$ crashes at~$t+1$, consider instead the
$t+1$-run~$x^{\diamond}$ which is identical to $x^{\dagger}$ except
that $p$ does not crash and does not disseminate a message at $t+1$.
If $x^{\diamond}$ is bivalent, we are done.  It cannot be 1-valent, as
there is an execution extending it (with $p$ never disseminating and
with all other players active at timeslots $>t+1$ being honest) that
is indistinguishable for all honest players from one
extending~$x^{\dagger}$ (which is 0-valent).  Thus $x^{\diamond}$,
like $x^{\dagger}$, is a 0-valent extension of~$x$ in~$X_{t+1}$, and
it can substitute for~$x^{\dagger}$ in the rest of the argument.


Assuming now that~$p$ delays the dissemination of~$m$
in~$x^{\dagger}$, let $y_0$ be the $t+2$-run in $X_{t+2}$ that extends
$x^{\dagger}$ and such that all players in $P^{\ast}_{t+2}$ are
honest, player $p$ disseminates $m$ at $t+2$, and $m$ is received by
all active players at the next timeslot.\footnote{The
  $t+2$-run~$y_0 \in X_{t+2}$ is well defined, as in the $t+1$-run
  $x^{\dagger}$, $p$ receives all the messages and oracle responses
  needed to disseminate the message~$m$ at~$t+1$ (which, in~$y_0$, $p$
  does at~$t+2$).}  Because $y_0$ extends $x^{\dagger}$, it is
0-valent. Let~$x_0$ be the $t+2$-run in $X_{t+2}$ that is identical to
$y_0$, except that $p$ disseminates $m$ at $t+1$
, while it remains the case that no players in $P^{\ast}_{t+2}$
receive that dissemination at $t+2$. Then $x_0$ must also be 0-valent,
because for each execution extending either of $y_0$ or $x_0$, there
is an execution extending the other that is indistinguishable from the
first execution for all honest players.  (Intuitively, honest players
cannot know whether $m$ was disseminated at $t+1$ and took two
timeslots to be delivered, or was disseminated at $t+2$ and took a single
timeslot to be delivered.)

Let $p_1,\dots,p_s$ be an enumeration of the players in
$P^{\ast}_{t+2}$. For each $k\in [1,s]$, let $x_k$ be the $t+2$-run
that is identical to $x_0$, except that all players $p_i$ for
$1\leq i\leq k$ receive $m$ at $t+2$. Because $x_s$ extends~$x^{\ast}$,
it must be 1-valent. The rest of the argument is identical to that in
Case~1, with~$x_s$ playing the previous role of~$x_r$ and the player
set~$\{p_1,\ldots,p_s\}$ playing the previous role
of~$\{p_1,\ldots,p_r\}$.


\subsection{Discussion} 

The proof of Theorem~\ref{fmt} is driven by two properties of the
fully permissionless setting: (i) there is no bound on the number of
players; (ii) the set of players that are active can change arbitrarily at each timeslot (subject to the constraint that the adversary be $\rho$-bounded). 
Theorem~\ref{fmt} is trivially false if all players are always active, 
as otherwise the players allocated non-zero stake by the initial stake
distribution could carry out a standard permissioned Byzantine
Agreement protocol.  We leave as an open question whether
Theorem~\ref{fmt} continues to hold when there is a determined and
finite upper bound on the number of players (but an arbitrary and
undetermined player allocation).  We also leave open whether
Theorem~\ref{fmt} holds for protocols that use only external resources
(i.e, that have no initial stake distribution) in the setting in which
all players are always active.


The proof of Theorem~\ref{fmt} does not require sybils, in the sense
that it holds even when $|\mathtt{id}(p)|=1$ for all $p$.  The
positive result of Khanchandani and
Wattenhofer~\cite{khanchandani2021byzantine} shows that
Theorem~\ref{fmt} no longer holds when the player set is finite,
identifiers are (undetermined but) unique (i.e., $|\mathtt{id}(p)|=1$
for all $p$), and all players are always active. This positive result
does not require an initial stake distribution.

Theorem~\ref{fmt} trivially fails to hold if $\rho=0$. To see this,
consider the following protocol that makes use of a single permitter
$O$.  Active players are instructed to send a single query $(b,i)$ to
the permitter at timeslot 1, where $b$ is their external resource
balance at that time and $i$ is their protocol input. 
The permitter ignores queries with~$b=0$---which, because~$\rho=0$,
includes any queries issued by faulty players---and ``certifies'' a
query~$(b,i)$ with~$b > 0$ by responding with a type-$O$ message with
a single entry~$i$. Players are instructed to immediately disseminate
the permitter's responses.
Players can then, at the first timeslot $\geq 1+\Delta$ at which they
are active, safely output~$i$ if all received type-$O$ messages
consist of the entry~$i$, and a default value otherwise.

\section{Proof of Theorem \ref{psm}}  \label{proofofpsm} 

\subsection{The intuition} 

Towards a contradiction, 
fix a 0-resilient protocol that solves probabilistic BA in the
dynamically available and partially synchronous setting with security
parameter~$\epsilon < \tfrac{1}{3}$.
Recall the definition of a protocol instance from Section
\ref{instance}. We consider protocol instances with $d=\infty$ and a
finite set of players, all of whom are honest.  Partition a
set~$\mathcal{P}$ of players into two non-empty sets $P_0$
and $P_1$; suppose players of $P_i$  receive
the protocol input $i$ (for $i \in \{0,1\}$).  We suppose further that
$\mathcal{R}^O(p,t)>0$ for every $p,t$ and for each permitter $O$.

We then consider a network partition during which messages
disseminated by players in each $P_i$ are received by other
players in $P_i$ but not by players in
$P_{1-i}$. If the network partition is sufficiently long then,
with significant probability, all players will eventually terminate
before the partition ends. Because the protocol is assumed to solve
probabilistic BA in the dynamically available setting, the Validity
condition forces players in $P_0$ to,
with significant probability, terminate and output~0.
The reason is
that, from their viewpoint, the protocol execution is
indistinguishable
from a synchronous protocol execution in which the active player set is
$P_0$.
Invoking the Validity condition again, players in~$P_1$
must, with significant probability, terminate and output~1; from their
viewpoint, the protocol execution is indistinguishable from a
synchronous protocol execution in which the active player set is~$P_1$.
This leads to a violation of the Agreement condition.

\subsection{The details}
Let $\mathcal{P}=P_0 \cup P_1$, where each
$P_i$ is non-empty and where $P_0\cap P_1=\emptyset$. We consider three different instances of
the protocol, $\mathtt{I}_0$, $\mathtt{I}_1$, and $\mathtt{I}_2$. In all three instances, the player set is $\mathcal{P}$. 

\vspace{0.1cm} 
\noindent \textbf{The inputs and the player allocations}. 
In all three instances, each player will either be active at every
timeslot or inactive at every time slot.  In  $\mathtt{I}_0$,
$\mathtt{I}_1$, and $\mathtt{I}_2$, the active players are
$\mathcal{P}$, $P_0$, and~$P_1$, respectively.
In all three instances, every player is honest 
and has a unique identifier,
and all players are given the determined parameter inputs $d=\infty$
and $\Delta$ (say, $\Delta=2$).  In $\mathtt{I}_0$, players in $P_i$
are given input $i$. In $\mathtt{I}_1$, all players are given input 0.
In $\mathtt{I}_2$, all players are given input 1.
In all three instances: 
\begin{itemize} 

\item Every resource allocation $\mathcal{R}^O$ satisfies the
  condition that $\mathcal{R}^O(p,t)=1$ for all $p$ active at $t$.

\item   We consider an initial stake distribution $S_0$ which
allocates all players in $\mathcal{P}$ resource balance 1 (whether or
not the player is active) and an environment that does not issue any transactions.

\end{itemize} 
%

\vspace{0.1cm} 
\noindent \textbf{The timing rules}. The timing rules 
in the three instances are as follows: 
\begin{itemize} 

\item In $\mathtt{I}_0$, a dissemination at $t$  by any player in
  $P_i$ is received by other players in $P_i$ at
  $t+1$ and is received by players in $P_{1-i}$ at timeslot
  $t'=\text{max} \{ t+1, \text{GST} \}$.

\item In $\mathtt{I}_{1+i}$, a dissemination at $t$ by any player in
  $P_i$ is received by other players in $P_i$ at
  $t+1$.

\end{itemize} 

%

\vspace{0.2cm} 
\noindent \textbf{Producing a probable violation of Agreement}.  It
follows by induction on timeslots $<$GST that:
\begin{enumerate} 

\item[$(\dagger)$] For each $i\in \{ 0,1 \}$ and all $t <$GST,
  $\mathtt{I}_0$ and $\mathtt{I}_{1+i}$ are indistinguishable for
  players in $P_i$ until $t$.

\end{enumerate} 
Because the given protocol is assumed to solve BA (and, in particular,
satisfy Validity) with probability
at least $1-\epsilon$ in the dynamically available 
setting with a 0-bounded adversary, there is a timeslot~$t^{\ast}$ by
which: (i) in $\texttt{I}_1$, with probability at least $1-\epsilon$,
all players of~$P_0$ terminate and output~0; (ii) in $\texttt{I}_2$,
with probability at least $1-\epsilon$, all players of~$P_1$
terminate and output~1.

Using $(\dagger)$, in instance~$\mathtt{I}_0$, if GST is chosen to be
sufficiently large (greater than~$t^{\ast}$), then with probability
at least $1-2\epsilon$, by timeslot~$t^*$ all players
of~$P_i$ have terminated with output~$i$ (for
$i \in \{0,1\}$).  Because $\epsilon < \tfrac{1}{3}$, this likelihood
of a disagreement contradicts the initial assumption that the
protocol satisfies Agreement with security parameter $\epsilon$ in the
partially synchronous setting.

\vspace{0.2cm} 
\noindent \textbf{Discussion}. The proof of Theorem~\ref{psm} 
is driven solely by the possibility of inactive players; this
assumption is necessary, as
otherwise the set of players allocated non-zero stake by the initial
stake distribution could carry out a standard permissioned Byzantine
Agreement protocol.  
The player set~$\mathcal{P}$ used in the proof can be known and
finite, with each player in possession of only a single identifier.

\section{Proof of Theorem \ref{ortheorem}} \label{proofofortheorem}

The proof resembles that of Theorem~\ref{psm}.
Towards a contradiction, fix a 0-resilient blockchain protocol that is
optimistically responsive in the dynamically available,
authenticated, and synchronous setting, with security
parameter~$\epsilon < \tfrac{1}{3}$. 
Suppose the protocol satisfies the definition of optimistic
responsiveness with liveness parameter $\ell=O(\delta)$ and grace
period parameter $\Delta^*=O(\Delta)$, where~$\delta$ denotes the
(undetermined) realized maximum message delay.  Let $K$ be a constant
such that $\ell< K\delta$ for all~$\delta \ge 1$, and choose $\Delta>K$
(thereby fixing~$\Delta^*$).



\vspace{0.2cm} 
\noindent \textbf{The player set}.  Let $\mathcal{P}=P_0 \cup P_1$,
where the $P_i$'s are non-empty and disjoint sets.
We consider three instances of
the protocol, $\mathtt{I}_0$, $\mathtt{I}_1$,  and $\mathtt{I}_2$. In all three instances, the player set is $\mathcal{P}$. 

\vspace{0.1cm} 
\noindent \textbf{The inputs and the player allocations}.  In all
three instances, each player will either be active in all timeslots or
active for all and only the timeslots $\le \Delta^*$.
The  players active at timeslots $>\Delta^*$ are 
$P_0$ in $\mathtt{I}_0$,
$P_1$ in $\mathtt{I}_1$, and
$\mathcal{P}$ in $\mathtt{I}_2$.
In all three instances, every player is honest 
and has a unique identifier,
and all players are given the determined parameter inputs
$\Delta$
and $d=\Delta^* +2\Delta$.  
In all three instances: 
\begin{itemize} 

\item Every resource allocation $\mathcal{R}^O$ satisfies the
  condition that $\mathcal{R}^O(p,t)=1$ for all $p$ active at $t$.

\item   We consider an initial stake distribution $S_0$ which
allocates all players in $\mathcal{P}$ resource balance~1 (whether or
not the player is active). 

\end{itemize} 
Let $\mathtt{tr}_0$ and $\mathtt{tr}_1$ denote two transactions that
are conflicting relative to $S_0$.  In $\mathtt{I}_0$, the environment
sends $\mathtt{tr}_0$ to all players in $P_0$ at timeslot
$\Delta^*+1$.  In $\mathtt{I}_1$, the environment sends
$\mathtt{tr}_1$ to all players in~$P_1$ at timeslot $\Delta^*+1$.  In
$\mathtt{I}_2$, the environment sends $\mathtt{tr}_0$ to all players
in $P_0$ at timeslot $\Delta^*+1$ and sends $\mathtt{tr}_1$ to all
players in $P_1$ at timeslot $\Delta^*+1$.

%

\vspace{0.1cm} 
\noindent \textbf{The timing rules}. The timing rules 
in the three instances are as follows: 
\begin{itemize} 

\item In all three instances, a dissemination by any player at a
  timeslot $\leq \Delta^*$ is received by all other players at the
  next timeslot.

\item In $\mathtt{I}_0$, a dissemination at timeslot $t> \Delta^*$ by
  a player in~$P_0$ is received by the other players in~$P_0$
  at $t+1$. (By this point, only the players of~$P_0$ are active in this
  instance.) 
  
\item In $\mathtt{I}_1$, a dissemination at timeslot $t> \Delta^*$ by
  a player in~$P_1$ is received by the other players in~$P_1$
  at $t+1$. (By this point, only the players of~$P_1$ are active in this
  instance.) 

\item In $\mathtt{I}_{2}$, a dissemination at timeslot $t> \Delta^*$
  by any player in $P_0$ is received by other players in $P_0$ at
  $t+1$ and is received by players in $P_1$ at timeslot $t+\Delta$.  A
  dissemination at timeslot $t> \Delta^*$ by any player in $P_1$ is
  received by other players in $P_1$ at $t+1$ and is received by
  players in $P_0$ at timeslot $t+\Delta$.
  
\end{itemize} 

%

\vspace{0.2cm} 
\noindent \textbf{Producing a probable consistency violation}.  It
follows by induction on timeslots $\leq \Delta^*+ \Delta$ that:
\begin{enumerate} 

\item[$(\dagger)$] For each $i\in \{ 0,1 \}$ and all
  $t \leq \Delta^* + \Delta$, $\mathtt{I}_i$ and $\mathtt{I}_{2}$ are
  indistinguishable for players in $P_i$ until the end of timeslot
  $t$.

\end{enumerate} 
Because~$\delta=1$ in the instances $\texttt{I}_0$ and $\texttt{I}_1$ 
and $\Delta > K$, the definition of optimistic responsiveness implies
that: (i) in $\texttt{I}_0$, with
probability at least $1-\epsilon$, $\mathtt{tr}_0$ will be confirmed
for all players in~$P_0$ by timeslot $\Delta^*+ \Delta$; (ii) in
$\texttt{I}_1$, with probability at least $1-\epsilon$,
$\mathtt{tr}_1$ will be confirmed for all players in~$P_1$ by timeslot
$\Delta^*+ \Delta$.

From $(\dagger)$, it follows that, with probability at least
$1-2\epsilon$ in $\mathtt{I}_2$, $\mathtt{tr}_0$ will be confirmed for
all players in $P_0$ and $\mathtt{tr}_1$ will be confirmed for all
players in $P_1$ by timeslot $\Delta^*+ \Delta$.  
Because $\mathtt{tr}_0$ and~$\mathtt{tr}_1$ are conflicting
transactions and $\epsilon < \tfrac{1}{3}$, this provides the required
contradiction.

\vspace{0.2cm} 
\noindent \textbf{Discussion}. 
The proof of Theorem \ref{ortheorem}, like that of Theorem~\ref{psm},
is driven solely by the possibility of inactive players; this
assumption is necessary, as
otherwise the set of players allocated non-zero stake by the initial
distribution could carry out a standard permissioned  protocol.  
The player set~$\mathcal{P}$ used in the proof can be known and
finite, with each player using a unique and known identifier.

\section{Proofs of Theorems \ref{posPoS} and \ref{posPoS2}} \label{thirteen} 
  
We present a proof of Theorem \ref{posPoS2}, as this is the stronger
result. Readers interested only in Theorem~\ref{posPoS} can omit all
aspects of the proof that are concerned with ensuring that, if $p$ is
honest and $\mathtt{T}$ is the set of transactions confirmed for $p$
at $t$, then $\mathtt{S}(S_0,\mathtt{T})$ allocates the adversary at
most a $\rho$ fraction of the stake (as this condition is
automatically guaranteed by the hypotheses of Theorem~\ref{posPoS}).

\subsection{Preliminaries}
  
We prove Theorem \ref{posPoS2} by modifying the (permissioned)
HotStuff protocol~\cite{yin2019hotstuff} to work as a PoS protocol.
While HotStuff is designed to achieve optimistic responsiveness and
simultaneously minimize message complexity, our concerns here are
different: For the sake of simplicity, we ignore issues of efficiency
and strive for the simplest protocol with the required properties.
Section~\ref{effic} describes some ways to improve the efficiency
of our protocol.  

\vspace{0.2cm} 
\noindent \textbf{Removing the `baked in' knowledge of time}.  The
definition of a state diagram in Section~\ref{ss:st} allows a player's
state transition at a given timeslot to depend on the number of that
timeslot. This modeling choice serves to make our impossibility
results as strong as possible.  For the present proof, however, such
knowledge arguably goes against the spirit of the setting.  We
therefore weaken the model so that a player's state transition at a
given timeslot can depend only on the player's current state and the
oracle responses and messages received by that time.

\vspace{0.2cm} 
\noindent \textbf{Allowing clock drift}.  Section~\ref{ss:inputs}
defined a version of the partial synchrony model in which message
delivery may take an arbitrary amount of time prior to GST, but in
which players have ``internal clocks'' of identical speed. That is,
active players proceed in lockstep, carrying out one round of
instructions according to their state diagram at each timeslot. This
definition is more generous to a protocol than other versions of the
partial synchrony model (e.g., \cite{DLS88}). This modeling choice
only strengthens our impossibility results in the partially
synchronous setting, Theorems~\ref{psm} and \ref{niceneg}. To make our
positive result here as strong as possible, however, we extend the
model to allow for bounded clock drift. 

Formally, the player allocation (see Section~\ref{ss:inputs}) will now
specify, for each player and timeslot at which that player is active,
whether the player is \emph{waiting} or \emph{not
  waiting} at that timeslot.
(Intuitively, the slower a player's
clock, the more frequently they will be waiting.)  If honest $p$ is
active and not waiting at timeslot $t$, it carries out the protocol
instructions as specified in Section~\ref{ss:st}. If $p$ is active but
waiting at $t$, it acts as if inactive at $t$ (and does not receive
any messages or oracle responses at that timeslot). We assume that
there is some known bound $\kappa \in (0,1]$ such that, for each
$t,\ell \in \mathbb{N}_{\geq 1}$ and each honest player $p$ active in
the interval $(t, t+\ell]$, there are at least
$\lfloor \kappa \cdot \ell \rfloor$ many timeslots in that interval at
which $p$ is active and not waiting (with $p$ active but waiting at
the other timeslots in $(t, t+\ell]$).  Intuitively, the
closer~$\kappa$ is to~1, the smaller the difference in the rates at
which the internal clocks of honest players advance.

In the version of the partial synchrony model considered in this
section, we assume that if $p$ disseminates $m$ at $t$, and if
$p'\neq p$ is active and not waiting at
$t'\geq \text{max} \{ \text{GST}, t+\Delta \}$, then $p'$ receives
that dissemination at a timeslot $\leq t'$.

\subsection{Recalling (permissioned) HotStuff}  \label{ph} 

To segue into the description of our PoS-HotStuff protocol in
Sections~\ref{overv}--\ref{tfs}, we first provide a rough description
of a simplified specification of the (permissioned) HotStuff protocol
that makes no effort to minimize communication complexity.
Accordingly, in this subsection, we consider a set of $n$ players, of
which at most $f$ may be Byzantine, with $n\geq 3f+1$.

\vspace{0.2cm} 
\noindent \textbf{Views in HotStuff}.   
The protocol instructions are divided into \emph{views}. Within each
view, there are three stages of voting, each of which is an
opportunity for players to vote on a \emph{block} 
of transactions proposed by the \emph{leader} of that view.  (All
these notions will subsequently be formalized.)  Each block of
transactions $B$ is associated with a specific view, denoted
$v^{\ast}(B)$.  If the first stage of voting in a view establishes a
\emph{stage 1 quorum certificate} (QC) for a block $B$ proposed in
that view (in the form of~$2f+1$ votes for~$B$, signed by distinct
players), then the second stage of the view may establish a
\emph{stage 2} QC for $B$, whereupon the third stage of the view may
establish a \emph{stage 3} QC for $B$.  If $Q$ is a QC for $B$, then
we define $v^*(Q):=v^*(B)$ and $B^*(Q):=B$.  Deviating somewhat from
the exposition in~\cite{yin2019hotstuff}, each honest player maintains
two local variables $Q^1$ and $Q^2$ (persistent across views) that
implement a ``locking'' mechanism; for~$i=1,2$, $Q^i$ records the
stage $i$ QC that triggered the player's most recent stage $i+1$ vote.

The instructions within a view are deterministic.  If the
leaders of views are also chosen determinis\-tically---for example, with
player~$p_i$ the leader of view~$v$ whenever
$i= v \text{ mod }n$---the entire protocol is deterministic.

\vspace{0.2cm} 
\noindent \textbf{View changes}. Given that we ignore communication
complexity, moving from one view to the next is straightforward in the
permissioned setting.  
When an honest player $p$ enters a view $v$, it
sets a timer.
When the timer expires, or if $p$ sees a stage 3 QC for
a block proposed by the leader of view~$v$, $p$ disseminates a
(signed) $\mathtt{view}\ v+1$ message $(v+1,Q^1)$.  A \emph{view
  certificate} (VC) for view $v+1$ is a set of $2f+1$
$\mathtt{view}\ v+1$ messages, signed by distinct players. Player $p$
enters view $v+1$ upon receiving a VC for view $v+1$, if presently in
a lower view.

\vspace{0.2cm} 
\noindent \textbf{Block proposals in HotStuff}.  When an honest player
enters a view $v$ for which it is the leader, it chooses, among all
QCs received to-date as part of a $\mathtt{view}$~$v$ message, the
most recent such QC~$Q$ (i.e., with~$v^*(Q)$ as large as possible).
The leader then disseminates a block proposal $(B,Q)$ in which~$B$
lists~$B^*(Q)$ as its parent block.  An honest player~$p$ that
receives such a block proposal while in view~$v$ will regard it as
\emph{valid} if and only if: (i) $Q$ is a QC for a block $B^*(Q)$ that
is the parent of $B$; and (ii) $v^*(Q^2)\leq v^*(Q)$, where~$Q^2$
denotes the current value of $p$'s corresponding local variable. Upon
receiving such a valid proposal, $p$ will disseminate a stage 1 vote
for $B$. Upon seeing a stage 1 QC for $B$ (a set of $2f+1$ stage 1
votes signed by distinct players) while in view $v$, $p$ will
disseminate a stage 2 vote for $B$ (and reset its local $Q^1$
variable, after which~$B^*(Q^1)=B$ and $v^*(Q^1)=v$). Upon seeing a
stage 2 QC for $B$ (a set of $2f+1$ stage 2 votes signed by distinct
players) while in view~$v$, $p$ will disseminate a stage 3 vote for
$B$ (and reset its local $Q^2$ variable, after which $B^*(Q^2)=B$ and
$v^*(Q^2)=v$). Upon seeing a stage 3 QC for any block, an honest
player considers that block and all its ancestors \emph{confirmed}.

\vspace{0.2cm} 
\noindent \textbf{Consistency and liveness for HotStuff}.  We now
argue (briefly and somewhat informally) that this protocol satisfies
consistency and liveness. For the former, call two distinct blocks
{\em incompatible} if neither is an ancestor of the other.  It
suffices to argue that, once a stage 3 QC~$Q$ is created for a
block~$B$ in a view~$v$---a prerequisite for confirmation---no QCs
will be created for any block incompatible with~$B$ at any view
$v' \ge v$.  Let~$H$ denote the honest players that contributed (stage
3) votes to~$Q$ (updating at that time their local~$Q^2$-values to a
view-$v$ stage 2 QC); because $n \ge 3f+1$ and a QC requires votes
from at least~$2f+1$ distinct players, $|H| \ge f+1$. Any QC thus
requires the participation of at least one player from~$H$.  An honest
player votes at most once in each stage of each view, so at
view~$v'=v$, players of~$H$ will not vote for blocks other than~$B$
and no QCs for such blocks will be created.  At view $v'=v+1$, a valid
block proposal must either list~$B$ as the parent block (in which case
the proposed block is compatible with~$B$) or a block~$B''$ that is
supported by a QC from some view~$v'' < v$ (in which case no player
of~$H$ will vote for it). We conclude that, at view~$v'=v+1$, no QC
will be created for any block incompatible with~$B$. The same argument
applies inductively to views~$v'=v+2,v+3,\ldots$.


For liveness, consider a view~$v$ with an honest leader~$p$, with $v$
large enough that, at GST, every honest player is in some view earlier
than~$v$.  Let~$t$ denote the first timeslot at which some honest
player enters view~$v$. (The timeslot~$t$ must exist because no honest
player can enter a view~$v' > v$ prior to some honest player entering
view~$v$.)  Because $t \ge \text{GST}$, all honest players enter
view~$v$ by timeslot $t + \Dk$ and the
leader~$p$ will disseminate a block proposal~$(B,Q)$ by this timeslot.
But will honest players vote in support of this proposal?  Consider
such a player~$p'$ with locally defined variable~$Q^2$.  Because $p'$
set this value upon seeing a stage 2 QC for $B^*(Q^2)$ while in view
$v^*(Q^2)$, at least $f+1$ honest players must have seen a stage 1 QC
for $B^*(Q^2)$ while in view $v^*(Q^2)$, and must have set their local
value $Q^1$ so that $v^*(Q^1)=v^*(Q^2)$ while in this view. At least
one of the $2f+1$ $\mathtt{view}\ v$ messages in the VC for view $v$
received by the leader $p$ must be from one of these $f+1$ honest
players, meaning that $v^*(Q)\geq v^*(Q^2)$; therefore, $p'$ will
regard the block proposal~$(B,Q)$ made by~$p$ as valid.  Thus, so long
as timer expiries are set appropriately, stage 1, 2, and 3 QCs
will be produced for $B$ in view~$v$, resulting in the
block's confirmation.

  
  
\subsection{An informal overview of the required modifications} \label{overv} 
 
We define a protocol $(\Sigma,\mathcal{O},\mathcal{C},\mathcal{S})$
such that $\mathcal{O}$ is empty and $\mathcal{S}$ contains only the
stake allocation function.  Modifying the permissioned HotStuff
protocol to give a proof of Theorem \ref{posPoS2} involves three key
considerations: dealing with changing sets of players, 
determining a safe way to end epochs, defining a
deterministic protocol, and ensuring that the adversary remains
$\rho$-bounded.

\vspace{0.2cm} 
\noindent \textbf{Accommodating changing players}.  To deal with a
changing player set, we divide the instructions into \emph{epochs} and
have the player set change with each epoch. (Each epoch loosely
corresponds to a single execution of the permissioned HotStuff
protocol; due to asynchrony, however, different players may be in
different epochs at the same time.)  Each epoch consists of an
(unbounded) number of views 
and, for some $N$, is responsible for producing confirmed blocks with
heights in $(eN,eN+N)]$.
(While the view number will reset to~0 with every new epoch, block
heights will continue to increment throughout the execution of the
protocol; all of this will be made precise in Section~\ref{tfs}.)
To ensure that the player set changes only
once with each epoch, blocks are confirmed in batches, with all blocks
with heights in $(eN,eN+N)]$ confirmed at the same time (once blocks
of all such heights have been produced).
After such a batch confirmation, executions consistent with the
quasi-permissionless setting will see all honest players that control
a non-zero amount of stake (with respect to the newly confirmed
transactions) being active and carrying out the instructions for
epoch~$e+1$.

\vspace{0.2cm}
\noindent \textbf{How to complete an epoch}.  In the permissioned
HotStuff protocol, a block may be confirmed indirectly (without any
stage 3 QC ever being produced for it), by inheriting confirmation
from a descendant block for which stage 1, 2 and 3 QCs have been
received.  This aspect of the protocol's design poses a challenge to
safely updating the player set: if a potentially epoch-ending
block~$B$ at height~$eN+N$ receives only stage 1 and 2 QCs (say),
which identifiers should be allowed to propose and vote on descendants
of~$B$? If the protocol could be sure that~$B$ would eventually be
indirectly confirmed, it could use the stake allocation resulting from
the execution of the transactions in~$B$ and its ancestors. But
another block~$B'$ at height~$eN+N$ might well also receive stage 1
and 2 QCs, leading to conflicting opinions about
the next epoch's player set.

Our solution is to allow the player set for epoch $e$ to continue
proposing and voting on blocks of heights greater than $eN+N$ until
they produce a (directly or indirectly) confirmed block~$B^*$ of
height $eN+N$. Once they do so, any blocks produced in epoch~$e$ with
height $>eN+N$, together with the votes for those blocks, are kept as
part of the protocol's permanent history (as proof that the protocol
operated as intended); crucially, any transactions in these blocks are
treated by the protocol as if they had never been included in any
block.\footnote{Indeed, treating these transactions as confirmed could
  lead to consistency violations. The issue is that the player set for
  epoch~$e+1$ may not have been present to implement HotStuff's
  locking mechanism during epoch~$e$.  Concretely, suppose~$B'$ and
  $B''$ are both blocks of height $>eN+N$ produced during epoch $e$,
  with $B'$ the parent of $B''$.  Suppose both blocks receive stage 1,
  2, and 3 QCs; this can occur during asynchrony, as the QCs for $B''$
  may be produced before any player has seen all of the QCs for
  $B'$. If the protocol were to regard the transactions in $B'$ and
  $B''$ as confirmed, then some of the epoch-$(e+1)$ player set may
  initially see $B''$ as confirmed, even though most of that player
  set has only seen stage 1, 2, and~3 QCs for~$B'$.  The new player
  set might then confirm new blocks that are descendants of $B'$
  incompatible with $B''$, resulting in a consistency violation.}
The player set for epoch~$e+1$ treats~$B^*$ as a genesis block for the
epoch; transactions from the ``orphaned'' epoch-$e$ blocks (i.e.,
those with height $>eN+N$) can then be re-included in epoch-$(e+1)$
blocks.


\vspace{0.2cm} 
\noindent \textbf{Defining a deterministic protocol}.  What is~$N$? If
a randomized protocol were acceptable, each leader could be chosen
with probability proportional to its current stake, and the set of
active players could change with every confirmed block (i.e.,
$N=1$). To obtain a deterministic protocol, we take advantage of the
assumptions introduced in Section~\ref{envrho}: every transaction
preserves the total amount of stake in existence; and stake is
denominated in some indivisible unit (and so stake amounts can be
regarded as integers).  Then, to obtain the simplest-possible solution
that is deterministic and guarantees at least one honest leader
produces a confirmed block in each epoch following GST, we set~$N$
equal to total amount of stake:
$N:= \sum_{p\in \mathcal{P}}\sum_{id\in \mathtt{id}(p)}
\mathtt{S}(S_0,\emptyset,id)$.
This choice of~$N$ ensures that, in each epoch, every stake-holding
honest player has at least one opportunity as the leader of a view.

A more reasonable (and still deterministic) alternative would be a PoS
protocol in which, for some $N'$, participating players must lock up
in escrow some positive integer multiple of $1/N'$ of the total stake.
Then, $N'$ plays (roughly) the previous role of $N$ but can be
chosen to be much smaller. See Section~\ref{effic} for further details
on the changes to the basic PoS-HotStuff protocol that are required to
implement such an ``in escrow'' version.

\vspace{0.2cm} 
\noindent \textbf{Ensuring the adversary remains $\rho$-bounded}. As
highlighted by Theorem \ref{niceneg}, we must also be careful to
ensure that confirmed sets of transactions do not allocate more than a
$\rho$-fraction of the stake to the adversary. As we will see in
Section \ref{secprof}, the assumption that the environment is
$\rho$-bounded
is the key ingredient in an inductive
proof that this property always holds (and that, in fact, consistency
is never violated).

\subsection{PoS-HotStuff: The formal specification}  \label{tfs} 
  
As noted in Section \ref{ss:inputs}, it is convenient to assume that
whenever honest $p$ disseminates a message at some timeslot $t$, $p$
regards that dissemination as received (by $p$) at the next timeslot.
  
\vspace{0.2cm} 
\noindent \textbf{Signed messages}. Recall from
Section~\ref{ss:permitted} that, to model the use of a signature
scheme, we suppose that an entry of signed type must be of the form
$(id,m)$ such that $id$ is an identifier (and where $m$ is
arbitrary). In the \emph{authenticated} setting that we work in here,
an entry $e=(id,m)$ of signed type is permitted for $p$ at $t$ if
$id\in \mathtt{id}(p)$, or if $p$ has previously received a message
which contained the entry $e$ (in which case~$p$ is free to repeat it
to others).
  
\vspace{0.2cm} 
\noindent \textbf{Transaction gossiping and the variable
  $\mathtt{T}^{\ast}$}. We assume that, whenever any honest player
first receives a transaction $\mathtt{tr}$, from the environment or
another player, they disseminate $\mathtt{tr}$ at the same
timeslot. Each honest player $p$ maintains a variable
$\mathtt{T}^{\ast}$ which is the set of all transactions received by
$p$ thus far.
  
  
\vspace{0.2cm} 
\noindent \textbf{Blocks}. A block (other than the genesis block) is a
tuple $B=(id, (h,e,v,\mathtt{T},B',Q))$, where $h\in \mathbb{N}$
denotes the height of the block, $e$ denotes an epoch, $v$ denotes a
view in epoch $e$, $\mathtt{T}$ is a set of transactions, and $B'$ is
a block that is the \emph{parent} of $B$.  
The entry $Q$ is a
QC that will be used to verify the validity of $B$; the intention is
that~$Q$ will record a stage 1 QC for~$B'$.  When an honest player
receives $B$, it also regards the transactions in $\mathtt{T}$ as
having been received, and so disseminates those transactions if it has
not already done so. We set $id^*(B):=id$, $h^*(B):=h$, $e^*(B):=e$,
$v^*(B):=v$, $\mathtt{T}^*(B):=\mathtt{T}$, and $Q^*(B):=Q$. The
\emph{genesis block} $B_{g}:=(\emptyset)$ is regarded by honest
players as having been received at the beginning of the protocol
execution. The genesis block has no parent and has only itself as an
\emph{ancestor}, and we define $h^*(B_{g}):=0$, $e^*(B_{g}):=-1$,
$v^*(B_{g}):=0$, $\mathtt{T}^*(B_{g}):=\emptyset$, and
$Q^*(B_{g}):=\emptyset$; $id^*(B_g)$ is undefined. For all other
blocks $B$, $h^*(B)\geq 1$, $e^*(B) \ge 0$, and $v^*(B) \ge 0$;
otherwise honest players will regard $B$ as invalid.  The
\emph{ancestors} of such a block $B$ are $B$ and all ancestors of its
parent.  Two blocks are \emph{incompatible} if neither is an ancestor
of the other.
  
\vspace{0.2cm} 
\noindent \textbf{Votes}.  A vote is a tuple $V=(id,(c,B,s))$, where
$id$ is an identifier, $c\in \mathbb{N}$  represents an amount of
stake, $B$ is a block, and $s\in \{ 1,2,3 \}$ specifies the voting
stage within the view. We define $id^*(V):=id$, $c^*(V):=c$,
$B^*(V):=B$, and $s^*(V):=s$. A vote inherits the epoch, view, and
height of its block: 
$e^*(V):=e^*(B)$, $v^*(V):=v^*(B)$, and $h^*(V):=h^*(B)$.

\vspace{0.2cm} 
\noindent $\mathtt{T}$-\textbf{Quorum Certificates (QCs)}.  Because
quorum certificates (QCs) must be stake-weighted (to accommodate
sybils) and transactions can transfer stake from one identifier to
another, QCs must be defined relative to a set of transactions (for
us, the set of transactions confirmed in previous epochs).  For the
genesis block, by definition,
the empty set~$\emptyset$ constitutes a stage $s$
$\emptyset$-QC for all $s\in \{ 1, 2,3 \}$.  For a block $B$ with
$h(B)>0$ and $s\in \{ 1, 2,3 \}$, a stage $s$ $\mathtt{T}$-QC for $B$
is, roughly, a set of correctly specified stage $s$ votes for $B$ from
players who own at more than two-thirds of the stake according to
$\mathtt{S}(S_0,\mathtt{T})$. Precisely, we say that a set of votes
$Q$ is a \emph{stage $s$ $\mathtt{T}$-QC for $B$} if the following
conditions are all satisfied:
\begin{enumerate}

\item[(QC1)] {\em Agreement:} For every $V\in Q$, $B^*(V)=B$ and
  $s^*(V)=s$. 
In this case, we define $B^*(Q):=B$, $e^*(Q):=e^*(B)$, $v^*(Q):=v^*(B)$, $h^*(Q):=h^*(B)$,  and
  $s^*(Q):=s$.


\item[(QC2)] {\em Correct stake amounts (w.r.t.~$\mathtt{T}$).} For each
  $V\in Q$, 
$c^*(V) = \mathtt{S}(S_0,\mathtt{T},id^*(V))$.
  
\item[(QC3)] {\em No double-voting.} There do not exist distinct
  $(id,(c,B,s)), (id',(c',B',s')) \in Q$ with $id=id'$.

\item[(QC4)] {\em Supermajority of stake.} $\sum_{V\in Q} c^*(V)> \frac{2}{3}N$. 

\end{enumerate}

\vspace{0.2cm} 
\noindent \textbf{Selecting leaders}. We denote by $\mathtt{leader}$ a
function that allocates leaders to views within an epoch in proportion
to their stake. Because this function is stake-dependent, it must be
defined relative to a set of transactions~$\mathtt{T}$ (again, for us,
the transactions confirmed in previous epochs). Formally, we let
$\mathtt{leader}$ be any function that satisfies: If
$\mathtt{S}(S_0,\mathtt{T},id)=c$, then, for each $k\in \mathbb{N}$
there exist $c$ many~$v$ in the interval $[kN,(k+1)N)$ such that
$\mathtt{leader}(\mathtt{T},v)=id$. For example, given $\mathtt{T}$,
$\mathtt{leader}(\mathtt{T},\cdot)$ could repeatedly iterate through
all~$N$ coins (with $\mathtt{leader}(\mathtt{T},v)$ the identifier
that owns the coin $v \bmod N$) in some canonical order.
 
\vspace{0.2cm} 
\noindent \textbf{The variables $\mathtt{T}_e$ and} $B_{g,e}$.  Each
player $p$ maintains local variables $\mathtt{T}_e$ and $B_{g,e}$ (for
each $e\in \mathbb{N}$), which are initially undefined and become
defined upon the player beginning epoch $e$:
intuitively,~$\mathtt{T}_e$ is the set of transactions that $p$
regards as confirmed by epochs $e'<e$, while $B_{g,e}$ is the block of
height $eN$ that $p$ first sees as confirmed (regarded as a ``genesis
block'' for epoch $e$).  The transactions of~$\mathtt{T}_e$ will be
the ones that determine who should be leaders and who should be voting
(and with weight) during epoch $e$. The instructions of epoch $e$
attempt to build the blockchain above $B_{g,e}$.


\vspace{0.2cm} 
\noindent \textbf{The variable} $M$. Each honest player~$p$ maintains
a local variable $M$, which is the set of all messages received by $p$
thus far. Initially, $M:= \{ B_{g} \} $.

\vspace{0.2cm}
\noindent \textbf{Defining} $B^*_{g,e}(M), \mathtt{T}^*_e(M)$
\textbf{and block validity}. We next specify the functions that an
honest player uses to define its local variables~$B_{g,e}$ and~$\mathtt{T}_e$
based on its local variable~$M$.  The function $\mathtt{T}^*_e$ is
defined on input~$M$ if and only if~$\mathtt{B}^*_{g,e}$ is; in this
case,
$\mathtt{T}^*_{e}(M)=\cup_{B''\in \mathcal{B}} \mathtt{T}^*(B'')$,
where $\mathcal{B}$ denotes the ancestors of $B^*_{g,e}(M)$.  That is,
whenever~$\mathtt{T}^*_e(M)$ is defined, it is all the transactions
included in the blocks between the original genesis block~$B_g$ and
$\mathtt{B}^*_{g,e}(M)$ (inclusive).

We can now define (or leave undefined), for a fixed set of
messages~$M$, the value of~$B^*_{g,e}(M)$ for all~$e \ge 0$. We
proceed by induction on~$e$. For~$e=0$, $B^*_{g,0}(M):=B_g$ (and
$\mathtt{T}^*_0(M) := \mathtt{T}^*_{-1}(M) := \emptyset$).  Now
consider arbitrary~$e \ge 0$ and suppose that $B^*_{g,e}(M)$ 
is
defined (as otherwise we leave~$B^*_{g,e+1}(M)$ undefined as well).
We next define
what it means for a block~$B$ in epoch~$e$ to be ``valid'';
intuitively,~$B$ is valid if it is formed by the correct leader, if
its parent is valid and $B$ correctly reports a stage 1 QC for the
parent, if it correctly reports its own height, and if it includes a
set of transactions that is valid relative to those already included
in ancestor blocks other than $B$.\footnote{Thus, formally, we define
  validity for a block $B$ (with respect to~$M$) by induction
  on~$e^*(B)$ and then on height.}  Precisely, the block
$B=(id, (h,e,v,\mathtt{T},B',Q))$ is {\em $M$-valid} if and only if
$B=B_g$ or:
\begin{itemize} 
\item {\em Correct leader:} $id=\mathtt{leader}(T^*_e(M),v)$. 
\item {\em Valid transactions:} $\cup_{B''\in \mathcal{B}}
  \mathtt{T}^*(B'')$ is a valid set of transactions relative to $S_0$,
  where $\mathcal{B}$ denotes the ancestors of $B$. 
\item {\em Valid predecessor:} $B'$ is $M$-valid, with $h(B')=h-1$.
\item {\em Parent supported by QC (case 1):} 
If $e^*(B')=e-1$ then $Q$ is a stage 1 $\mathtt{T}_{e-1}^*(M)$-QC for $B'$. 
\item {\em Parent supported by QC (case 2):} 
If $e^*(B')=e$ then $Q$ is a stage 1 $\mathtt{T}_{e}^*(M)$-QC for $B'$. 
\end{itemize}
In particular, $p$ regards a block~$B$ as $M$-valid only if it regards
all of~$B$'s ancestors as $M$-valid, and in particular has seen stage
1 QCs for all of $B$'s ancestors. (The player may not have seen stage
3 QCs for all of these blocks, however.)

Finally, 
$B^*_{g,e+1}(M)$ is defined if and only if:
\begin{itemize} 
\item there is at least one~$M$-valid block~$B$ with~$e^*(B)=e$
and $h^*(B) \ge eN+N$ for which~$M$ contains stage 1, 2, and
3~$\mathtt{T}_e$-QCs for~$B$; 
\item no two such blocks are incompatible; and
\item $B$ has $B^*_{g,e}(M)$ as an ancestor;
\end{itemize}
in this case, $B^*_{g,e+1}(M)$ is defined as the unique ancestor~$B'$
of~$B$ at height $eN+N$ (and~$\mathtt{T}^*_{e+1}(M)$ is then defined, as usual,
as all transactions included in~$B'$ and its ancestors).

\vspace{0.2cm} 
\noindent \textbf{The confirmation rule}. Given a set~$M$ of
messages, the confirmed transactions are defined as
$\mathcal{C}(M):= \mathtt{T}_e^*(M)$, where $e$ the largest value for which
$\mathtt{T}_e^*(M)\downarrow$.  

\vspace{0,2cm}
\noindent \textbf{Comment}.  It is a subtle but important point in the
definitions above that, while executing the instructions for epoch $e$,
honest players may produce blocks of height $>eN+N$ with stage 3
$\mathtt{T}_{e}$-QCs (where $\mathtt{T}_{e}$ is as defined for some
honest player), but the transactions in such blocks will not be
considered confirmed by the confirmation rule. The reason for this is explained in Section \ref{overv}. 
  

  
\vspace{0.2cm}
\noindent \textbf{View certificates (VCs)}. Recall that, when a player $p$ wishes
to enter view $v$ in our permissioned version of HotStuff, $p$
disseminates a $\mathtt{view}\ v$ message. The role of this message is
two-fold. First, it alerts other players that $p$ wishes to enter view
$v$. Second, the message relays $p$'s present value for its local
variable $Q^1$ to the leader of view $v$: as explained in the liveness
argument of Section \ref{ph}, this is the information the leader needs
to propose a block that other honest players will vote for. In our PoS
version of the protocol, we use a similar setup.

A $\mathtt{view}\ (e,v)$ message is a tuple $r=(id,(c,e,v,Q))$. Such a
message is {\em $M$-valid} if all of the following are true:
\begin{itemize}
\item {\em Epoch~$e$ has begun according to~$M$:}  $\mathtt{T}^*_e(M)\downarrow$;
\item {\em Correct stake amount (w.r.t.~$\mathtt{T}^*_e(M)$):} 
$c = \mathtt{S}(S_0,T^*_e(M),id)$;
\item {\em Supported by QC:} $Q=\emptyset$ or $Q$ is a stage 1
  $\mathtt{T}^*_e(M)$-QC  for some block $B$ that is $M$-valid with $e^*(B)=e$.
\end{itemize} 
For such a tuple~$r$, define $c^*(r):=c$. A set $R\subseteq M$ of
$\mathtt{view}\ (e,v)$ messages is then an {\em $M$-VC for $(e,v)$} if
and only if the following three requirements are met:
\begin{itemize} 
\item {\em Validity:} every $r\in R$ is $M$-valid;
\item {\em No double-voting:} there do not exist distinct
  $(id,(c,e,v,Q)), 
  (id',(c',e',v',Q')) \in R$ with $id=id'$; and 
\item {\em Supermajority of stake:} $\sum_{r\in R} c^*(r) > \frac{2}{3}N$.
\end{itemize}    
  
\vspace{0.2cm} 
\noindent \textbf{Choosing blocks to propose}.  When an honest
player~$p$ controls the identifier~$id$ that is the leader of a
view~$v$ (within some epoch~$e$), they are instructed to propose a
block as follows. Suppose~$p$ entered~$v$ on account of receiving an
$M$-VC for~$(e,v)$, which we denote by~$R$.  First, suppose that no
one has an opinion about which existing block to build on, in the
sense that,
for every $r=(id,(c,e,v,Q)) \in R$,
$Q = \emptyset$.
Then,
$p$ is instructed to extend the block it sees as the
genesis block of this epoch with a new block that contains all of the
not-yet-confirmed transactions that~$p$ is aware of. Formally, in this
case, define the block proposal function $\mathtt{block}$ by
\begin{equation}\label{eq:block1}
\mathtt{block}(\underbrace{B_{g,e},\mathtt{T}^*,M}_{\text{$p$'s local vars}},id):= (id,
  (h^*(B_{g,e})+1,e,v,\mathtt{T},B_{g,e},Q)),
\end{equation}
where~$\mathtt{T}^*$ denotes the set of all transactions that $p$ has
received; $\mathtt{T}=\mathtt{T}^*-\cup_{B''\in \mathcal{B}}
\mathtt{T}^*(B'')$, where $\mathcal{B}$ denotes the ancestors of
$B_{g,e}$; and~$Q$ is a stage 1 $\mathtt{T}^*_{e-1}(M)$-QC for
$B_{g,e}$. 

Otherwise, the leader proceeds similarly to the permissioned HotStuff
protocol, choosing the QC~$Q$ with largest view and proposing a block
with $B^*(Q)$ as the parent. Formally, from amongst all $Q'$ such that
there exists $r=(id,(c,e,v,Q')) \in R$, choose $Q$ such that $v^*(Q)$
is maximal and define $\mathtt{block}$ by
\begin{equation}\label{eq:block2}
\mathtt{block}(B_{g,e},\mathtt{T}^*,M,id):= (id,
  (h^*(Q)+1,e,v,\mathtt{T},B^*(Q),Q));
\end{equation}
here
$\mathtt{T}=\mathtt{T}^*-\cup_{B''\in \mathcal{B}} \mathtt{T}^*(B'')$,
where $\mathcal{B}$ denotes the ancestors of $B^*(Q)$.

\vspace{0.2cm} 
\noindent \textbf{The variables $Q^1_e$ and $Q^2_e$}. Each player
maintains variables $Q^1_e$ and $Q^2_e$ which, analogous to the
variables $Q^1$ and $Q^2$ in our description of permissioned HotStuff
in Section~\ref{ph}, implement a locking mechanism within epoch $e$.

\vspace{0.2cm} The pseudocode for the PoS-HotStuff protocol is
described in Algorithm 1. There are six main blocks of pseudocode,
which describe when and how a player enters a new epoch (and what to
do at that time); when and how to enter a new view (and, if the view's
leader, disseminate a block); when and how to disseminate a stage 1,
stage 2, and stage 3 vote; and when and how to disseminate a
$\mathtt{view}$ message.
  
\begin{algorithm}
\caption{PoS-HotStuff (instructions for a player $p$)}
\begin{algorithmic}[1]

    \State \textbf{Local variables} 
    
   \State $Q^1_e$, initially set to $\emptyset$, $Q^2_e$, initially undefined        \Comment Used to implement locking
   
      \State $B_{g,e}, \mathtt{T}_e$, initially undefined          \Comment Defined upon entering epoch $e$
        
    \State $\mathtt{T}^{\ast}$, initially $\mathtt{T}^{\ast}:=\emptyset$          \Comment The set of all transactions received by $p$
             
    \State $M$, initially $M:= \{ B_{g} \} $          \Comment The set of all messages received by $p$
    
        \State $e$, initially $e:= -1 $          \Comment $p$'s present epoch 
        
            \State $v$, initially $v:= -1 $          \Comment $p$'s present view
            
                  \State $B$, initially undefined        \Comment Records a block value 

    \State 
%
%
%
%
%
    
            \State \textbf{At every timeslot $t$}
        
        \State \hspace{0.3cm} Update $M$ and $\mathtt{T}^{\ast}$; 
        
         \State \hspace{0.3cm} Disseminate all transactions in  $\mathtt{T}^{\ast}$ received for the first time at $t$; 
        
            \State

            \State \textbf{Upon} $\mathtt{T}_{e'}^*(M)\downarrow$ for $e'>e$: 
            
            \State  \hspace{0.3cm} Set $\mathtt{T}_{e'}:= \mathtt{T}_{e'}^*(M)$, $B_{g,e'}:=B_{e'}^*(M)$, $e:=e'$,  $v:=-1$;       \Comment Enter new epoch
            
            \State  \hspace{0.3cm}    For each $id\in \mathtt{id}(p)$ such that $\mathtt{S}(S_0,\mathtt{T}_e,id)>0$: 
        
       \State  \hspace{0.6cm}         Let $c:= \mathtt{S}(S_0,\mathtt{T}_e,id)$;  
        
              \State  \hspace{0.6cm}   Disseminate the view $(e,0)$  message  $(id,(c,e,0,\emptyset))$;

            \State

                       \State \textbf{Upon} receiving an $M$-VC for $(e,v')$ with $v'>v$:
                       
                                 \State  \hspace{0.3cm} Set  $v:=v'$; Set $B\uparrow$;                \Comment Enter new view
                       
                           \State  \hspace{0.3cm} Set  the $(e,v)$-timer to expire after $5\lceil \Delta/\kappa \rceil$ active non-waiting timeslots;    \Comment Set timer
                       
                          \State  \hspace{0.3cm} \textbf{If} $\mathtt{leader}(\mathtt{T}_e,v)=id\in \mathtt{id}(p)$ \textbf{then}:

           \State  \hspace{0.6cm}  Disseminate $
           \mathtt{block}(B_{g,e},\mathtt{T}^*,M,id)$; \Comment Leader
           disseminates block as per~\eqref{eq:block1} or~\eqref{eq:block2}
            
                        \State

               \State \textbf{Upon} first receiving an $M$-valid block  $(id, (h,e,v,\mathtt{T},B',Q))$:

          \State  \hspace{0.3cm} \textbf{If} $Q_e^2\uparrow$ or
          $Q_e^2 \leq v^*(Q)$ \textbf{then}:     \Comment Checks lock

         \State  \hspace{0.6cm} Set $B:=(id, (h,e,v,\mathtt{T},B',Q))$;
          
                    \State  \hspace{0.6cm}  For each $id\in \mathtt{id}(p)$ such that $\mathtt{S}(S_0,\mathtt{T}_e,id)>0$: 
        
       \State  \hspace{0.9cm}         Let $c:= \mathtt{S}(S_0,\mathtt{T}_e,id)$;  
        
              \State  \hspace{0.9cm}   Disseminate the vote  $(id,(c,B,1))$;      \Comment Stage 1 vote

                                    \State

          \State \textbf{Upon} $B\downarrow $ and first receiving $Q'$ which is a stage 1 $\mathtt{T}_e$-QC for $B$:

                \State \hspace{0.3cm} Set $Q^1_e:= Q'$;                    \Comment Sets $Q^1_e$

                       \State  \hspace{0.3cm}  For each $id\in \mathtt{id}(p)$ such that $\mathtt{S}(S_0,\mathtt{T}_e,id)>0$: 
        
       \State  \hspace{0.6cm}         Let $c:= \mathtt{S}(S_0,\mathtt{T}_e,id)$;  
        
              \State  \hspace{0.6cm}   Disseminate the vote  $(id,(c,B,2))$;      \Comment Stage 2 vote

                                                \State

              \State \textbf{Upon} $B\downarrow $ and first receiving $Q''$ which is a stage 2 $\mathtt{T}_e$-QC for $B$:

                \State \hspace{0.3cm} Set $Q^2_e:= Q''$;                    \Comment Sets $Q^2_e$

                       \State  \hspace{0.3cm}  For each $id\in \mathtt{id}(p)$ such that $\mathtt{S}(S_0,\mathtt{T}_e,id)>0$: 
        
       \State  \hspace{0.6cm}         Let $c:= \mathtt{S}(S_0,\mathtt{T}_e,id)$;  
        
              \State  \hspace{0.6cm}   Disseminate the vote  $(id,(c,B,3))$;      \Comment Stage 3 vote

              \State            
              
                 \State \textbf{Upon} $(e,v)$-timer expiry \textbf{or} first receiving a stage 3 $\mathtt{T}_e$-QC for $B$:                                 
         
            \State  \hspace{0.3cm}  For each $id\in \mathtt{id}(p)$ such that $\mathtt{S}(S_0,\mathtt{T}_e,id)>0$: 
        
       \State  \hspace{0.6cm}         Let $c:= \mathtt{S}(S_0,\mathtt{T}_e,id)$;  
        
              \State  \hspace{0.6cm} Disseminate the $\mathtt{view} (e,v+1)$ message $(id,(c,e,v+1,Q^1_e))$;

\end{algorithmic}
\end{algorithm}

\subsection{Consistency, accountability, and weakly $\rho$-bounded adversaries} \label{secprof} 

We prove in this section that, whenever an execution of the
PoS-HotStuff protocol leads to a consistency violation, it is possible
to extract from the disseminated messages a set of identifiers that
are guaranteed to belong to Byzantine players and that collectively
control at least one-third of the overall stake.  That is, we prove
that the protocol is~$(1/3,1)$-accountable in the sense of
Section~\ref{ss:da_imp2}. In particular, the protocol satisfies
consistency with respect to a $\rho$-bounded adversary (in the sense
of Section~\ref{bps}) for all $\rho < 1/3$. We conclude this section
with the additional arguments needed to accommodate a weakly
$\rho$-bounded adversary, as required by Theorem~\ref{posPoS2}.

\vspace{0.2cm} 
\noindent \textbf{Defining a blame function}.  The definition of
accountability involves a {\em blame function}~$F$ that maps a set of
messages to a set of identifiers controlled by Byzantine
players. Intuitively, our blame function will return identifiers that
double-voted in the protocol's execution. Precisely, given a message
set~$M$, define~$F(M)$ as the set of all identifiers~$id$ that cast
two votes~$V,V' \in M$ such that:
\begin{itemize} 
\item [(F1)] $V$ is a stage 3 vote for some block $B$, signed by identifier $id$; 
\item [(F2)] $V'$ is a stage 1 vote for some block $B'\neq B$, signed by identifier $id$; 
\item [(F3)] $e^*(B)=e^*(B')$ and $v^*(V)\leq v^*(V')$; and
\item [(F4)] $v^*(Q^*(B')) <v^*(V)$.
\end{itemize} 
To see why~$F(M)$ can never include an identifier controlled by an
honest player, suppose that such a player~$p$ casts a stage 3 vote for
a block~$B$ in the view~$v:=v^*(V)$.  Within a view, an honest player
votes for at most one block (the first block proposal it receives from
the view's leader), so~$p$ will not vote for any block other than~$B$
in the view~$v$. Because~$p$ casts a stage 3 vote for~$B$ in view~$v$,
it resets its local variable~$Q^2_e$ at that time so that
$v^*(Q^2_e)=v$; further resets at later views within the epoch
will only increase its value.  Thus, at all views $v' > v$, $p$ will
never cast a vote for a block~$B'$ for which $v^*(Q^*(B')) < v$.


\vspace{0.2cm} 
\noindent \textbf{Consistency failures}.  Consider an execution of the
PoS-HotStuff protocol that sees a consistency failure, in the form of
two incompatible blocks that are confirmed by different honest
players.  By the definition of block validity, players that agree on
their local values of~$B_{g,e}$ and~$B_{g,e+1}$ for some $e \ge 0$
also agree on all blocks in between. By the definition of the
confirmation rule, a player only confirms descendants of~$B_{g,e}$ for
values of~$e \ge 0$ for which both~$B_{g,e}$ and~$B_{g,e+1}$ are
defined. Thus, a consistency failure occurs only if, for
some~$e \ge 0$, two honest players set distinct values
for their local variable~$B_{g,e+1}$.

We say that a block $B$ is \emph{seen as valid} if there exists
some honest $p$ and some timeslot $t$ at which~$B$ is $M$-valid, where
$M$ is as locally defined for $p$ at $t$.
Let~$e$ be the first epoch at which honest players set distinct values
for~$B_{g,e+1}$.
Recalling the definition of the function~$B^*_{g,e+1}(\cdot)$, 
there must exist incompatible blocks $B$ and $B''$---descendants of
the competing values for~$B_{g,e+1}$, each with height
at least $eN+N$---such that:
\begin{itemize} 
\item  $e^*(B)=e^*(B'')=e$;
\item $B$ and $B''$ are both seen as valid; 
\item  $B$ and $B''$ both receive stage 1,2, and 3 $\mathtt{T}_e$-QCs. 
\end{itemize} 
Without loss of generality, suppose $v^*(B)\leq v^*(B'')$, set
$v:=v^*(B)$, and let $Q$ be a stage 3 $\mathtt{T}_e$-QC for~$B$. 
Among all blocks~$B'$ incompatible with~$B$ that received a stage 1
QC~$Q'$ at a view $v' \ge v$, choose one with~$v'$ as small as
possible. (Such a block exists because~$B''$ is one candidate.)
This choice of~$v'$ ensures that the parent of~$B'$ is from a
view prior to~$v$ (i.e., $v^*(Q^*(B'))<v$).

A pair of votes for~$Q$ and~$Q'$ satisfies~(F1)--(F4), so every
identifier that casts such a pair will be included in~$F(Q \cup Q')$.
Because  $Q$ is a stage 3 $\mathtt{T}_e$-QC for $B$, no identifier
appears twice and $\sum_{V \in Q}
\mathtt{S}(S_0,\mathtt{T}_e,id^*(V)))> \frac{2}{3}N$.
Because  $Q'$ is a stage 1 $\mathtt{T}_e$-QC for $B'$, no identifier
appears twice and $\sum_{V' \in Q'}
\mathtt{S}(S_0,\mathtt{T}_e,id^*(V')))> \frac{2}{3}N$.
It follows that:  
\begin{enumerate} 
\item[$(\dagger_1)$]
  $\sum_{id\in F(Q \cup Q')} \mathtt{S}(S_0,\mathtt{T}_e,id) > \frac{1}{3}N$,
  implying that the protocol is
  $(1/3,1)$-accountable.\footnote{Accountability, as defined in
    Section~\ref{ss:da_imp2}, requires only that the blamed
    identifiers be assigned a large amount of stake with respect to
    some subset of confirmed transactions. Here, this subset is
    precisely specified and natural: the set~$\mathtt{T}_e$ of transactions
    confirmed by all honest players in the epochs that precede the
    first epoch~$e$ with a consistency failure.}

\item[$(\dagger_2)$] $\mathtt{T}_e$ allocates $\geq \frac{1}{3}N$
  stake to Byzantine players, implying that, for every $\rho < 1/3$,
  the protocol satisfies consistency with respect to a $\rho$-bounded
  adversary.
\end{enumerate}

 
\vspace{0.2cm} 
\noindent \textbf{Proving that the adversary remains $\rho$-bounded}.
To satisfy the requirements of Theorem~\ref{posPoS2} (as opposed to
only Theorem~\ref{posPoS}), we consider next a weakly $\rho$-bounded
adversary, meaning that we make no blanket assumptions on the fraction
of stake controlled by Byzantine players throughout the protocol
execution and instead restrict attention to an environment that is
$\rho$-bounded in the sense of Section~\ref{envrho}.

Assume that the environment is $\rho$-bounded for some $\rho<1/3$.
We argue by induction that, for every epoch~$e \ge 0$:
\begin{enumerate} 
\item [(P1)] For all $e' \le e$, no two honest players define distinct
values for $B_{g,e'}$ or $\mathtt{T}_{e'}$.
\item [(P2)] If~$\mathtt{T}_e \downarrow$,
$\mathtt{S}(S_0,\mathtt{T}_e)$ allocates at most a
  $\rho$-fraction of the total stake to Byzantine players.
\end{enumerate}  
The base case of~$e=0$ holds by the initialization of the protocol and
the assumption of a $\rho$-bounded environment.  Now suppose~(P1)
and~(P2) hold for some epoch~$e \ge 0$.  Using $(\dagger_2)$, above,
property~(P1) will hold also for epoch~$e+1$. To prove~(P2), suppose
there is some timeslot at which $B_{g,e+1}$ and~$\mathtt{T}_{e+1}$ become
defined for some honest player.  (Otherwise, there is nothing to prove.)
Because $\mathtt{T}_e$ allocates the adversary at most a
$\rho$-fraction of the stake with~$\rho < 1/3$, at least one vote in
any stage 1 $\mathtt{T}_e$-QC for~$B_{g,e+1}$ must be by an honest
player. Let $t^*$ be a timeslot at which some honest player~$p$
disseminates a stage 1 vote for $B_{g,e+1}$. At $t^*$, $\mathtt{T}_e$
is the set of confirmed transactions for $p$, and all transactions
in~$\mathtt{T}_{e+1}$ have already been issued by the environment. Let
$\mathtt{T}_{\text{hon}}$ denote the set of all honest transactions issued
by the environment at timeslots $\leq t^*$.  From the definition of a
$\rho$-bounded environment in
Section~\ref{envrho}---with~$\mathtt{T}_e$
and~$\mathtt{T}_{\text{hon}}$ playing the roles of~$\mathtt{T}_1$
and~$\mathtt{T}_2$, respectively---it follows that $\mathtt{T}_e$
contains all the transactions required for any transaction in
$\mathtt{T}_{\text{hon}}$ to be valid relative to~$S_0$,
$\mathtt{T}_e \cup \mathtt{T}_{\text{hon}}$ is a valid set of
transactions relative to $S_0$, and
$\mathtt{S}(S_0, \mathtt{T}_e \cup \mathtt{T}_{\text{hon}})$ allocates
the adversary at most a $\rho$-fraction of the stake.  Because honest
transactions can only increase the adversary's share of stake,
$\mathtt{S}(S_0, \mathtt{T}_e \cup (\mathtt{T}_{e+1} \cap
\mathtt{T}_{\text{hon}}))$ also allocates the adversary at most a
$\rho$-fraction of the stake.  Because Byzantine transactions can only
decrease the adversary's share of stake,
$\mathtt{S}(S_0, \mathtt{T}_{e+1})$ must allocate the adversary at
most a $\rho$-fraction of the stake, completing the argument.

\subsection{Liveness and optimistic responsiveness}\label{ss:or}

We first show that the PoS-HotStuff protocol satisfies liveness (in
the sense of Section~\ref{ss:liveness}) with respect to a weakly
$\rho$-bounded adversary with $\rho < 1/3$, and then extend the
argument to show that the protocol satisfies optimistic responsiveness
(in the sense of Section~\ref{ss:da_imp3}). We make no attempt to
optimize the constants in our liveness parameters.

\vspace{0.2cm} \noindent \textbf{Proving liveness}.  Order epoch-view
pairs $(e,v)$ lexicographically; in an abuse of terminology, we often
refer to such pairs as ``views.''  Because the environment is
$\rho$-bounded, the consistency argument in Section~\ref{secprof}
implies that honest players always agree on the values of
$B_{g,e}$ and $\mathtt{T}_e$ (whenever they are
defined). When~$B_{g,e}$ (and hence $\mathtt{T}_e$) is defined, the
``leader of view $(e,v)$'' refers to the identifier
$\mathtt{leader}(\mathtt{T}_e,v)$.  In this case, call a view~$(e,v)$
{\em post-GST} if every honest player is in some view earlier
than~$(e,v)$ at GST, and call it {\em good} if it is post-GST and
$\mathtt{leader}(\mathtt{T}_e,v)$ is controlled by an honest player.
Call a player a {\em stakeholder} in view~$(e,v)$ if at least one if
its identifiers is allocated a non-zero amount of stake in
$\mathtt{S}(S_0,\mathtt{T}_e$).  By the definition of the
$\mathtt{leader}$ function, the leader of a view must be controlled by
a stakeholder.  By the definition of the quasi-permissionless setting,
honest stakeholders are always active.

The first claim asserts that good views result in blocks with
stage 3 QCs, unless some honest player transitions to a larger epoch
first.

\vspace{0.2cm} 
{\sc Claim 1:}
If~$(e,v)$ is a good view that is first reached by an
honest player ($p$, say) at timeslot~$t$ (with~$t$ necessarily larger
than GST), then by timeslot~$t + 5\Dk$, either:
\begin{itemize}

\item [(i)] a block proposed in view~$(e,v)$ receives stage 1, 2,
  and 3 QCs in that view, which are seen by all honest stakeholders;
  or

\item [(ii)] some honest player has entered an epoch $e' > e$.

\end{itemize}

\begin{proof}
Because
$t \ge \text{GST}$, by timeslot $t + \lceil \Delta/\kappa \rceil$,
every honest stakeholder will have received all the messages that
had been received by~$p$ by timeslot~$t$ and must therefore be in some
view at least as large as~$(e,v)$.  
Suppose first that no honest stakeholder enters a view higher than~$(e,v)$
at or before timeslot~$t + 4\Dk$.
In this case, the (honest and active) leader~$p'$ of view~$(e,v)$ 
will disseminate a block proposal~$B$ for the view by
timeslot $t + \Dk$.  All honest stakeholders will receive this
proposal by timeslot
$t+ 2 \lceil \Delta/\kappa \rceil$, and will regard the proposed block
as valid by that timeslot (having then received all messages received
by $p'$ by
timeslot $t+ \lceil \Delta/\kappa \rceil$).  
Let~$p''$ denote an honest stakeholder with the maximum value
of~$v^*(Q^2_e)$ upon entering view~$(e,v)$. Because~$p''$ set
this value upon seeing a stage 2 $\mathtt{T}_e$-QC for $B^*(Q^2_e)$, there
must exist a set of honest stakeholders~$P_1$, controlling at least~$N/3$
units of stake (according to~$\mathtt{T}_e$), that saw a stage 1
$\mathtt{T}_e$-QC for $B^*(Q^2_e)$ while in view $v^*(Q_e^2)$ and
therefore set their local values $Q^1_e$ so that
$v^*(Q^1_e)=v^*(Q_e^2)$ while in this view. If any player in $P_1$
subsequently updated their local value $Q^1_e$, then $v^*(Q^1_e)$ must
have increased with this update.  Now let $R$ denote the VC received by
the leader $p'$ that caused it to enter view~$(e,v)$, and let $P_2$ denote
the stakeholders that disseminated the messages in $R$.  Because the
players in $P_2$ own total stake more than $2N/3$ (according
to~$\mathtt{T}_e$), there must be some (honest) stakeholder
in the intersection $P_1 \cap P_2$ that contributed a
$\mathtt{view}~(e,v)$ message to~$R$ with a QC~$Q$
satisfying~$v^*(Q) \ge v^*(Q_e^2)$. Because $p'$ is honest, and by the
definition of the function $\mathtt{block}$, 
the block proposal by~$p'$ satisfies $v^*(Q^*(B)) \geq v^*(Q_e^2)$.
Thus, due to our choice of~$p''$, all honest stakeholders will disseminate
stage 1 votes for $B$ by timeslot $t+ 2 \lceil \Delta/\kappa \rceil$,
and will then disseminate stage 2 and 3 votes for $B$ by timeslots
$t + 3\Dk$ and $t+ 4 \lceil \Delta/\kappa \rceil$, respectively.  
All honest stakeholders will have therefore seen stage 1, 2, and 3 QCs
for~$B$ by timeslot~$t + 5\Dk$, and so condition~(i) holds.

Finally, suppose that some honest stakeholder enters a
view~$(e',v')$ that is larger than~$(e,v)$ at some timeslot
$\le t + 4\Dk$. If~$e' > e$, condition~(ii) holds.  Suppose
that~$e'=e$ and $v' > v$. Because no honest player entered
view~$(e,v)$ prior to time~$t$, no honest player's timer for that view
has expired by timeslot~$t + 4\Dk$. Thus, some honest player must have
received by this timeslot a VC for view~$(e,v+1)$, implying that some
honest player must have received by this timeslot a stage 3 QC for a
block~$B$ proposed in view~$(e,v)$. All honest stakeholders will
therefore receive stage 1, 2, and 3 QCs for~$B$ by
timeslot~$t + 5\Dk$, and so condition~(i) holds.
\end{proof}

Call an epoch~$e$ {\em post-GST} if view~$(e,0)$ is post-GST.  The
next claim bounds the length of such an epoch.

\vspace{0.2cm} 
{\sc Claim 2:} If~$e$ is a post-GST epoch first reached by an
honest player at timeslot~$t$, 
then for every timeslot~$\ge t +
12N \Dksq$, $B_{g,e+1}$ is defined for all honest players active at that
time. Moreover, more than half of the confirmed epoch-$e$
blocks are proposed by honest players.

\begin{proof}
Every view~$(e,v)$ is post-GST. Every such view with an honest leader
is a good view. By Claim~1, every good view results in either the
reception by all honest stakeholders of a stage 3 QC for a block proposed
in that view or some honest player proceeding to a subsequent epoch.
If any honest player proceeds to an epoch~$e' > e$ by
timeslot~$t+(12N-1)\Dksq$, then all active honest players proceed to
such an epoch by timeslot~$t+12N\Dksq$ and have their local
variables $B_{g,e+1}$ defined by that time.
So, assume that no honest player proceeds to an epoch larger than~$e$
at or before timeslot $t+(12N-1)\Dksq$.

Because honest stakeholders within an epoch are always active, and
because each view lasts for at most $6\Dksq$ timeslots---at most $5\Dksq$
timeslots for timer
expiries, plus~$\le \Dk \le \Dksq$ timeslots to form VCs to enter the next
view---the first~$2N-1$ views of epoch~$e$ complete by the timeslot
$t+(12N-1)\Dksq$.\footnote{We assume that a player does not know
  whether its clock is fast or slow, in the sense that it cannot
  discern the passage of timeslots at which it is active but waiting.
  Accordingly, in the PoS-HotStuff protocol, a player's timer is set
  to expire after~$5\Dk$ active and non-waiting timeslots, which may
  translate to as many as $5\Dksq$ actual timeslots. (Technically, the
  bound should be $\lceil 5 \Dk /\kappa \rceil$, but we'll work with
  the approximation $\Dksq$ for simplicity.)}
Because honest stakeholders control more than
two-thirds of the overall stake (according to $\mathtt{T}_e$) and the
function~$\mathtt{leader}$ assigns leaders to views proportional to
stake, at least~$N$ of the first~$2N-1$ views of epoch~$e$ must be
good. (For simplicity, assume that~$N$ is at least, say, 10.)  Under
the assumption that no honest player proceeds to an epoch larger
than~$e$ at any timeslot $\le t+(12N-1)\Dksq$, Claim~1 implies that
every such good view results in the reception by all honest stakeholders of
a stage 3 QC for a block proposed in that view.  Thus, by the time the
first~$2N-1$ views of epoch~$e$ have completed, all honest
stakeholders will have received stage 3 QCs for at least~$N$ different
blocks.  The proof of consistency in Section~\ref{secprof} shows that
no two such blocks can be incompatible.  All such blocks
have~$B_{g,e}$ as an ancestor (as honest players will only vote for
such blocks).  At least one of these blocks~$B$ must have height at
least~$eN+N$.  Thus, by timeslot $t+(12N-1)\Dksq$, all honest
epoch-$e$ stakeholders will have defined their local
variable~$B_{g,e+1}$ as the (unique) block at height~$eN+N$ that is an
ancestor of~$B$.  It follows that, for every
timeslot~$\ge t + 12N\Dksq$, $B_{g,e+1}$ is likewise defined for all
honest players active at that time.

For the second statement of the claim, because honest stakeholders
control more than two-thirds of the overall stake (according to
$\mathtt{T}_e$) and the function~$\mathtt{leader}$ assigns leaders to
views proportional to stake, more than $N/2$ of the
views~$(e,0),\ldots,(e,N-2)$ have honest leaders and hence are good
views. Epoch~$e$ cannot end until stage 1 QCs have been produced for
at least~$N$ different blocks. Because each view results in a stage 1
QC for at most one block (as no honest player will vote for two
different blocks in the same view), no honest player can enter an
epoch $e' > e$ prior to some honest player entering the
view~$(e,N-1)$.  Thus, each of the $> N/2$ good views
among~$(e,0),\ldots,(e,N-2)$ results in a stage~3 QC for a block with
height $< eN+N$ that was proposed by an honest player.  All such
blocks will be ancestors of~$B_{g,e+1}$ (as no two epoch-$e$ blocks
with stage 3 QCs can be incompatible); as such, they will be confirmed
at the conclusion of epoch~$e$.
\end{proof}

Claim~2 implies that, if a transaction $\mathtt{tr}$ has been received
by all honest stakeholders by the start of a post-GST epoch, then, as
long as $\mathtt{tr}$ remains valid, it will be confirmed by the end
of the epoch. (The (honest) leader of the first good view in the epoch
will include $\mathtt{tr}$ in its block proposal~$B$, unless
$\mathtt{tr}$ is invalid or has already been included in some ancestor
of~$B$.)

The final claim bounds the amount of time between GST and the start of
the first post-GST epoch.

\vspace{0.2cm} 
{\sc Claim 3:} If~$t \ge \text{GST}$, there is a post-GST epoch that is
first reached by an honest player at some
timeslot in~$[t,t+(12N+7)\Dksq]$.

\begin{proof}
Let~$(e,v)$ denote the largest view that any honest player 
is in at timeslot~$t$, and define~$t^*:=t+7\Dksq$.
Let~$(e^*,v^*)$ denote the largest view that any honest player
is in at timeslot~$t^*$. 
The proof is complete if~$e^* > e$,
so suppose~$e^*=e$. Then, $v^* > v$ must hold: if
nothing else, because~$t \ge \text{GST}$,
all honest (active) epoch-$e$ stakeholders will enter
view~$(e,v)$ by timeslot $t+\Dk$, their view-$(e,v)$ timers will expire
by timeslot~$t+\Dk+5\Dksq$, and they will enter
view~$(e,v+1)$ by timeslot~$t+2\Dk+5\Dksq \le t+7\Dksq$ (upon reception
of a VC for view~$(e,v+1)$).
In particular,~$(e,v^*)$ and all subsequent views are post-GST.
Arguing as in the proof of the first statement of Claim~2 shows that,
by timeslot~$t^*+12N \Dksq$, $B_{g,e+1}$ will be defined for all honest
players active at that time and thus a post-GST epoch must have begun.
\end{proof}

Claims~1--3 imply that the PoS-HotStuff protocol satisfies liveness
with liveness parameter $\ell := (24N+8) \Dksq$. For suppose a
transaction~$\mathtt{tr}$ is first received by some honest player at
some timeslot~$t$. Let~$t^*$ denote $\max\{\text{GST},t\}$ and assume
that $\mathtt{tr}$ remains valid for all honest players through
timeslot~$t^*+\ell$.  Because honest players immediately disseminate
all received transactions,
by timeslot~$t^* + \Dk \le t^*+\Dksq$, $\mathtt{tr}$ is received by all
active honest players.  By Claim~3, there is a post-GST epoch~$e$ that
begins at some timeslot in~$[t^* + \Dksq,t^* + (12N+8)\Dksq]$.  By
Claim~2, by timeslot~$t^* + (24N+8)\Dksq$, epoch~$e$ must have
completed, with~$B_{g,e+1}$ defined for all honest players active at
that time; if the transaction $\mathtt{tr}$ is not included already in
some block confirmed in an earlier epoch, it is included in one of the
confirmed epoch-$e$ blocks (if nothing else, the block proposed by the
leader of the first good view in epoch~$e$).

\vspace{0.2cm} 
\noindent \textbf{Proving optimistic responsiveness}.  We now adapt
the proof of liveness above to show that the PoS-HotStuff protocol is
optimistically responsive, even with a grace period~$\Delta^*$ of~0.
Assume that all players are honest and let $\delta$ be the realized
maximum message delay after GST.
\begin{itemize}

\item Claim~1 and its proof hold verbatim with all occurrences of the
maximum message delay~$\Delta$ replaced by the maximum realized
message delay~$\delta$.  

\item Claim~2 and its proof simplify somewhat because we now assume that all
players are honest. Specifically, every view of the post-GST epoch~$e$
is good, and every such view results, within~$5\dk$ timeslots, in
either the reception by all (honest) stakeholders of a stage~3 QC for
a block proposed in that view or some player proceeding to a
subsequent epoch.  Arguing as in the original proof of Claim~2, if~$t$
is the first timeslot in which some player reaches epoch~$e$, then for
every timeslot~$\ge t + (5N+1)\dk$, $B_{g,e+1}$ is defined for all
players active at that time. Trivially, all of the confirmed epoch-$e$
blocks are proposed by honest players.

\item The modified statement of Claim~3 is that if~$t \ge \text{GST}$, then
there is a post-GST epoch that is first reached by a (honest) player
at some timeslot in~$[t,t+(5N+7)\dk]$. As before, let~$(e,v)$ denote
the largest view that any player is in at timeslot~$t$.  Every
epoch-$e$ stakeholder is then in a view at least as large as~$(e,v)$
by timeslot~$t+\dk$.  Let~$(e^*,v^*)$ denote the largest view that any
player is in at timeslot~$t^*:=t+6\dk$. As before, we can assume
that~$e^*=e$.  As before, if~$v^* > v$, arguing as in (the modified
version of) Claim~2 shows that, by timeslot $t^*+(5N+1)\dk$,
$B_{g,e+1}$ will be defined for all active players at that time and
thus a post-GST epoch must have begun.  We can therefore complete the
proof by arguing that all stakeholders must have entered a view larger
than~$(e,v)$ by timeslot~$t+6\dk$. Indeed, if all such players remain
in view~$(e,v)$ for all timeslots in $[t+\dk,t+5\dk])$, then, arguing
as in Claim~1, all epoch-$e$ stakeholders must see a stage 3 QC for a
block proposed in~$(e,v)$ by the beginning of timeslot~$t+5\dk$ and
disseminate a $\mathtt{view}~(e,v+1)$ message at that time. Thus, all
such players receive a VC for view~$(e,v+1)$ by the beginning of
timeslot~$t+6\dk$, and must have therefore entered a view larger
than~$(e,v)$ by that timeslot.

\end{itemize}
Analogous to the liveness argument above, the modified versions of
Claims~1--3 imply that a transaction~$\mathtt{tr}$ that is first
received by a player at some timeslot~$t$ will be confirmed for all
active players by timeslot~$\max\{\text{GST},t\} + (10N+9)\dk$.  That is, the
PoS-HotStuff protocol is optimistically responsive with liveness
parameter~$\ell=(10N+9)\dk = O(\delta)$ and grace period
parameter~$\Delta^*=0$, as claimed.

\subsection{Efficiency considerations} \label{effic} 

\textbf{Stake in escrow}. We next take up the suggestion from
Section~\ref{overv} to consider a variant of the PoS-HotStuff protocol
in which, for some positive integer~$N'$ that divides the total stake amount~$N$, participating players must
hold some positive integer multiple of $1/N'$ of the total stake
``in escrow.'' The motivation for this change is to 
allow the set of participating players to change every~$N'$ blocks
rather than every~$N$ blocks, thereby decreasing the time required to
confirm a transaction (while retaining determinism).

In more detail, we adopt the extended definition of a $\rho$-bounded
environment from footnote~\ref{foot:escrow} of Section~\ref{envrho},
according to which the adversary will always own at most a
$\rho$-fraction of the stake in escrow provided it always owns at most
a $\rho$-fraction of the stake.
The only significant modification to the protocol that is required
concerns the block heights at which the set of participating players
changes. Consider first the total stake held in escrow according to
the initial stake distribution $S_0$, and suppose that this is a
$(k_0/N')$-fraction of the total stake. Then, rather than the first
change in the player set occurring at height $N$ (as in the basic
PoS-HotStuff protocol), this change occurs instead at height $k_0$.
Consider next the total stake held in escrow according to the
confirmed transactions in blocks of height $\leq k_0$, and suppose
that this is a $(k_1/N')$-fraction of the total stake. Then, the next
change in the set of participating players occurs at height $k_0+k_1$,
and so on. In this way, the protocol recursively defines a sequence of
changes to the set of participating players that occur at heights that
are defined dynamically as the execution progresses (to reflect the
changing amount of stake in escrow).
The definitions of the function $\mathtt{leader}$, the confirmation
rule~$\mathcal{C}$, and a stage $s$ QC must be adapted to accommodate
stake-in-escrow and variable epoch lengths, but the changes
required are straightforward.  Fundamentally, the proof of
Theorem~\ref{posPoS2} remains essentially unchanged.
 
\vspace{0.2cm} 
\noindent \textbf{Message length and message complexity}. In
Section~\ref{tfs}, taking advantage of our lack of concern with issues
of efficiency, we allowed each block to specify its parent by
including its parent within the block data. The standard approach to
avoiding this inefficiency and reducing message length is
to include instead a hash value for the parent block. Similarly, a
vote $V$ can specify the voted-for block~$B^*(V)$ by including a hash
of it, rather than the block itself.

The dissemination model of communication featured in this paper is
inspired by the gossip protocols commonly used for communication in
permissionless blockchain protocols. If the players assigned non-zero
stake by $\mathtt{T}_e$ can somehow communicate directly
(point-to-point), the message complexity of the PoS-HotStuff protocol
can be reduced using the same techniques as in the original HotStuff
protocol \cite{yin2019hotstuff}, with all-to-leader and leader-to-all
communication replacing all-to-all communication within a view.

\section{Proof of Theorem \ref{niceneg}} \label{twelve}


\subsection{The intuition}  \label{intu}

Towards a contradiction, suppose that such a protocol
$(\Sigma,\mathcal{O},\mathcal{C}, \mathcal{S})$ exists for some
$\rho>0$ and $\epsilon < 1/5$.  
Following the payment circle example
(Section~\ref{ss:circle}), there is a player set (with one identifier
each), an initial stake distribution~$S_0$, and sets of
transactions~$\mathtt{T},\mathtt{T}'$ with $\mathtt{T}'\subset
\mathtt{T}$ such that:
\begin{enumerate} 

\item[$(\dagger_1)$] $S_0$ allocates at most a $\rho$-fraction of
  the stake to the set of Byzantine players $P_B$ and the remainder of
  the stake to honest players in a set $P_S$.

\item[$(\dagger_2)$] $\mathtt{T}$ and $\mathtt{T}'$ are both valid
  relative to $S_0$.

\item[$(\dagger_3)$] $\mathtt{S}(S_0,\mathtt{T})$ allocates at most
  a $\rho$-fraction of the stake to players in $P_B$.

\item[$(\dagger_4)$] $\mathtt{S}(S_0,\mathtt{T}')$ allocates all stake
  to players in $P_B$.

\end{enumerate} 
%
Next, consider a protocol instance (in the sense of
Section~\ref{instance}) in which:
\begin{itemize} 

\item The transactions in $\mathtt{T}$ are sent by the environment to
  a set of honest players $P_T$ at timeslot 1, with
  $P_T\cap P_S=\emptyset$.

\item Transactions in $\mathtt{T}'$ are sent by the environment to all
  players in $P_B$ at timeslot 1.

\item No dissemination by any player in $P_T$ is received by any other
  player until GST.

\end{itemize} 
%
The environment may also subsequently send, to a set of honest
players, up to two additional transactions, $\mathtt{tr}_1$ and $\mathtt{tr}_2$.
Each of these transactions transfers stake between
two players in $P_B$, and they are chosen so that they conflict
relative to $\mathtt{T}'$: $\mathtt{T}' \cup \{ \mathtt{tr}_i \}$ is valid
relative to~$S_0$ (for~$i=1,2$), but $\mathtt{T}' \cup \{ \mathtt{tr}_1
, \mathtt{tr}_2 \}$ is not valid relative to~$S_0$.
Such an environment is, by definition, maximally $\rho$-bounded.

\vspace{0.2cm} 
\noindent \textbf{Two phases}. We then consider protocol instances for
which the timeslots prior to GST are divided into two phases. The
result of the first phase is that
all on-chain resources are transferred to players in $P_B$, at least
from the perspective of some honest players.
In the second phase, the adversary carries out a version of the
``split-brain attack'' described in the classic proof by Dwork, Lynch,
and Stockmeyer~\cite{DLS88}.
(The proof in~\cite{DLS88} concerns the permissioned model without
resources, so 
we 
rework it for the general quasi-permissionless setting that is
studied here.)

We next elaborate on the intuition behind the first phase.
%
%
%
%
We consider a protocol instance $\mathtt{I}$ with player set
$\mathcal{P}=P_B \cup P_S \cup P_T \cup P_E$.  The players of~$P_E$
are honest and are given sufficient external resources to ensure that the
adversary is externally $\rho$-bounded, but are otherwise cut off from
the rest (until GST).  Messages disseminated by players in $P_B$ and
$P_S$ during the first phase are received by all players at the next
timeslot.  During the first phase,
because~$\mathcal{S}$ may contain any finite number of on-chain
resources
and we are working in the general quasi-permissionless
setting (Section~\ref{ss:qp}), the instance must accommodate
requirements on the set of active players.
On the other hand:
\begin{enumerate} 

\item[(a)] By assumption, the set $\mathcal{S}$ of on-chain resources is
  reactive, and;

\item[(b)] During the first phase (with no messages received from
  players in~$P_T$ or~$P_E$), instance $\mathtt{I}$ is, for players
  in $P_B$ and $P_S$, indistinguishable from an instance
  $\mathtt{I}'$ in which the player set is $P_B\cup P_S$, all players
  are honest, all external resources are allocated to players in
  $P_B$, and $\mathtt{T}'$ is the set of all transactions sent by the
  environment.

\end{enumerate} 
  
\noindent 
Assuming that the first phase is longer than the liveness
parameter~$\ell$ of the protocol $(\Sigma,\mathcal{O},\mathcal{C},
\mathcal{S})$, properties~(a) and~(b) imply that, at the end of
this phase, from the perspective of the players in $P_B \cup P_S$, 
the players in~$P_B$ own the totality of all on-chain
resources. Players of~$P_S$ are then eligible to go
inactive---which, prior to GST, is indistinguishable from long message
delays---and the table has been set for the second phase.
%

\begin{table}
\begin{tabular}{|c|l|}\hline
$P_B$ & potentially Byzantine players (with up to a
    $\rho$-fraction of stake and external resources)\\ \hline
$P_S$ & honest players with non-zero initial stake (but no external
    resources)\\ \hline
$P_T$ & honest
    players with initial knowledge of
    $\mathtt{T}$ (but no stake or external resources) \\ \hline
$P_E$ & honest players with a
    non-zero amount of some external resource (but no stake)\\ \hline
\end{tabular}
\caption{Subsets of players used in the proof of Theorem~\ref{niceneg}.}\label{table:niceneg}
\end{table}

\subsection{The formal details} 

Suppose that the protocol
$(\Sigma,\mathcal{O},\mathcal{C}, \mathcal{S})$ is live with respect
to $\ell$ and the set $\mathcal{S}$ of designated on-chain resources
is reactive with respect to $\ell^*$ (with a security parameter
$\epsilon < 1/5$).  Let $O_1,\ldots,O_k \in \mathcal{O}$ denote the
permitters used by the protocol. (The set $\mathcal{O}$ may also
contain oracles that are not permitters, although these will not play
any role in the proof.)  Let
$\{ \mathtt{S}^{\ast}_1,\dots, \mathtt{S}^{\ast}_j \}$ denote the
on-chain resources in $\mathcal{S}$.  (One of the
$\mathtt{S}^{\ast}_i$'s may be the stake allocation function
$\mathtt{S}$.)
Fix the delay bound~$\Delta=2$, and some
duration~$d$ that is larger than $2\ell + \ell^*+1$.

We consider four pairwise-disjoint and non-empty finite sets of
players $P_B$, $P_S$, $P_T$, and $P_E$ (Table~\ref{table:niceneg}).
Let $S_0$, $\mathtt{T}$, and $\mathtt{T}'$ be defined as in Section
\ref{intu}, so that the conditions~$(\dagger_1)$--$(\dagger_4)$
specified in that section are satisfied with respect to $P_B$ and $P_S$.
Let $\mathtt{tr}_1$ and $\mathtt{tr}_2$ be transactions that transfer
stake between players in $P_B$ such that
$\mathtt{T}' \cup \{ \mathtt{tr}_i \}$ and
$\mathtt{T} \cup \{ \mathtt{tr}_i \}$ are both valid relative to $S_0$
(for each $i \in \{1,2\}$) and such that
$\mathtt{T}' \cup \{ \mathtt{tr}_1,\mathtt{tr}_2 \}$ is not valid
relative to $S_0$.

\vspace{0.2cm} 
\noindent \textbf{Four protocol instances}.
The proof uses several protocol instances to deduce a contradiction.
The role of the first instance is to force
honest players to confirm all the transactions in $\mathtt{T}'$.

\vspace{0.2cm} 
\noindent \textbf{Instance} $\mathtt{I}_0$: 
\begin{itemize} 

\item $\mathcal{P}=P_B\cup P_S$.

\item All players are honest and active at all timeslots.

\item All external resources are assigned to~$P_B$. Specifically, for
  every $p \in P_B$ and for some (arbitrarily chosen) positive and
  even integer~$x_p$,
  $\mathcal{R}^{O_i}(p,t) = x_p$ for all $i \in \{1,\ldots,k\}$ and $t
  \le d$.
($\mathcal{R}^{O_i}(p,t)=0$ for all $p \in P_S$, $i \in
\{1,\ldots,k\}$, and~$t \le d$.)
  
\item The environment sends the transactions in $\mathtt{T}'$ to
  players in $P_B$ at timeslot 1.

\item GST$=0$. 

\item All disseminations by all players are received by all other
  players at the next timeslot.

\end{itemize} 
Because all players are honest, this instance has an externally
0-bounded adversary and a maximally 0-bounded environment.

\vspace{0.2cm} 

The role of the second and third instances is to force subsets of
honest players to confirm the potentially conflicting transactions
$\mathtt{tr}_1$ and $\mathtt{tr}_2$.  Partition~$P_S$ into two
non-empty sets, $Y_S$ and $Z_S$.  Set $t^*=\ell+\ell^*+1$.

\vspace{0.2cm} 
\noindent \textbf{Instance} $\mathtt{I}_1$: 
\begin{itemize}

\item $\mathcal{P}=P_B \cup P_S$.

\item All players are honest and active at all timeslots.

\item At timeslots~$t < t^*$, all resource allocations are defined
  identically to instance $\mathtt{I}_0$. At timeslot~$t^*$, every
  player's resource balances are cut in half.
That is, $\mathcal{R}^{O_i}(p,t) = x_p/2$ for all $p \in P_B$, $i \in
\{1,\ldots,k\}$, and $t \in [t^*,d]$.
($\mathcal{R}^{O_i}(p,t)=0$ for all $p \in P_S$, $i \in
\{1,\ldots,k\}$, and~$t \le d$.)

\item The environment sends the transactions in $\mathtt{T}'$ to
  players in $P_B$ at timeslot 1.

\item The environment sends $\mathtt{tr}_1$ to players in $Y_S$ at
  $t^*$.

\item GST $>t^*+\ell$. 

\item Until timeslot $t^*$, every dissemination by a
  player is received by all players at the next timeslot.

\item At all timeslots $\geq t^{\ast}$ until GST, disseminations by
  players in $P_B \cup Y_S$ are received by all other players at
  the next odd timeslot, while disseminations by players in $Z_S$
  are not received by any other player until GST.

\end{itemize} 

\vspace{0.1cm} 

\noindent The next instance is the ``mirror image'' of $\mathtt{I}_1$, in which
the roles of $Y_S$ and $Z_S$ (and of $\mathtt{tr}_1$ and
$\mathtt{tr}_2$) are reversed.

\vspace{0.2cm} 
\noindent \textbf{Instance} $\mathtt{I}_2$: 
\begin{itemize}

\item $\mathcal{P}=P_B \cup P_S$.

\item All players are honest and active at all timeslots.

\item All resource allocations are defined identically to instance
  $\mathtt{I}_1$.

\item The environment sends the transactions in $\mathtt{T}'$ to
  players in $P_B$ at timeslot 1.

\item The environment sends $\mathtt{tr}_2$ to players in $Z_S$ at
  $t^*$.

\item GST $>t^*+\ell$. 

\item Until timeslot $t^*$, every dissemination by a player is
  received by all players at the next timeslot.

\item At all timeslots $\geq t^{\ast}$ until GST, disseminations by
  players in $P_B \cup Z_S$ are received by all other players at
  the next even timeslot, while disseminations by players in $Y_S$
  are not received by any other player until GST.

\end{itemize} 
As in instance~$\mathtt{I}_0$, in instances~$\mathtt{I}_1$
and~$\mathtt{I}_2$, the adversary is externally 0-bounded
and the environment is maximally 0-bounded.



\vspace{0.2cm} 

The fourth instance plays the role of instance~$\mathtt{I}$ in
Section~\ref{intu}---this is the instance in which the proof will
force a consistency violation.

\vspace{0.2cm} 
\noindent \textbf{Instance} $\mathtt{I}_3$: 
\begin{itemize}

\item $\mathcal{P}=P_B\cup P_S \cup P_T \cup P_E$.

\item Players in $P_B$ are Byzantine,
  all other players are honest.

\item All players are active at all timeslots.

\item For players in~$P_B$, all resource allocations are defined
  identically to instance $\mathtt{I}_0$.  For every $p \in P_E$ and
  for some nonnegative~$x_p$, $\mathcal{R}^{O_i}(p,t) = x_p$ for all
  $i \in \{1,\ldots,k\}$ and $t \le d$. Further, these $x_p$'s are
  chosen so that
  $\sum_{p \in P_B} x_p = \rho \sum_{p \in P_B \cup P_E} x_p$.
($\mathcal{R}^{O_i}(p,t)=0$ for all $p \in P_S \cup P_T$, $i \in
\{1,\ldots,k\}$, and~$t \le d$.)

\item
  At timeslot~1, the environment sends the transactions in $\mathtt{T}'$ to
  players in $P_B$ and the transactions in $\mathtt{T}$ to
  players in $P_T$.

\item At $t^*$, the environment sends $\mathtt{tr}_1$ to players in
  $Y_S$ and $\mathtt{tr}_2$ to players in $Z_S$.

\item GST $>t^*+\ell$. 

\item Disseminations by players in $P_T \cup P_E$ are not received
  until GST.

\item Until timeslot $t^*$, disseminations by players in
  $P_B \cup P_S$ are received by all players at the next timeslot.

\item At all timeslots $\geq t^*$ until GST, disseminations by players
  in $Y_S$ are received by players in $P_B \cup Y_S$ at the next odd
  timeslot and are not received by players in $Z_S$ until GST.
  Disseminations by players in $Z_S$ are received by players in
  $P_B \cup Z_S$ at the next even timeslot and are not received by
  players in $Y_S$ until GST.

\item Byzantine players behave honestly until~$t^*$.  At all timeslots
  $\geq t^*$ until GST, each player $p\in P_B$ simulates two
  fictitious players $p_1$ and $p_2$ as if
  $\mathcal{R}^{O_i}(p_1,t)= \mathcal{R}^{O_i}(p_2,t)= x_p/2$ for each
  $i \in \{1,\ldots,k\}$. Further:
\begin{itemize}
    
\item Fictitious player $p_1$ simulates what player~$p$ would have
  done in instance $\mathtt{I}_1$ were it honest and with all external
  resource
  balances equal to~$x_p/2$. That is, in the present instance, $p_1$
  acts honestly (but with reduced resource balances~$x_p/2$) except that
  it ignores disseminations received at even timeslots.

\item Disseminations by $p_1$ are received by players in
  $P_B \cup Y_S$ at the next odd timeslot, and are not received by
  players in $Z_S$ until GST.

\item Fictitious player $p_2$ simulates what player~$p$ would have
  done in instance $\mathtt{I}_2$ were it honest and with all external
  resource
  balances equal to~$x_p/2$. That is, in the present instance, $p_2$
  acts honestly (but with reduced resource balances~$x_p/2$) except that
  it ignores disseminations received at odd timeslots.

\item Disseminations by $p_2$ are received by players in
  $P_B \cup Z_S$ at the next even timeslot, and are not received by
  players in $Y_S$ until GST.

\end{itemize}
\end{itemize} 
In $\mathtt{I}_3$, the resource balances of players in~$P_E$ are
defined so that the adversary is externally $\rho$-bounded. By
construction of the transactions in $\mathtt{T}$, $\mathtt{tr}_1$, and
$\mathtt{tr}_2$, and because the environment sends each transaction
initially to one or more honest players, the environment is
maximally $\rho$-bounded.

\vspace{0.2cm} 
\noindent \textbf{Producing the contradiction}. Assume for a moment
that the protocol $(\Sigma,\mathcal{O},\mathcal{C}, \mathcal{S})$ is
deterministic (i.e., $\epsilon=0$). The following steps then
contradict the assumed consistency of the protocol:
\begin{enumerate}

\item Because the protocol
  $(\Sigma,\mathcal{O},\mathcal{C}, \mathcal{S})$ is live with respect
  to~$\ell$, in $\mathtt{I}_0$, all the transactions in $\mathtt{T}'$
  are confirmed for all players in $P_B \cup P_S$ by timeslot $\ell+1$.

\item By construction of the transactions in $\mathtt{T}'$, and because
  $\mathtt{tr}_1$ and $\mathtt{tr}_2$ transfer stake between
  players of~$P_B$, in $\mathtt{I}_0$, from the perspective of
  each player in $P_B \cup P_S$, for all timeslots $\ge \ell+1$ until
  GST, all stake is allocated to players of~$P_B$.

\item Because the set $\mathcal{S}$ of on-chain resources is reactive
  with respect to $\ell^*$, in $\mathtt{I}_0$,
  from the
  perspective of each player in $P_B \cup P_S$, for all timeslots
  $\ge t^* (:=\ell+\ell^*+1)$ until GST, for
  every~$i \in \{1,\ldots,j\}$, players outside of~$P_B$ own none of
  $\mathtt{S}^*_i$.

\item Because the instances $\mathtt{I}_0$, $\mathtt{I}_1$, and
  $\mathtt{I}_2$ are indistinguishable for players of $P_B \cup P_S$
  until~$t^*$, in all three instances, from the perspective of each
  player in $P_B \cup P_S$, for all timeslots
  $\ge t^* (:=\ell+\ell^*+1)$ until GST, for
  every~$i \in \{1,\ldots,j\}$, players outside of~$P_B$ own none of
  $\mathtt{S}^*_i$.

\item \label{foo}
  Consider a protocol instance~$\mathtt{I}_4$ that is identical
  to~$\mathtt{I}_1$, except that: (i) GST$=0$; and (ii) the players
  of~$Z_S$ become inactive at~$t^*$. (This player inactivity is
  consistent with the general quasi-permissionless setting as, from
  the perspective of any player in the instance, the players of~$Z_S$
  own none of any of the resources listed in $\mathcal{S}$.)
Because the protocol is live with respect
to~$\ell$, and because the set $\mathtt{T}' \cup \{ \mathtt{tr}_1 \}$
is valid relative to~$S_0$, in $\mathtt{I}_4$,
the transaction $\mathtt{tr}_1$ is confirmed for all players in $P_B
\cup Y_S$ by timeslot $t^*+\ell$.

\item \label{bar}
  Because the instances $\mathtt{I}_1$ and $\mathtt{I}_4$ are
  indistinguishable for players of~$P_B \cup Y_S$ until $t^*+\ell$, in
  $\mathtt{I}_1$, 
the transaction $\mathtt{tr}_1$ is confirmed for all players in $P_B
\cup Y_S$ by timeslot $t^*+\ell$.

\item Reasoning as in steps~(\ref{foo}) and~(\ref{bar}), using an
  analogous instance~$\mathtt{I}_5$, in $\mathtt{I}_2$,
the transaction $\mathtt{tr}_2$ is confirmed for all players in $P_B
\cup Z_S$ by timeslot $t^*+\ell$.

\item The strategies of the Byzantine players in $\mathtt{I}_3$ ensure
  that the instances $\mathtt{I}_1$ and $\mathtt{I}_3$ are
  indistinguishable for players of~$Y_S$ until $t^*+\ell$.  Thus, in
  $\mathtt{I}_3$, the transaction $\mathtt{tr}_1$ is confirmed for all
  players in $Y_S$ by timeslot $t^*+\ell$.

\item The strategies of the Byzantine players in $\mathtt{I}_3$ ensure
  that the instances $\mathtt{I}_2$ and $\mathtt{I}_3$ are
  indistinguishable for players of~$Z_S$ until $t^*+\ell$.  Thus, in
  $\mathtt{I}_3$, the transaction $\mathtt{tr}_2$ is confirmed for all
  players in $Z_S$ by timeslot $t^*+\ell$.


\item Because the transactions $\mathtt{tr}_1$ and $\mathtt{tr}_2$
  conflict (relative to~$S_0$ and $\mathtt{T}'$), the confirmation of
  the transactions of $\mathtt{T}' \cup \{ \mathtt{tr}_1 \}$
  and of $\mathtt{T}' \cup \{ \mathtt{tr}_2 \}$ for the 
 (honest) players of~$Y_S$ and~$Z_S$, respectively, constitute a
 consistency violation at timeslot~$t^*+\ell$ in~$\mathtt{I}_3$.
 
\end{enumerate}
Suppose now that the protocol
$(\Sigma,\mathcal{O},\mathcal{C}, \mathcal{S})$ is probabilistic, with
security parameter~$\epsilon$.
If a specific execution of this protocol for the instance
$\mathtt{I}_3$ does not suffer a consistency
violation by timeslot~$t^*+\ell$, we attribute it to 
the first of the following explanations that applies:
\begin{itemize}

\item Not all the transactions of $\mathtt{T}'$ were confirmed for all
  players in $P_B \cup P_S$ by timeslot $\ell+1$.

\item Some player outside of~$P_B$ still owned some of an on-chain
  resource in $\mathcal{S}$ at a timeslot $\ge t^*$ before GST.

\item The transaction $\mathtt{tr}_1$ was not confirmed for all
  players in $Y_S$ by timeslot $t^*+\ell$.

\item The transaction $\mathtt{tr}_2$ was not confirmed for all
  players in $Z_S$ by timeslot $t^*+\ell$.

\end{itemize}
Because the instances~$\mathtt{I}_0$ and~$\mathtt{I}_3$ are
indistinguishable to players of~$P_B \cup P_S$ until $t^*$, every
failure of the first type can be associated with a failure of liveness
in $\mathtt{I}_0$, which in turn occurs with probability at
most~$\epsilon$. Similarly, failures of the second, third, and fourth
types can be associated with failures of reactivity in $\mathtt{I}_0$,
liveness in $\mathtt{I_4}$, and liveness in $\mathtt{I}_5$,
respectively (each of which occurs with probability at
most~$\epsilon$). We conclude that the protocol suffers a consistency
violation in instance $\mathtt{I}_3$ with probability at least
$1-4\epsilon$.
Because $\epsilon < 1/5$, this produces a contradiction and completes
the proof.

\vspace{0.2cm} 
\noindent \textbf{Discussion}. The proof above relies on the
possibility that players can, under certain conditions, be
inactive (as in instances $\mathtt{I}_4$ and $\mathtt{I}_5$),
as well as the fact that the player set is unknown ($P_B \cup P_S$
versus $P_B \cup P_S \cup P_T \cup P_E$).  The proof does not rely on
the possibility of sybils (each player can use a unique identifier).

\section*{Acknowledgments} We would like to thank the following people
for a number of useful conversations on the paper: Ittai Abraham,
Jayamine Alupotha, Rainer B\"{o}hme, Christian Cachin,  
Justin Drake, Mahimna Kelkar, Duc V.\ Le, Giuliano Losa, 
Oded Naor, Kartik Nayak,  Joachim Neu, Ling Ren, Elaine Shi, Ertem
Nusret Tas, David Tse, and Luca Zanolini. 

The research of the second author at Columbia University is supported
in part by NSF awards CCF-2006737 and CNS-2212745.
The second author is also Head of Research at a16z Crypto, a venture
capital firm with investments in blockchain protocols.

\appendix

\section{Proof-of-space permitters - an example}\label{app:chia}

To further demonstrate the versatility of our framework, we briefly
describe how to represent the salient features of the proof-of-space
Chia protocol using the permitter formalism from Section~\ref{fp}.  We
consider the simplified form of proof-of-space described in the Chia
green paper~\cite{cohen2019chia}.


In Chia, proofs of space depend on a \emph{space parameter}
$\ell\in \mathbb{N}_{\geq 1}$.  Players who wish to generate proofs of
space must store ``lookup tables'', each of size
$N\approx \ell \cdot 2^{\ell+1}$ bits.  Each lookup table is used to
produce proofs-of-space corresponding to a specific identifier.  The
proofs of space are of a specific form that makes them efficiently
verifiable.  We interpret the resource balance of a player as the
number of lookup tables (each with its own identifier) that the player
is in a position to maintain; this number is proportional to the
amount of space devoted by the player to the protocol.

Upon receiving a \emph{challenge} $y$ from the protocol, a player can
employ an efficient process that uses its lookup table to generate
proofs of space that are specific to that challenge and the
corresponding identifier.  For roughly a $1-1/e \approx 0.632$
fraction of the challenges there will be at least one proof, and the
expected number of proofs (per lookup table) is 1. So, for a given
identifier, each challenge~$y$ produces an output which is a (possibly
empty) set of proofs of space, with the distribution on this output
being given by $\mathbb{P}$, say.

To model this process within our framework, we use a multi-use
permitter, meaning that each player $p$ can send any finite number of
requests to the permitter at each timeslot.  For convenience, we
suppose each player $p$ has a single identifier (i.e.,
$|\mathtt{id}(p)|=1$) and that $\mathcal{R}^O(p,t)\in \{ 0,1 \}$ for
all~$p$ and $t$.  (Intuitively, a `miner' generates one identifier for
each lookup table that it maintains, with each identifier effectively
acting independently.)  Each request made at a timeslot $t$ must be of
the form $(1,(id,y))$, where $y$ is a challenge and where $(id,y)$ is
an entry of signed type (and where, as always, we require
$1\leq \mathcal{R}^O(p,t)$).  (Because the permitter is multi-use, a
player can make multiple such requests in the same timeslot,
presumably with different challenges~$y$.)  At the beginning of the
protocol execution, the permitter takes each possible request
$(1,(id,y))$ and samples a (possibly empty) set of proofs-of-space
$X(id,y)$ according to the distribution $\mathbb{P}$. If $p$ ever
sends the request $(1,(id,y))$, the permitter immediately sends the
response $(y,id,X(id,y))$, where this response is of type $O$.

\section{Modeling protocols -- some examples} \label{apC} 

As noted in Section~\ref{ss:hierarchy}, typical PoW protocols (such as
Bitcoin~\cite{nakamoto2008bitcoin}) will generally operate in the
fully permissionless setting, typical PoS longest-chain protocols
(such as Snow White~\cite{daian2019snow} and
Ouroboros~\cite{kiayias2017ouroboros}) operate in the dynamically
available setting, while typical PoS protocols that employ classical
PBFT-style methods (such as Algorand~\cite{chen2016algorand} or PoS
implementations of Tendermint~\cite{buchman2018latest} or
HotStuff~\cite{yin2019hotstuff}) will generally operate in the
quasi-permissionless setting. This appendix further illustrates our
general framework by considering some protocols with non-standard
features, and determines the settings in which these would be expected
to operate.

\vspace{0.2cm} 
\noindent \textbf{Ethereum.} The present version of Ethereum employs a
consensus protocol (essentially Gasper~\cite{buterin2020combining})
that combines longest-chain (LMD-Ghost) voting with a PBFT-style
protocol (Casper \cite{buterin2017casper}).  This protocols aims to be
live and consistent simultaneously in two settings: the dynamically
available and synchronous setting, and the quasi-permissionless and
partially synchronous setting.\footnote{See Neu et
  al.~\cite{neu2021ebb} for a discussion of the extent
  to which the protocol satisfies these properties and for a description of an alternative protocol which satisfies these properties.}
Theorem~\ref{psm} establishes that moving to the quasi-permissionless
setting is necessary when communication is only partially
synchronous.


\vspace{0.2cm} 
\noindent \textbf{Prism and OHIE}. While Prism \cite{bagaria2019prism}
and OHIE \cite{yu2020ohie} may be ``non-standard'' PoW protocols in
the sense that they build multiple chains and (in the case of Prism)
use multiple block types, they can be modeled within our framework as
operating in the fully permissionless setting.

\vspace{0.2cm} 
\noindent \textbf{Byzcoin, Hybrid, and Solida.} Byzcoin
\cite{kokoris2016enhancing}, Hybrid \cite{pass2016hybrid}, and Solida
\cite{abraham2016solida} use PoW to select a rolling sequence of
committees. As noted in Section \ref{blockrewards}, the basic idea is
that, once a committee is selected, its members should (be active and)
carry out a classical PBFT-style protocol to implement the next
consensus decision. This consensus decision includes the next sequence
of transactions to be committed to the blockchain, and also determines
which players have provided sufficient PoW in the meantime for
inclusion in the next committee.  Such protocols can be modeled in the
quasi-permissionless setting by specifying a protocol-defined resource
that allocates non-zero balance to the members of the current
committee (i.e., the committee members according to the most recently
confirmed block).

\vspace{0.2cm} 
\noindent \textbf{Bitcoin-NG.} Bitcoin-NG \cite{eyal2016bitcoin}
increases the transaction processing rate of Bitcoin by considering
blocks of two types. ``Key blocks'' require PoW and are essentially
the same as blocks in the Bitcoin protocol. Once a miner produces a
key block, they become the ``leader'' and can then---until some other
miner takes the reins by finding a new key block---produce
``microblocks'' that include further transactions and which do not
require PoW.  While the protocol is (probabilistically) live and
consistent in the fully permissionless and synchronous setting (as
inherited from the Bitcoin protocol), to take advantage of microblock
production (and the corresponding increase in the rate of transaction
processing), the protocol should be modeled in the
quasi-permissionless setting by specifying a protocol-defined resource
that allocates non-zero balance to the miner who produced the most
recent block in the longest chain.


\section{Proof of Theorem~\ref{sep}}\label{app:lg}

Recall the definition of the Byzantine Agreement problem from
Section~\ref{2.5}, with each player~$p$ receiving a protocol
input~$x_p$ that is either~0 or~1. 

\vspace{0.2cm}
\noindent \textbf{The setup}.
The protocol $(\Sigma, \mathcal{O}, \mathcal{C})$ that achieves the
guarantees promised by Theorem~\ref{sep} is based on the protocol of
Losa and Gafni~\cite{losa2023consensus} and is a variation on that of
Dolev and Strong~\cite{dolev1983authenticated}.  The protocol takes
advantage of an initial stake distribution~$S_0$ but does not use any
oracles. For the purposes of solving BA, we can assume that there is
no environment and no transactions; the confirmation
rule~$\mathcal{C}$ is irrelevant.
The protocol is designed for the synchronous setting; let~$\Delta$
denote the (determined) maximum message delay.

The initial stake distribution assigns a positive and integer number
of coins to each of a finite set of identifiers; let~$N$ denote the
total number of coins. We assume that these coins are numbered in an
arbitrary way. By assumption, at most a~$\rho < 1/2$ fraction of these
coins are assigned to identifiers belonging to faulty players.

Players are instructed to disseminate messages only at timeslots that
are multiples of~$\Delta$; for a nonnegative integer~$r$, we refer to
timeslot $r\Delta$ as {\em round~$r$}.\footnote{In this protocol,
  players are assumed to know the number of the current timeslot, even
  after waking up from a period of inactivity; see also
  footnote~\ref{foot:know_timeslot} in Section~\ref{ss:st}.}  By the assumptions of the
synchronous setting, a message disseminated at round~$r$ will be
received by round~$r+1$ by all players active at that round.  There
will be a total of $N^*:=\lceil N/2 \rceil + 1$ rounds.  


\vspace{0.2cm}
\noindent \textbf{Activity messages and attestations}.  If a
player~$p$ is active at a round~$r \le N^*$, that player is instructed
to disseminate a {\em round-$r$ activity message}
$(c,r,\mathtt{active})$ for each of the
coins~$c$ assigned to its identifiers.

If a player is active at round~$1$, that player is instructed to
disseminate, for each of the coins~$c$ assigned to one of its
identifiers, a {\em length-$1$ attestation}~$(c,1,x_p)$
to its protocol input~$x_p$.

In general, if: (i) a player~$p$ is active at a round~$r+1 \le N^*$;
(ii) $p$ has received a length-$r$ attestation~$m$
to~$b \in \{0,1\}$, involving~$r$ different coins;
and (iii) $p$ has not already disseminated any
attestations to~$b$, then~$p$ is instructed to disseminate, for each
coin~$c$ assigned to one of its identifiers, a length-$(r+1)$
attestation~$(c,r+1,m)$ to~$b$.

\vspace{0.2cm}
\noindent \textbf{Output values}. Each player~$p$ is instructed to
output in the first timeslot $\ge (N^*+1)\Delta$ at which it is
active; let~$M$ denote all the messages it has received by this time.
Let~$A_r(M)$ denote the coins associated with round-$r$ activity
messages in~$M$. Let $V_{r,b}(M)$ denote the coins~$c$ associated with
length-$\le r$ attestations to~$b$ (with~$c$ the coin mostly recently
added to the attestation).
If there is a
round~$r \in \{1,2,\ldots, N^*\}$ such that a strict majority of the
coins in~$A_r(M)$ belong also to~$V_{r,b}(M)$, then~$p$ is {\em
  convinced of~$b$ by~$M$}. If~$p$ is convinced of exactly one
value~$b$ by~$M$, it outputs~$b$; otherwise, it outputs a default
value (0, say).

\vspace{0.2cm}
\noindent \textbf{Limited faulty behavior}.  Because we assume that
faulty players deviate from honest behavior only by delaying
message dissemination (or crashing):
\begin{itemize}

\item [(1)] 
A faulty player cannot disseminate messages that reference coins that
are not assigned to its identifiers.

\item [(2)] 
A faulty player cannot send round-$r$ activity messages for rounds~$r$ at
which it is not active.

\item [(3)] A faulty player cannot produce a length-1 attestation to 
  a value other than its protocol input.

\end{itemize}

\vspace{0.2cm}
\noindent \textbf{Correctness}. Termination is obvious from the
protocol instructions for outputting values.  For validity, suppose
that every player~$p$ is given the same protocol
input~$b \in \{0,1\}$.  Consider some honest player~$p$ and the
messages~$M$ it has received when it chooses an output.  By~(3),
$V_{r,1-b}(M)$ is empty for all~$r \in \{1,2,\ldots,N^*\}$. Thus,~$M$
cannot convince~$p$ of~$1-b$.  By~(1)--(2) and the $\rho$-bounded
assumption (with $\rho < 1/2$), a strict majority of the coins
in~$A_1(M)$ are assigned to honest players' identifiers.  The protocol
instructions for honest players active at round~1 ensure that all of
these coins appear also in~$V_{1,b}(M)$. Thus, $p$ is convinced of~$b$
(and only~$b$) by~$M$, and will output~$b$.

For agreement, it suffices to show that every honest player is
convinced of the same set of values by its timeslot of output---the
first timeslot $\ge (N^*+1)\Delta$ at which it is active.
So, suppose some honest player is convinced of~$b \in \{0,1\}$ on
account of receiving a set~$M$ of messages in which, for some $r \in
\{1,2,\ldots,N^*\}$, more than half the coins in~$A_r(M)$ are also
in~$V_{r,b}(M)$. 
By~(1)--(2) and the $\rho$-bounded
assumption (with $\rho < 1/2$), a strict majority of the coins
in~$A_r(M)$ are assigned to honest players' identifiers. Because a
strict majority of the coins in~$A_r(M)$ are also in~$V_{r,b}(M)$,
there is a coin in~$A_r(M) \cap V_{r,b}(M)$ assigned to an identifier
that belongs to an honest player~$p$. Membership of this coin in
$V_{r,b}(M)$ implies that~$p$ must have
disseminated a length-$\ell$ attestation to~$b$ for some~$\ell \le
r$. Because~$p$ is honest, this dissemination must have occurred at
round~$\ell$. We now consider two cases.

First, suppose that $\ell < N^*$.  Honest players active at
round~$\ell+1$, having then received $p$'s attestation, would
disseminate their own (length-$(\ell+1)$) attestations to~$b$, unless
they had done so already in some previous round.  Thus, for any
set~$M'$ of messages that might be received by a player by a timeslot
$\ge (N^*+1)\Delta \ge (\ell+2)\Delta$,
$V_{\ell+1,b}(M')$ must contain every coin assigned to an identifier of
an honest player that is active at round~$\ell+1$. By~(1)-(2) and the
$\rho$-bounded assumption, these coins constitute a strict majority of
the coins in~$A_{\ell+1}(M')$. Thus, every honest player will be
convinced of~$b$ by the messages~$M'$ it has received by its timeslot
$\ge (N^*+1)\Delta$ of output.

Finally, suppose that~$\ell=N^*$, meaning that the honest player~$p$
disseminated a (length-$N^*$) attestation to~$b$ at round~$N^*$. By
the $\rho$-bounded assumption (with $\rho < 1/2$) and our choice
of~$N^*$, at least one of the other~$N^*-1$ coins associated with this
attestation must be assigned to an identifier of an honest
player~$p'$.  Player~$p'$ must have disseminated, for some
$\ell' < N^*$, a length-$\ell'$ attestation to~$b$ at round $\ell'$;
thus, this case reduces to the preceding one (with~$p'$ and $\ell'$
playing the roles of~$p$ and $\ell$, respectively), and every honest
player will be convinced of~$b$ by its timeslot of output.

\end{document}